\documentclass[journal,letterpaper]{IEEEtran}
\usepackage[utf8]{inputenc}
\usepackage{amsmath,amssymb,amsthm,thmtools,amsfonts,mathtools,stmaryrd}
\usepackage{algpseudocode,algorithm}
\usepackage{accents,bm,bbm}
\usepackage{enumitem}
\usepackage{tikz,pgfplots}
\usepackage[labelformat=simple]{subcaption}
\DeclareCaptionLabelSeparator{periodspace}{.\quad}
\renewcommand{\mathbf}[1]{{\bm{#1}}}
\makeatletter
\def\munderbar#1{\underline{\sbox\tw@{$#1$}\dp\tw@\z@\box\tw@}}
\makeatother

\algnewcommand\algorithmicforeach{\textbf{for each}}
\algdef{S}[FOR]{ForEach}[1]{\algorithmicforeach\ #1\ \algorithmicdo}

\newcommand{\eg}{\emph{e.g.}}
\newcommand{\etal}{\emph{et al.}}
\newcommand{\ie}{\emph{i.e.}}

\declaretheorem[name=Theorem, sibling=algorithm]{theorem}

\declaretheorem[name=Proposition, sibling=algorithm]{proposition}

\declaretheorem[name=Definition, sibling=algorithm, style=definition, 
    qed=\qedsymbol]{definition}
\declaretheorem[name=Example, sibling=algorithm, style=definition, 
    qed=\qedsymbol]{example}

\DeclareMathSymbol{\Alpha}{\mathalpha}{operators}{"41}
\DeclareMathSymbol{\Beta}{\mathalpha}{operators}{"42}
\DeclareMathSymbol{\Epsilon}{\mathalpha}{operators}{"45}
\DeclareMathSymbol{\Zeta}{\mathalpha}{operators}{"5A}
\DeclareMathSymbol{\Eta}{\mathalpha}{operators}{"48}
\DeclareMathSymbol{\Iota}{\mathalpha}{operators}{"49}
\DeclareMathSymbol{\Kappa}{\mathalpha}{operators}{"4B}
\DeclareMathSymbol{\Mu}{\mathalpha}{operators}{"4D}
\DeclareMathSymbol{\Nu}{\mathalpha}{operators}{"4E}
\DeclareMathSymbol{\Omicron}{\mathalpha}{operators}{"4F}
\DeclareMathSymbol{\Rho}{\mathalpha}{operators}{"50}
\DeclareMathSymbol{\Tau}{\mathalpha}{operators}{"54}
\DeclareMathSymbol{\Chi}{\mathalpha}{operators}{"58}
\DeclareMathSymbol{\omicron}{\mathord}{letters}{"6F}

\newcommand{\Sum}{\sum\limits}
\newcommand{\Prod}{\prod\limits}

\DeclareMathOperator{\tr}{tr}

\newcommand{\defeq}{\triangleq}

\newcommand{\Herm}{\dag}
\newcommand{\bra}[1]{\left\lvert#1\right\rangle}
\newcommand{\ket}[1]{\left\langle#1\right\rvert}
\newcommand{\braket}[1]{\bra{#1}\!\ket{#1}}
\newcommand{\conj}[1]{\overline{#1}}
\newcommand{\abs}[1]{\left\lvert#1\right\rvert}
    \newcommand{\bigabs}[1]{\bigl\lvert#1\bigr\rvert}

\newcommand{\norm}[1]{\left\lVert#1\right\rVert}
\newcommand{\size}[1]{\left\lvert#1\right\rvert}
\newcommand{\D}{\mathrm{d}} 
\newcommand{\lvec}[1]{\accentset{\leftharpoonup}{#1}} 
\newcommand{\rvec}[1]{\accentset{\rightharpoonup}{#1}} 
\newcommand{\grad}{\nabla}
\newcommand{\id}{\mathcal{I}}
\newcommand{\tensor}{\otimes}
\newcommand{\set}[1]{\mathcal{#1}}
\newcommand{\operator}[1]{\mathcal{#1}}
\newcommand{\system}[1]{\mathsf{#1}} \newcommand{\sys}{\system}
\newcommand{\rv}[1]{\mathsf{#1}}
\newcommand{\entropy}{\boldsymbol{\mathrm{H}}}
\newcommand{\qEntropy}{\entropy}
\newcommand{\mutualInfo}{\boldsymbol{\mathrm{I}}}
\newcommand{\qmutualInfo}{\mutualInfo}
\newcommand{\infoRate}{\mathrm{I}}
\newcommand{\entropicRate}{\mathrm{H}}
\newcommand{\capacity}{\mathrm{C}}
\newcommand{\DensOp}{\mathfrak{D}}

\DeclarePairedDelimiterX{\infdivx}[2]{(}{)}{#1\;\delimsize\|\;#2}
\newcommand{\infdiv}[2]{\mathrm{D_{KL}}\infdivx*{#1}{#2}}
\newcommand{\IRUB}{\bar{\infoRate}}
\newcommand{\IRLB}{\munderbar{\infoRate}}

\newcommand{\hilbert}{\mathcal{H}}

\newcommand{\What}{\hat{W}}
\newcommand{\QW}{(Q\mkern-1.5mu W)}
\newcommand{\QWaux}{(Q\mkern-1.5mu\hat{W})}

\newcommand{\vect}[1]{\mathbf{#1}}

\newcommand{\vs}{\vect{s}}

\newcommand{\cs}{\check{s}}

\newcommand{\vx}{\vect{x}}
\newcommand{\cx}{\check{x}}
\newcommand{\cvx}{\check{\vect{x}}}

\newcommand{\vy}{\vect{y}}
\newcommand{\cy}{\check{y}}
\newcommand{\cvy}{\check{\vect{y}}}

\newcommand{\muY}{\mu^{\rv{Y}}}
\newcommand{\muXY}{\mu^{\rv{XY}}}

\newcommand{\bmuY}{\bar \mu^{\rv{Y}}}
\newcommand{\bmuXY}{\bar \mu^{\rv{XY}}}

\newcommand{\sigmaY}{\sigma^{\rv{Y}}}
\newcommand{\sigmaXY}{\sigma^{\rv{XY}}}

\newcommand{\bsigmaY}{\bar \sigma^{\rv{Y}}}
\newcommand{\bsigmaXY}{\bar \sigma^{\rv{XY}}}

\newcommand{\lambdaY}{\lambda^{\rv{Y}}}
\newcommand{\lambdaXY}{\lambda^{\rv{XY}}}

\newcommand{\pgood}{p_{\mathrm{g}}}
\newcommand{\pbad}{p_{\mathrm{b}}}

\newcounter{mytempeqcounter}
\newcommand{\bigformulatop}[2]{%
    \begin{figure*}[!t]
        \normalsize
        \setcounter{mytempeqcounter}{\value{equation}}
        \setcounter{equation}{#1}
        #2
        \setcounter{equation}{\value{mytempeqcounter}}
        \hrulefill
        \vspace*{4pt}
    \end{figure*}
}

\pgfplotsset{compat = 1.7}
\usetikzlibrary{positioning,shapes,arrows}
\usetikzlibrary{decorations.markings}
\pgfdeclarelayer{bg}
\pgfdeclarelayer{fg}
\pgfsetlayers{bg,main,fg}
\tikzset{-*/.style={shorten >=#1, decoration={
    markings, mark={at position 1 with {\draw[fill] circle [radius=#1];}}},
    postaction=decorate},
    -*/.default=1.3pt}
\tikzset{*-/.style={shorten <=#1, decoration={
    markings, mark={at position 0 with {\draw[fill] circle [radius=#1];}}},
    postaction=decorate},
    *-/.default=1.3pt}
\tikzset{*-*/.style={shorten <=#1,shorten >=#1,decoration={
    markings, mark={at position 0 with {\draw[fill] circle [radius=#1];}},
    mark={at position 1 with {\draw[fill] circle [radius=#1];}},},
    postaction=decorate},
    *-*/.default=1.3pt}
\newcommand{\scalefactorA}{0.85}
\newcommand{\scalefactorB}{0.75}

\begin{document}
\title{Bounding and Estimating the Classical Information Rate of
       Quantum Channels with Memory}
\author{Michael~X.~Cao\thanks{M.~X.~Cao is with the Department of Information
        Engineering, The Chinese University of Hong Kong, Shatin, N.T.,
        Hong Kong. E-mail: m.x.cao@ieee.org.},~\IEEEmembership{Student
        Member,~IEEE,} and
        Pascal~O.~Vontobel\thanks{P.~O.~Vontobel is with the Department of
        Information Engineering and the Institute of Theoretical Computer
        Science and Communications, The Chinese University of Hong Kong.
        Email: pascal.vontobel@ieee.org.},~\IEEEmembership{Senior Member,~IEEE}
        \thanks{The work described in this paper was partially supported by
                grants from the Research Grants Council of the Hong Kong
                Special Administrative Region, China
                (Project Nos. CUHK 14209317 and CUHK 14207518).}
        \thanks{This paper was presented in part at the IEEE International
                Symposium on Information Theory (ISIT), Aachen, Germany,
                July 2017~\cite{cao2017estimating}, and the IEEE International
                Symposium on Information Theory (ISIT), Paris, France,
                July 2019~\cite{cao2019optimizing}.}
        \thanks{Submitted. Date of current version: \today.}
}
\maketitle
\begin{abstract}
We consider the scenario of classical communication over a finite-dimensional
quantum channel with memory using a separable-state input ensemble
and local output measurements.
We propose algorithms for estimating the information rate of such communication
setups, along with algorithms for bounding the information rate based on
so-called auxiliary channels.
Some of the algorithms are extensions of their counterparts for (classical)
finite-state-machine channels. Notably, we discuss suitable graphical models
for doing the relevant computations. Moreover, the auxiliary channels are
learned in a data-driven approach; \ie, only input/output sequences of the
true channel are needed, but not the channel model of the true channel.
\end{abstract}
\section{Introduction} \label{sec:1:Introduction}
\IEEEPARstart{W}{e} consider the transmission rate of classical information
over a finite-dimensional quantum channel with memory~\cite{bowen2004quantum,
kretschmann2005quantum, caruso2014quantum}. 
Recall that in the memoryless case, given an input system $\system{A}$ and
an output system $\system{B}$, described by some Hilbert spaces
$\hilbert_\system{A}$ and $\hilbert_\sys{B}$,
respectively, a memoryless quantum channel can be modeled as a 
\emph{completely positive trace-preserving} (CPTP) map from the set of density
operators on $\hilbert_\system{A}$ to the set of density operators on 
$\hilbert_\system{B}$~\cite{nielsen2011quantum, wilde2017quantum}; such a
quantum channel is said to be finite-dimensional if both  
$\hilbert_\system{A}$ and $\hilbert_\system{B}$ are of finite dimension.
A \emph{quantum channel with memory} is a quantum channel equipped
with a memory system $\system{S}$; namely it is a CPTP map from
the set of density operators on 
    $\hilbert_\system{A}\tensor\hilbert_\system{S}$ 
to the set of  density operators on
    $\hilbert_\system{B}\tensor\hilbert_{\system{S}'}$,
where $\hilbert_\system{S}$ $(\equiv\hilbert_{\system{S}'})$ is the Hilbert
space describing $\system{S}$, and $\tensor$ stands for the tensor product.
The system $\system{S}$ can be understood either as a state of the channel
(as illustrated in Fig.~\ref{fig:interpret:1}), or as a part of the environment
that does not decay between consecutive channel uses (as illustrated in
Fig.~\ref{fig:interpret:2}). Interesting examples of quantum channels with 
memory include spin chains~\cite{bose2003quantum} and fiber optic 
links~\cite{ball2004exploiting}.
\par 
Classical communication over such channels is accomplished by encoding 
classical data into some density operators before the transmission and
applying measurements at the outputs of the 
channel~\cite{nielsen2011quantum, wilde2017quantum}.
In the most generic case, a joint input ensemble and a joint output measurement
across multiple channels can be used for encoding and decoding, respectively.
The scenario involving a $k$-channel joint ensemble and a $k$-channel joint
measurement is depicted in Fig.~\ref{fig:generic_memory}, where
\begin{itemize}
\item the encoding process $\mathcal{E}$ is described by some ensemble 
    $\{P_\rv{X}(x),\rho_{\system{A}_1^k}^{(x)}\}_{x\in\set{X}}$
    on the joint input system $(\system{A}_1,\ldots,\system{A}_k)$, with
    $\set{X}$ being the input alphabet, $P_\rv{X}(x)$ being the input 
    distribution, and $\rho_{\system{A}_1^k}^{(x)}$ being the density
    operator on the input systems $\system{A}_1^k$ corresponding to
    the classical input $x$;
\item the decoding process $\mathcal{D}$ is described by some positive-operator 
    valued measure (POVM) 
    $\{\Lambda_{\system{B}_1^k}^{(y)}\}_{y\in\set{Y}}$ 
    on the joint output system $(\system{B}_1,\ldots,\system{B}_k)$, with 
    $\set{Y}$ being the output alphabet;
\item the classical input and output are represented by some
    random variables $\rv{X}$ and $\rv{Y}$, respectively.
\end{itemize}
For comparison, Fig.~\ref{fig:generic_memoryless} shows the corresponding
memoryless setup.
The above arrangement results in a (classical) channel
from $\rv{X}$ to $\rv{Y}$, whose rate of transmission is given by
\begin{align}\label{eq:capacity:jointk}
\infoRate(\mathcal{E},\operator{N}^{\boxtimes k},\mathcal{D})=
\limsup_{n\rightarrow\infty}\frac{1}{n}
    \mutualInfo(\rv{X}_1^{n};\rv{Y}_1^{n}),
\end{align}
where we use the above transmission scheme $n$ times consecutively (as depicted in Fig.~\ref{fig:generic_memory_multiple}), and where 
\[
\begin{aligned}
\operator{N}^{\boxtimes k} \defeq
    &\left(\operator{N}_{\system{A}_k\system{S}_{k-1}\rightarrow
                        \system{B}_k\system{S}_k}
          \otimes \id_{\system{B}_1^{k-1}\rightarrow\system{B}_1^{k-1}}\right)
    \circ\\
    &\left(\id_{\system{A}_k\rightarrow\system{A}_k} \otimes
          \operator{N}_{\system{A}_{k-1}\system{S}_{k-2}\rightarrow
                        \system{B}_{k-1}\system{S}_{k-1}}
          \otimes \id_{\system{B}_1^{k-2}\rightarrow\system{B}_1^{k-2}}\right)
    \circ\\ 
    &\cdots \circ
    \Big(\id_{\system{A}_2^k\rightarrow\system{A}_2^k}
          \otimes \operator{N}_{\system{A}_1\system{S}_0\rightarrow
                        \system{B}_1\system{S}_1}\Big).
\end{aligned}
\]
Here, $\mutualInfo$ stands for the mutual information. As a fundamental result,
this quantity can be simplified to $\mutualInfo(\rv{X};\rv{Y})$ for 
the memoryless case~\cite{shannon2001mathematical,cover2012elements}.
Optimizing $\infoRate(\mathcal{E},\operator{N}^{\boxtimes k},\mathcal{D})$ over
$\mathcal{E}$ and $\mathcal{D}$ (with $k\rightarrow\infty$) yields the
classical capacity of the quantum channel $\operator{N}$, namely
\begin{equation}
\capacity(\operator{N}) = \limsup_{k}\frac{1}{k}\sup_{\mathcal{E},\mathcal{D}}
    \infoRate(\mathcal{E},\operator{N}^{\boxtimes k},\mathcal{D}).
\end{equation}
\par 
\begin{figure}[t]
\centering
\begin{subfigure}[b]{0.48\columnwidth}
\centering 
    \begin{tikzpicture}[
    smallcircle/.style={draw, minimum size = 10pt, circle},
    scale=0.9,every node/.style={transform shape}]
\node[draw, minimum width = 20pt, minimum height = 15pt] (C) {$\operator{N}$};
\node[smallcircle, below = 12pt of C] (M) {};
\node[left = 2.5pt of M.center, anchor=center, smallcircle] {};
\node[right = 2.5pt of M.center, anchor=center, smallcircle] {};
\draw (C.west|-M.south) -- (C.east|-M.south);
\draw (C.west|-M.south) arc (270:90:13pt);
\draw (C.east|-M.south) arc (-90:90:13pt);
\draw ([yshift=3pt]C.west) -- ([yshift=3pt,xshift=-15pt]C.west)
    node[pos=1,left](A) {$\system{A}$};
\draw ([yshift=3pt]C.east) -- ([yshift=3pt,xshift=15pt]C.east)
    node[pos=1,right](B) {$\system{B}$};
\node[yshift=13pt] at (A|-M.south) {$\system{S}$};
\node[yshift=13pt] at (B|-M.south) {$\system{S}^\prime$};
\node[below=0pt of M,font=\scriptsize] {quantum memory};
\end{tikzpicture}
\caption{Memory as the state of the channel}
\label{fig:interpret:1}
\end{subfigure}
~
\begin{subfigure}[b]{0.48\columnwidth}
\centering 
    \begin{tikzpicture}[
    scale=0.9,every node/.style={transform shape}]
\node[draw, minimum width = 20pt, minimum height = 50pt] (U) {$\operator{U}$};
\draw ([yshift=18pt]U.west) -- ([yshift=18pt,xshift=-15pt]U.west)
    node[left, pos=1] (A) {$\system{A}$};
\draw (U.west) -- ([xshift=-15pt]U.west)
    node[left, pos=1] (E) {$\bra{0}$};
\draw[latex-] ([yshift=-18pt]U.west) -- ([yshift=-18pt,xshift=-15pt]U.west)
    --([yshift=-32pt,xshift=-15pt]U.west)
    --([yshift=-32pt,xshift=28pt]U.east)
    --([yshift=-18pt,xshift=28pt]U.east)
    --([yshift=-18pt]U.east);
\draw (U.east) -- ([xshift=9pt]U.east) 
    node [pos=1, right, draw, minimum size=12pt] (M) {};

\node[minimum size = 1.5pt, fill = black, inner sep = 0pt, outer sep = 0pt,
    below = 3pt of M.center, circle, anchor = center] (m) {};
\draw [decoration={markings,mark=at position 1 with
    {\arrow[scale=0.5,>=latex]{>}}},postaction={decorate}, draw = none]
    (m.center) -- ([yshift=7pt, xshift=3pt]m.center);
\draw (m.center) -- ([yshift=6.3pt, xshift=2.7pt]m.center);
\draw ([xshift=4pt]m.center) arc (0:180:4pt);

\draw ([yshift=18pt]U.east) -- ([yshift=18pt,xshift=28pt]U.east)
    node[right, pos=1] (B) {$\system{B}$};
\draw[line width = 1.3pt] (M.east) -- ([xshift=13.5pt]M.east) 
    node[pos=1, minimum size = 5pt, inner sep = 0pt, outer sep = 0pt] (Eend) {};
\draw[line width = 1.3pt] (Eend.south west) -- (Eend.north east);
\draw[line width = 1.3pt] (Eend.north west) -- (Eend.south east);

\node[yshift=-18pt] at (A|-U.west) (S) {$\system{S}$};
\node[yshift=-18pt] at (B|-U.east) {$\system{S}^\prime$};

\path[draw=none] (E) edge[draw=none] node[xshift=-12pt](env) {$\Bigg\{$} (S);
\node[left=0pt of env, anchor = center, rotate=90, font = \scriptsize]
    {environment};
\node[right=82pt of env] {$\Bigg\}$};
\end{tikzpicture}
\caption{Memory as undecayed partial environment}
\label{fig:interpret:2}
\end{subfigure}
\caption{Interpretations of quantum channels with memory.}
\end{figure}
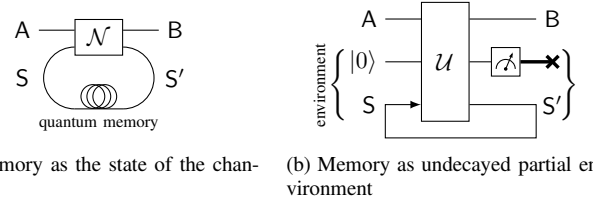
\begin{figure}[t]
\centering
\begin{subfigure}[b]{0.48\columnwidth}
\begin{center} 
    \begin{tikzpicture}[
	bignode/.style={minimum height = 70pt, minimum width = 15pt, draw, font=\small},
	smallnode/.style={minimum size = 5pt, draw, font=\small},
	scale=0.9,every node/.style={transform shape}]
\node[bignode] (cq) {$\mathcal{E}$};
\node[bignode, right = 50pt of cq] (qc) {$\mathcal{D}$};
\path ([yshift=30pt]cq.east) edge[draw=none] node[smallnode, midway] (T1)
	{$\operator{N}$} ([yshift=30pt]qc.west);
\path ([yshift=30pt]cq.east) edge node[above=-2pt, midway, font=\footnotesize]
    {$\system{A}_1$} (T1.west);
\path (T1.east) edge node[above=-2pt, midway, font=\footnotesize]
    {$\system{B}_1$} ([yshift=30pt]qc.west);

\path ([yshift=7pt]cq.east) edge[draw=none] node[smallnode, midway] (T2)
	{$\operator{N}$} ([yshift=7pt]qc.west);
\path ([yshift=7pt]cq.east) edge node[above=-2pt, midway, font=\footnotesize]
    {$\system{A}_2$} (T2.west);
\path (T2.east) edge node[above=-2pt, midway, font=\footnotesize]
    {$\system{B}_2$} ([yshift=7pt]qc.west);

\path ([yshift=-7pt]cq.east) edge[draw=none] node {$\vdots$} ([yshift=-7pt]qc.west);

\path ([yshift=-30pt]cq.east) edge[draw=none] node[smallnode, midway] (Tn)
	{$\operator{N}$} ([yshift=-30pt]qc.west);
\path ([yshift=-30pt]cq.east) edge node[above=-2pt, midway, font=\footnotesize]
    {$\system{A}_k$} (Tn.west);
\path (Tn.east) edge node[above=-2pt, midway, font=\footnotesize]
    {$\system{B}_k$} ([yshift=-30pt]qc.west);

\path (cq.west) edge[line width=1.5pt] node[pos=0.75,above,font=\small]
	{$\rv{X}$} ([xshift=-20pt]cq.west);
\path (qc.east) edge[line width=1.5pt] node[pos=0.75,above,font=\small]
	{$\rv{Y}$} ([xshift=20pt]qc.east);
\draw[draw=none] ([yshift=10pt]T1.north) -- (T1.north) 
    node[pos=0,above=0pt,font=\footnotesize,color=white] {$\system{S}_0$};
\draw[draw=none] ([yshift=-10pt]Tn.south) -- (Tn.south) 
    node[pos=0,below=0pt,font=\footnotesize,color=white] {$\system{S}_n$};
\end{tikzpicture} \end{center} \vspace{-10pt}
\caption{Memoryless channel ($\operator{N}^{\tensor k}$)}
\label{fig:generic_memoryless}
\end{subfigure}
~
\begin{subfigure}[b]{0.48\columnwidth}
\begin{center} 
    \begin{tikzpicture}[
	bignode/.style={minimum height = 70pt, minimum width = 15pt, draw},
	smallnode/.style={minimum size = 5pt, draw,font=\small},
	scale=0.9,every node/.style={transform shape}]
\node[bignode] (cq) {$\mathcal{E}$};
\node[bignode, right = 50pt of cq] (qc) {$\mathcal{D}$};
\path ([yshift=30pt]cq.east) edge[draw=none] node[smallnode, midway] (T1)
	{$\operator{N}$} ([yshift=30pt]qc.west);
\path ([yshift=30pt]cq.east) edge node[above=-2pt, midway, font=\footnotesize]
    {$\system{A}_1$} (T1.west);
\path (T1.east) edge node[above=-2pt, midway, font=\footnotesize]
    {$\system{B}_1$} ([yshift=30pt]qc.west);

\path ([yshift=7pt]cq.east) edge[draw=none] node[smallnode, midway] (T2)
	{$\operator{N}$} ([yshift=7pt]qc.west);
\path ([yshift=7pt]cq.east) edge node[above=-2pt, midway, font=\footnotesize]
    {$\system{A}_2$} (T2.west);
\path (T2.east) edge node[above=-2pt, midway, font=\footnotesize]
    {$\system{B}_2$} ([yshift=7pt]qc.west);

\path ([yshift=-9pt]cq.east) edge[draw=none] node {$\vdots$} ([yshift=-9pt]qc.west);

\path ([yshift=-30pt]cq.east) edge[draw=none] node[smallnode, midway] (Tn)
	{$\operator{N}$} ([yshift=-30pt]qc.west);
\path ([yshift=-30pt]cq.east) edge node[above=-2pt, midway, font=\footnotesize]
    {$\system{A}_k$} (Tn.west);
\path (Tn.east) edge node[above=-2pt, midway, font=\footnotesize]
    {$\system{B}_k$} ([yshift=-30pt]qc.west);

\draw ([yshift=10pt]T1.north) -- (T1.north) 
    node[pos=0,above=0pt,font=\footnotesize] {$\system{S}_0$};
\draw (T1) -- (T2) node[midway,right=-2pt,font=\footnotesize] {$\system{S}_1$};
\draw (T2.south) -- ([yshift=-5pt]T2.south) 
	node[pos=1,right=-2pt,font=\footnotesize] {$\system{S}_2$};
\draw ([yshift=5pt]Tn.north) -- (Tn.north)
	node[pos=0,right=-2pt,font=\footnotesize] {$\system{S}_{k-1}$};
\draw ([yshift=-10pt]Tn.south) -- (Tn.south) 
    node[pos=0,below=0pt,font=\footnotesize] {$\system{S}_k$};
\path (cq.west) edge[line width=1.5pt] node[pos=0.75,above,font=\small]
	{$\rv{X}$} ([xshift=-20pt]cq.west);
\path (qc.east) edge[line width=1.5pt] node[pos=0.75,above,font=\small]
	{$\rv{Y}$} ([xshift=20pt]qc.east);


\end{tikzpicture} \end{center} \vspace{-10pt}
\caption{Channel with memory ($\operator{N}^{\boxtimes k}$)}
\label{fig:generic_memory}
\end{subfigure}
~
\begin{subfigure}[b]{0.96\columnwidth}
\begin{center} 
    \begin{tikzpicture}[smallnode/.style={minimum size = 5pt, draw, font=\small},
                    scale=0.9,every node/.style={transform shape}]
\node[smallnode] (cq1) {$\mathcal{E}$};
\node[smallnode, below =30pt of cq1.center, anchor = center] (cq2)
     {$\mathcal{E}$};
\node[font=\small,below =20pt of cq2.center, anchor = center] (cqdots)
     {$\vdots$};
\node[smallnode,below =30pt of cqdots.center, anchor = center] (cqn)
     {$\mathcal{E}$};

\node[smallnode, right = 110pt of cq1] (qc1) {$\mathcal{D}$};
\node[smallnode, right = 110pt of cq2] (qc2) {$\mathcal{D}$};
\node[font=\small, below = 20pt of qc2.center, anchor = center] (qcdots) {$\vdots$};
\node[smallnode, right = 110pt of cqn] (qcn) {$\mathcal{D}$};

\path (cq1) edge[draw=none] node[smallnode, midway] (T1)
    {$\operator{N}^{\boxtimes k}$} (qc1);
\path (cq1) edge node[above=-2pt, midway, font=\footnotesize]
    {$\system{A}_1^k$} (T1);
\path (T1) edge node[above=-2pt, midway,font=\footnotesize]
    {$\system{B}_1^k$} (qc1);

\path (cq2) edge[draw=none] node[smallnode, midway] (T2)
    {$\operator{N}^{\boxtimes k}$} (qc2);
\path (cq2) edge node[above=-2pt, midway,font=\footnotesize]
    {$\system{A}_{k+1}^{2k}$} (T2);
\path (T2) edge node[above=-2pt, midway,font=\footnotesize]
    {$\system{B}_{k+1}^{2k}$} (qc2);

\path (cqdots) edge[draw=none] node[font=\small,midway] {$\vdots$} (qcdots);
\path (cqn) edge[draw=none] node[smallnode, midway] (Tn)
    {$\operator{N}^{\boxtimes k}$} (qcn);
\path (cqn) edge node[above=-2pt, midway,font=\footnotesize]
    {$\system{A}_{(n\!-\!1)\cdot k+1}^{n\cdot k}$} (Tn);
\path (Tn) edge node[above=-2pt, midway,font=\footnotesize]
    {$\system{B}_{(n\!-\!1)\cdot k+1}^{n\cdot k}$} (qcn);

\draw[line width = 1.5pt] (cq1.west) -- ([xshift=-20pt]cq1.west) 
    node[above=-2pt, font=\footnotesize, pos = 0.8] {$\rv{X}_1$};
\draw[line width = 1.5pt] (cq2.west) -- ([xshift=-20pt]cq2.west)
    node[above=-2pt, font=\footnotesize, pos = 0.8] {$\rv{X}_2$};
\draw[line width = 1.5pt] (cqn.west) -- ([xshift=-20pt]cqn.west)
    node[above=-2pt, font=\footnotesize, pos = 0.8] {$\rv{X}_n$};

\draw[line width = 1.5pt] (qc1.east) -- ([xshift=20pt]qc1.east)
    node[above=-2pt, font=\footnotesize, pos = 0.7] {$\rv{Y}_1$};
\draw[line width = 1.5pt] (qc2.east) -- ([xshift=20pt]qc2.east)
    node[above=-2pt, font=\footnotesize, pos = 0.7] {$\rv{Y}_2$};
\draw[line width = 1.5pt] (qcn.east) -- ([xshift=20pt]qcn.east)
    node[above=-2pt, font=\footnotesize, pos = 0.7] {$\rv{Y}_n$};

\draw ([yshift=10pt]T1.north) -- (T1.north) 
    node[pos=0,above=0pt,font=\footnotesize] {$\system{S}_0$};
\draw (T1) -- (T2) node[midway,right=-2pt,font=\footnotesize] {$\system{S}_k$};
\draw (T2.south) -- ([yshift=-7pt]T2.south)
	node[pos=1,right=-2pt,font=\footnotesize] {$\system{S}_{2k}$};
\draw ([yshift=9pt]Tn.north) -- (Tn.north)
	node[pos=0,right=-2pt,font=\footnotesize] {$\system{S}_{(n\!-\!1)\cdot k}$};
\draw ([yshift=-10pt]Tn.south) -- (Tn.south) 
    node[pos=0,below=0pt,font=\footnotesize] {$\system{S}_{n\cdot k}$};
\end{tikzpicture} \end{center} \vspace{-10pt}
\caption{Generic classical communication corresponding to
         Eq.~\eqref{eq:capacity:jointk}}
\label{fig:generic_memory_multiple}
\end{subfigure}
\caption{Classical communications over quantum channels.}
\label{fig:ccq}
\end{figure}
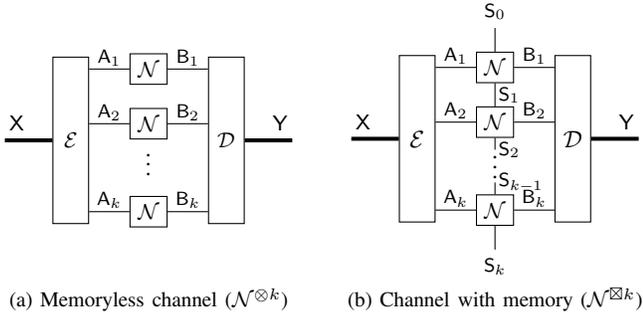
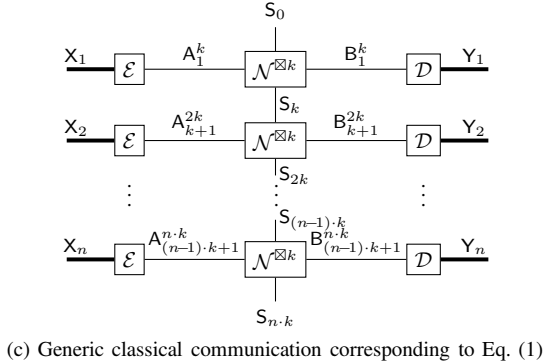
In this paper, we are interested in computing and bounding the information 
rate as in~\eqref{eq:capacity:jointk} for finite-dimensional quantum 
channels with memory using only separable input ensembles and local output 
measurements, \ie, the case $k=1$, which is depicted in Fig.~\ref{fig:ccq2}.
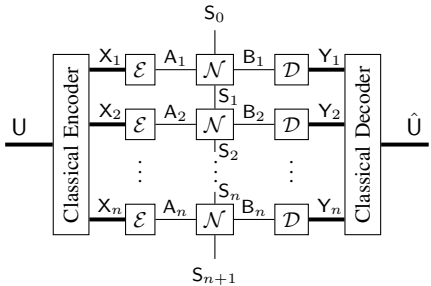
\begin{figure}[t]
\begin{center} 
    \begin{tikzpicture}[
	smallnode/.style={minimum size = 5pt, draw, font=\small},
	scale=0.9,every node/.style={transform shape}]
\node[smallnode] (cq1) {$\mathcal{E}$};
\node[smallnode,below =23pt of cq1.center, anchor = center] (cq2) {$\mathcal{E}$};
\node[font=\small,below =17pt of cq2.center, anchor = center] (cqdots) {$\vdots$};
\node[smallnode,below =23pt of cqdots.center, anchor = center] (cqn) {$\mathcal{E}$};

\node[smallnode, right = 50pt of cq1] (qc1) {$\mathcal{D}$};
\node[smallnode, right = 50pt of cq2] (qc2) {$\mathcal{D}$};
\node[font=\small, below = 17pt of qc2.center, anchor = center] (qcdots) {$\vdots$};
\node[smallnode, right = 50pt of cqn] (qcn) {$\mathcal{D}$};

\path (cq1) edge[draw=none] node[smallnode, midway] (T1)
	{$\operator{N}$} (qc1);
\path (cq1) edge node[above=-2pt, midway, font=\footnotesize]
    {$\system{A}_1$} (T1);
\path (T1) edge node[above=-2pt, midway, font=\footnotesize]
    {$\system{B}_1$} (qc1);

\path (cq2) edge[draw=none] node[smallnode, midway] (T2)
	{$\operator{N}$} (qc2);
\path (cq2) edge node[above=-2pt, midway, font=\footnotesize]
    {$\system{A}_2$} (T2);
\path (T2) edge node[above=-2pt, midway, font=\footnotesize]
    {$\system{B}_2$} (qc2);

\path (cqdots) edge[draw=none] node[font=\small,midway] {$\vdots$} (qcdots);

\path (cqn) edge[draw=none] node[smallnode, midway] (Tn)
	{$\operator{N}$} (qcn);
\path (cqn) edge node[above=-2pt, midway, font=\footnotesize]
    {$\system{A}_n$} (Tn);
\path (Tn) edge node[above=-2pt, midway, font=\footnotesize]
    {$\system{B}_n$} (qcn);

\node[left = 15pt of cq1.north west, draw, rectangle, minimum width = 15pt,
minimum height = 75.8159pt, anchor = north east] (M) {};
\node[rotate=90, font=\small] at (M) {Classical Encoder};
\draw[line width = 1.5pt] (cq1.west-|M.east) |- (cq1.west) 
    node[above=-2pt, font=\footnotesize, pos = 0.8] {$\rv{X}_1$};
\draw[line width = 1.5pt] (cq2.west-|M.east) |- (cq2.west)
    node[above=-2pt, font=\footnotesize, pos = 0.8] {$\rv{X}_2$};
\draw[line width = 1.5pt] (cqn.west-|M.east) |- (cqn.west)
    node[above=-2pt, font=\footnotesize, pos = 0.8] {$\rv{X}_n$};
\path (M.west) edge[line width = 1.5pt] node[above, pos=0.7] {$\rv{U}$}
	([xshift=-20pt]M.west);

\node[right = 15pt of qc1.north east, draw, rectangle, minimum width = 15pt,
minimum height = 75.8159pt, anchor = north west] (Mp) {};
\node[rotate=90, font=\small] at (Mp) {Classical Decoder};
\draw[line width = 1.5pt] (qc1.east-|Mp.west) |- (qc1.east)
    node[above=-2pt, font=\footnotesize, pos = 0.7] {$\rv{Y}_1$};
\draw[line width = 1.5pt] (qc2.east-|Mp.west) |- (qc2.east)
    node[above=-2pt, font=\footnotesize, pos = 0.7] {$\rv{Y}_2$};;
\draw[line width = 1.5pt] (qcn.east-|Mp.west) |- (qcn.east)
    node[above=-2pt, font=\footnotesize, pos = 0.7] {$\rv{Y}_n$};;
\path (Mp.east) edge[line width = 1.5pt] node[above, pos=0.7] {$\hat{\rv{U}}$}
	([xshift=20pt]Mp.east);

\draw ([yshift=10pt]T1.north) -- (T1.north) 
    node[pos=0,above=0pt,font=\footnotesize] {$\system{S}_0$};
\draw (T1) -- (T2) node[midway,right=-2pt,font=\footnotesize] {$\system{S}_1$};
\draw (T2.south) -- ([yshift=-5pt]T2.south)
	node[pos=1,right=-2pt,font=\footnotesize] {$\system{S}_2$};
\draw ([yshift=5pt]Tn.north) -- (Tn.north)
	node[pos=0,right=-2pt,font=\footnotesize] {$\system{S}_n$};
\draw ([yshift=-10pt]Tn.south) -- (Tn.south) 
    node[pos=0,below=0pt,font=\footnotesize] {$\system{S}_{n+1}$};
\end{tikzpicture} \end{center}
\caption{Classical communication over a quantum channel with memory
         using a separable ensemble and local measurements.}
\label{fig:ccq2}
\end{figure}
This restriction is equivalent to the scenario where no quantum computing
device is present at the sending or receiving end; or the scenario where
our manipulation of the channel is limited to a single-channel use.
The difficulty of the problem lies with the presence of the quantum memory.
In the simplest situation, the memory system exhibits classical properties
under certain ensembles and measurements. In this case, the resulting
classical communication setup is equivalent to a finite-state-machine
channel (FSMC)~\cite{gallager1968information}. Though the evaluation of the
information rate of an FSMC is nontrivial in general, efficient stochastic
methods for estimating and bounding this quantity have been
developed~\cite{arnold2006simulation, sadeghi2009optimization}.
\par 
Our work is highly inspired by~\cite{arnold2006simulation}, where the authors 
considered the information rate of FSMCs. In particular,
for an indecomposable FSMC~\cite{gallager1968information} with channel law $W$,
its information rate, which is independent from the initial channel state,
is given by
\begin{equation}\label{eq:def:fsmc:ir:1}
\infoRate_W(Q) = \lim_{n\rightarrow\infty}
    \frac{1}{n}\mutualInfo(\rv{X}_1^n;\rv{Y}_1^n),
\end{equation}
where $\rv{X}_1^n = \left(\rv{X}_1,\ldots,\rv{X}_n\right)$ is the channel input 
process characterized by some sequence of distributions $\{Q^{(n)}\}_{n}$,
and where $\rv{Y}_1^n = \left(\rv{Y}_1,\ldots,\rv{Y}_n\right)$ is the channel
output process. Although, except for very special cases, there are no 
single-letter or other simple expressions for information rates available, 
efficient stochastic techniques have been developed for estimating the 
information rate for \emph{stationary} and \emph{ergodic} input processes 
$\{Q^{(n)}\}_{n}$~\cite{arnold2006simulation, sharma2001entropy, 
pfister2001achievable}.
(For these techniques, under mild conditions, the numerical estimate of the
information rate converges with probability one to the true value when the 
length of the channel input sequence goes to infinity.) In this paper, we
extend such techniques to quantum channels with memory; in particular, 
we use similar (but extended) graphical models, namely \emph{factor graphs}
for quantum probabilities~\cite{loeliger2017factor} for estimating quantities
of interest. These graphical models are useful for visualizing the relevant
computations and for providing a clear comparison between the setup considered
in this paper and its classical counterparts in~\cite{arnold2006simulation}
and~\cite{sadeghi2009optimization}.\footnote{Clearly, the graphical models
that we use are very similar to tensor networks (see, for example, the
discussion in Appendix~A of~\cite{loeliger2017factor}).
A benefit of the graphical models that we use (including the corresponding
terminology), is that they are compatible with the graphical models that are
being used in classical information processing.}
\par 
Our work is also partially inspired by~\cite{sadeghi2009optimization}, 
where the authors proposed upper and lower bounds based on some so-called
auxiliary FSMCs, which are often lower-complexity approximations of the original
FSMC. They also provided efficient methods for optimizing these bounds. 
Such techniques have been proven useful for FSMCs with large state spaces,
when the above-mentioned information rate estimation techniques can be overly 
time-consuming. Interestingly enough, the lower bounds represent achievable 
rates under mismatched decoding, where the decoder bases its computations not
on the true FSMC but on the auxiliary FSMC~\cite{ganti2000mismatched}. (See the
paper~\cite{sadeghi2009optimization} for a more detailed discussion of
this topic and for further references.) In this paper, we also consider 
auxiliary channels and their induced bounds. However, the auxiliary channels
of our interest are chosen from a larger set of channels called
\emph{quantum-state channels}, which will be defined in Section~\ref{sec:3:QCM}.
We also propose a method for optimizing these bounds. In particular, our method
for optimizing the lower bound is ``data-driven'' in the sense that only 
the input/output sequences of the original channel are needed, but not the
mathematical model of the original channel.
\par 
One must note that even if we can efficiently compute or bound the information
rate, it is still a long way to go to compute the classical capacity of a quantum
channel with memory. On the 
one hand, maximizing $\infoRate(\mathcal{E},\operator{N},\mathcal{D})$ is a
difficult problem. 
(The analogous classical problems have been addressed in 
\cite{arimoto1972algorithm}, \cite{blahut1972computation}, and 
\cite{vontobel2008generalization}.)
On the other hand, due to the superadditivity 
property~\cite{hastings2009superadditivity} of quantum channels, 
which happens to be more common for quantum channels with memory 
\cite{macchiavello2004transition, karimipour2006entanglement, 
lupo2010transitional} (compared with memoryless quantum channels),
it is inevitable to consider joint ensembles on input systems
and joint measurements on output systems across multiple channel uses.
\par 
The rest of this paper is organized as follows. Section~\ref{sec:2:FSMC}
reviews the method of estimating the information rate of an FSMC.
Section~\ref{sec:3:QCM} models the classical communication scheme
over a quantum channel with memory, and defines the notion of 
quantum-state channels as an equivalent description. A graphical notation for 
representing such channels is also presented in this section.
Section~\ref{sec:4:IR} estimates the information rate of such channels. 
Section~\ref{sec:5:UBLB} considers the upper and lower bounds induced by 
auxiliary quantum-state channels, and presents methods for optimizing them.
Section~\ref{sec:6:example} contains numerical examples. 
Section~\ref{sec:7:conclusion} concludes the paper.
\subsection{Further references}
In the following, we assume that the reader is familiar with the basic
elements of quantum information theory (see \cite{nielsen2011quantum} or 
\cite{wilde2017quantum} for an introduction).  For a general
introduction to quantum channels with memory, we refer to the papers by
Kretschmann and Werner~\cite{kretschmann2005quantum} and by Caruso
\etal~\cite{caruso2014quantum}.
\par 
Moreover, some familiarity with graphical models (like factor
graphs)~\cite{kschischang2001factor, forney2001codes, loeliger2004introduction}
and with techniques for estimating the information rate of an FSMC as
presented in~\cite{arnold2006simulation, sadeghi2009optimization} will be
beneficial. Recall that graphical models are a popular approach for representing 
multivariate functions with \emph{non-trivial} factorizations and for doing 
computations like marginalization~\cite{kschischang2001factor, 
forney2001codes,loeliger2004introduction}. In particular, graphical models can 
be used to represent joint probability mass functions (pmfs) / probability 
density functions (pdfs). In the present paper we will heavily rely on the
paper~\cite{loeliger2017factor}, which 
discussed an approach for using normal factor graphs (NFGs) for representing
functions that typically appear when doing computations w.r.t. some quantum
systems. Alternatively, we could also have used the slightly more compact 
double-edge normal factor graphs (DE-NFGs)~\cite{cao2017double}.
Probabilities of interest are then obtained by suitably applying the
sum-product algorithm or closing-the-box operations.
\subsection{Notations}
\label{sec:notations:1}
We use the following conventions throughout the paper:
\begin{itemize}
\item Vectors are denoted using boldface letters.
\item Sans-serif letters are being used to denote either random variables
      or quantum systems.
\item Lower and upper indices are used as the starting and ending indices,
      respectively, of the elements in a vector or an ordered collections
      of random variables or quantum systems. For example,
      \begin{itemize}
          \item $\vx_1^n \equiv (x_1,x_2,\ldots,x_n)$ denotes an 
                $n$-tuple with elements $x_1$ up to $x_n$;
          \item $\rv{X}_1^n\equiv (\rv{X}_1,\rv{X}_2,\ldots,
                 \rv{X}_n)$ denotes a sequence of random variables;
          \item $\system{S}_0^n\equiv (\system{S}_0,\system{S}_1,\ldots,
                 \system{S}_n)$ denotes the collective quantum system 
                 consisting of subsystems $\system{S}_0$ up to $\system{S}_n$.
      \end{itemize}
\item The set of all density operators over a Hilbert space $\hilbert$ is
      denoted by $\DensOp\left(\hilbert\right)$; its elements are represented 
      using Greek letters, \eg, $\rho_\system{S}$ denotes a density operator
      of some quantum system $\system{S}$.
\end{itemize}
As it should be clear from the context, we also overload the symbol $\entropy$ 
to denote either the Shannon entropy or the von Neumann entropy, and the 
symbol $\mutualInfo$ to denote either the classical or quantum mutual 
information. 
\section{Review of (Classical) Finite-State Machine Channels:
         Information Rate, its Estimation, and Bounds}
\label{sec:2:FSMC}
In this section, we review the methods developed 
in~\cite{arnold2006simulation} for estimating the information rate of a
(classical) FSMC,
and the auxiliary-channel-induced upper and lower bounds studied
in~\cite{sadeghi2009optimization}. As we will see, the development in later
sections about quantum channels will have many similarities, but also some
important differences. We emphasize that this section is a \emph{brief review}
of~\cite{arnold2006simulation} and~\cite{sadeghi2009optimization} for the 
purpose of introducing necessary tools and ideas for later sections.
\subsection{Finite-State Machine Channels (FSMCs) and their Graphical 
            Representation}
A (time-invariant) finite-state machine channel (FSMC) consists of an input 
alphabet $\set{X}$, an output alphabet $\set{Y}$, a state alphabet $\set{S}$,
all of which are finite, and a channel law $W(y,s'|x,s)$, where the latter
equals the probability of receiving $y\in\set{Y}$ and ending up in state 
$s'\in\set{S}$  given channel input $x\in\set{X}$ and previous channel 
state $s\in\set{S}$. The relationship among the input, output, and state 
processes $\rv{X}_1^n,\rv{Y}_1^n,\rv{S}_0^n$ of $n$-channel uses can be 
described by the conditional pmf
\begin{equation}\label{eq:FSMC:n:conditional}
\begin{aligned}
W(\vy_1^n,\vs_1^n|\vx_1^n,s_0) 
&\defeq P_{\rv{Y}_1^n,\rv{S}_1^n|\rv{X}_1^n,\rv{S}_0}
       (\vy_1^n,\vs_1^n|\vx_1^n,s_0)\\
&= \Prod_{\ell=1}^{n} W(y_\ell,s_\ell|x_\ell,s_{\ell-1}),
\end{aligned}
\end{equation}
where $x_\ell\in\set{X}$, $y_\ell\in\set{Y}$, and $s_\ell\in\set{S}$ for each
$\ell$.
\begin{example}[Gilbert--Elliott channels]\label{example:GEC}
A notable class of examples of FSMCs are the Gilbert--Elliott
channels~\cite{mushkin1989capacity}, 
which behave like a binary symmetric channel (BSC) with cross-over
probability $p_s$ controlled by the channel state
$s\in\{``\mathrm{b}",``\mathrm{g}"\}$, where usually
$\bigl| \pbad - \frac{1}{2} \bigr| < \bigl| \pgood - \frac{1}{2} \bigr|$.
The state process itself is a first-order stationary
ergodic Markov process that is independent of the input
process.\footnote{The independence of the state process on the input process
is a particular feature of the Gilbert--Elliott channel. In general, the
state process of a finite-state channel can depend on the input process.}
(For more details, see, \eg, the discussions in~\cite{sadeghi2009optimization}.)
\end{example}
Given an input process $\{Q^{(n)}\}_{n}$ and an initial state pmf
$P_{\rv{S}_0}(s_0)$, we can write down the joint pmf of 
$(\rv{X}_1^n,\rv{Y}_1^n,\rv{S}_0^n)$ as 
\begin{align}
\label{eq:FSMC:glabal:distribution}
g(\vx_1^n,\vy_1^n,\vs_0^n)
&\defeq P_{\rv{X}_1^n,\rv{Y}_1^n,\rv{S}_0^n}(\vx_1^n,\vy_1^n,\vs_0^n)\\
&\hspace{-20pt}=P_{\rv{S}_0}(s_0)\cdot Q^{(n)}(\vx_1^n) \cdot
    \Prod_{\ell=1}^{n} W(y_\ell,s_\ell|x_\ell,s_{\ell-1}).
\end{align}
The factorization of $g(\vx_1^n,\vy_1^n,\vs_0^n)$ as shown
in~\eqref{eq:FSMC:glabal:distribution} can be visualized with the help of 
a normal factor graph (NFG) as in Fig.~\ref{fig:FMSC:high:level:1}.
In this context, $g(\vx_1^n,\vy_1^n,\vs_0^n)$ is called the global function
of the NFG.
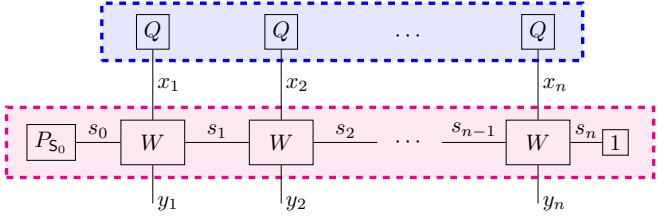
\begin{figure}
  \centering\resizebox{\columnwidth}{!}{
     \begin{tikzpicture}[
    factor/.style={rectangle, minimum width=1cm, minimum height=.7cm,draw},
    sfactor/.style={rectangle, minimum size=.4cm, draw}]
\node[sfactor] (S) {$P_{\rv{S}_0}$};
\node[factor] (E1) [right=.7cm of S] {$W$};
\draw (S) -- (E1) node[above=-.05cm,midway] {$s_0$};
\node[sfactor] (X1) [above=1.1cm of E1] {$Q$};
\draw (X1) -- (E1) node[right=-.05cm, midway] {$x_1$};
\draw (E1.south) -- ([yshift=-.6cm]E1.south) node[right=-.05cm] {$y_1$};

\node[factor] (E2) [right=1cm of E1] {$W$};
\draw (E1) -- (E2) node[above=-.05cm,midway] {$s_1$};
\node[sfactor] (X2) [above=1.1cm of E2] {$Q$};
\draw (X2) -- (E2) node[right=-.05cm, midway] {$x_2$};
\draw (E2.south) -- ([yshift=-.6cm]E2.south) node[right=-.05cm] {$y_2$};

\node[factor,draw=none] (Edummy) [right=1cm of E2] {$\cdots$};
\node[sfactor,draw=none] (Xdummy) [above=1.1cm of Edummy] {$\cdots$};

\node[factor] (En) [right=1cm of Edummy] {$W$};
\draw (E2) -- (Edummy) node[above=-.05cm,midway] {$s_2$};

\draw (Edummy) -- (En) node[above=-.05cm,midway] {$s_{n-1}$};
\node[sfactor] (Xn) [above=1.1cm of En] {$Q$};
\draw (Xn) -- (En) node[right=-.05cm, midway] {$x_n$};
\draw (En.south) -- ([yshift=-.6cm]En.south) node[right=-.05cm] {$y_n$};

\node[sfactor, inner sep=0pt] (ee) [right=.5cm of En] {$1$};
\draw (En) -- (ee) node[above=-.05cm, midway] {$s_n$};
    
\begin{pgfonlayer}{bg}
    \draw[dashed, blue, line width=1.5pt, fill=blue!10] 
        ([xshift=-.8cm,yshift=0.45cm]X1) rectangle 
        ([xshift=.7cm,yshift=-0.45cm]Xn);
    \draw[dashed, magenta, line width=1.5pt, fill= magenta!10] 
        ([xshift=-.7cm,yshift=0.55cm]S) rectangle
        ([xshift=.6cm,yshift=-0.55cm]ee);
\end{pgfonlayer}
\end{tikzpicture}}
  \caption{Channel with a classical state: closing the {\color{cyan}top} box
  yields the input process $Q^{(n)}$, closing the
  {\color{magenta}bottom} box yields the joint channel law 
  $W(\vy_1^n|\vx_1^n)$.}
  \label{fig:FMSC:high:level:1}
\end{figure}
In particular:
\begin{enumerate}[label=\alph*)]
\item \label{closing:the:box} The part of the NFG inside the
  {\color{magenta}bottom} box
  represents $W\left(\vy_1^n,\vs_1^n|\vx_1^n,s_0\right)$, \ie, the
  probability of obtaining $\vy_1^n$ and $\vs_1^n$ given $\vx_1^n$ and $s_0$.
  After applying the \emph{closing-the-box operation}, \ie, after summing over 
  all the variables associated with edges completely inside the
  {\color{magenta}bottom} box, we 
  obtain the joint channel law $W(\vy_1^n|\vx_1^n)\defeq \sum_{\vs_0^n}
  P_{\rv{S}_0}(s_0)\cdot W(\vy_1^n,\vs_1^n|\vx_1^n,s_0) $.
\item The part of the NFG inside the {\color{cyan}top} box represents the
  input process
  $Q^{(n)}(\vx_1^n)$. Here, for simplicity, the input process is an i.i.d. 
  process characterized by the pmf $Q$, \ie, $Q^{(n)}(\vx_1^n) = 
  \prod_{\ell=1}^{n} Q(x_{\ell})$.
\item The function $g(\vx_1^n,\vy_1^n) \defeq \sum_{\vs_0^n} g(\vx_1^n,
  \vy_1^n,\vs_0^n)$, which is obtained by summing the global function
  $g(\vx_1^n,\vy_1^n,\vs_0^n)$ over $\vs_0^n$,
  represents the corresponding marginal pmf over $\vx_1^n$ and $\vy_1^n$.
  The function $g(\vs_0^n)\defeq \sum_{\vx_1^n,\vy_1^n} g(\vx_1^n,\vy_1^n,
  \vs_0^n)$, which is obtained by summing the global function over $\vx_1^n$
  and $\vy_1^n$, represents the corresponding marginal pmf over $\vs_0^n$.
  Other marginal pmfs can be obtained similarly.
\end{enumerate}
Equipped with the notion of the closing-the-box operation
(see item~\ref{closing:the:box} above), such NFG representations can be 
useful in computing a number of quantities of interests. For example, to prove 
that~\eqref{eq:FSMC:n:conditional} is indeed a valid conditional pmf,
it suffices to show that
\begin{equation}\label{eq:verfy:FSMC:conditional:distribution}
\sum_{\vs_1^n,\vy_1^n} W(\vy_1^n,\vs_1^n|\cvx_1^n,\cs_0) = 1
\quad \forall \cvx_1^n\in\set{X}^{n},\,\cs_0\in\set{S},
\end{equation}
which can be verified via a sequence of closing-the-box operations as shown in 
Fig.~\ref{fig:FMSC:high:level:2} in the Appendix.
Such techniques are at the heart of the information-rate-estimation methods
as in~\cite{arnold2006simulation}.
The details are reviewed in the next subsection.
\subsection{Information Rate Estimation}
\label{sec:FSMC:IR}
The approach of~\cite{arnold2006simulation} for estimating information rates of 
FSMCs, as reviewed in this section, is based on the Shannon--McMillan--Breiman 
theorem (see \eg,~\cite{cover2012elements})
and suitable generalizations.
We make the following assumptions.
\begin{itemize}
\item As already mentioned, the derivations in this paper are for the case
  where the input process $\rv{X} = (\rv{X}_1,\rv{X}_2,\ldots)$ is an i.i.d.
  process. The results can be generalized to other stationary ergodic input
  processes that can be represented by a finite-state-machine source
  (FSMS). Technically, this is done by defining a new state that combines the
  source state and the channel state.
\item We assume that the FSMC is indecomposable, which roughly means that in
  the long term the behavior of the channel is independent of the initial
  channel state distribution $P_{\rv{S}_0}$
  (see~\cite[Section~4.6]{gallager1968information}
  for the exact definition). For such channels and stationary ergodic input 
  processes, the information rate $\infoRate_W$ in~\eqref{eq:def:fsmc:ir:1} is
  well defined.
\end{itemize}
Let $W(\vy_1^n|\vx_1^n)$ be the joint channel law of an FSMC satisfying the
assumptions above. As aforementioned, the information rate of such a channel 
using the i.i.d. input distribution $\{Q^{(n)}\defeq Q^{\tensor{n}}\}_{n}$ is 
given by~\eqref{eq:def:fsmc:ir:1}, \ie, by
\begin{equation*}
\infoRate_W(Q) = \lim_{n\rightarrow\infty}
    \frac{1}{n}\mutualInfo(\rv{X}_1^n;\rv{Y}_1^n),
\tag{\ref{eq:def:fsmc:ir:1}'}
\end{equation*}
where the input process $\rv{X}_1^n$ and the output process $\rv{Y}_1^n$ are
jointly distributed according to 
\begin{equation}\label{eq:joint:fsmc}
P_{\rv{X}_1^n,\rv{Y}_1^n}\left(\vx_1^n,\vy_1^n\right) = 
    \Prod_{\ell=1}^{n} Q(x_\ell) \cdot W(\vy_1^n|\vx_1^n).
\end{equation}
One can rewrite~\eqref{eq:def:fsmc:ir:1} as
\begin{equation}\label{eq:def:fsmc:ir:2}
    \infoRate_W(Q) = 
    \entropicRate(\rv{X})+\entropicRate(\rv{Y})-\entropicRate(\rv{X},\rv{Y}),  
\end{equation}
where the \emph{entropic rates} $\entropicRate(\rv{X})$, $\entropicRate(\rv{Y})$
and $\entropicRate(\rv{X},\rv{Y})$ are defined as
\begin{align}
    \entropicRate(\rv{X}) &\defeq
        \lim_{n\to\infty} \frac{1}{n} \entropy(\rv{X}_1^n), \\
    \entropicRate(\rv{Y}) &\defeq 
        \lim_{n\to\infty} \frac{1}{n} \entropy(\rv{Y}_1^n), \\
    \entropicRate(\rv{X},\rv{Y}) &\defeq 
        \lim_{n\to\infty} \frac{1}{n} \entropy(\rv{X}_1^n,\rv{Y}_1^n).
  \end{align}
\par 
We proceed as in~\cite{arnold2006simulation}. (For more background
information, see the references in~\cite{arnold2006simulation}, in 
particular~\cite{ephraim2002hidden}.)
Namely, because of~\eqref{eq:def:fsmc:ir:2} and 
\begin{alignat}{3}
  && -\frac{1}{n}\log{P_{\rv{X}_1^n}(\rv{X}_1^n)} 
      &\stackrel{n\rightarrow\infty}
      {\longrightarrow} \entropicRate(\rv{X}) 
      \quad && \text{w.p. $1$},
      \label{eq:converge:x:1} \\
  && -\frac{1}{n}\log{P_{\rv{Y}_1^n}(\rv{Y}_1^n)} 
      &\stackrel{n\rightarrow\infty}
      {\longrightarrow}\entropicRate(\rv{Y})
      \quad && \text{w.p. $1$},
      \label{eq:converge:y:1} \\
  && -\frac{1}{n}\log{P_{\rv{X}_1^n,\rv{Y}_1^n}(\rv{X}_1^n,\rv{Y}_1^n)} 
      &\stackrel{n\rightarrow\infty}
      {\longrightarrow} \entropicRate(\rv{X},\rv{Y})
      \quad && \text{w.p. $1$},
      \label{eq:converge:xy:1}
\end{alignat}
by choosing some large number $n$, we have the approximation
\begin{equation}\label{eq:FSMC:ir:estimate:1}
    \infoRate_W(Q) \approx 
    \begin{aligned}[t]
        &-\frac{1}{n}\log{P_{\rv{X}_1^n}(\cvx_1^n)}
        -\frac{1}{n}\log{P_{\rv{Y}_1^n}(\cvy_1^n)}\\
        &+\frac{1}{n}\log{P_{\rv{X}_1^n,\rv{Y}_1^n}(\cvx_1^n,\cvy_1^n)},
    \end{aligned}
\end{equation}
where $\cvx_1^n$ and $\cvy_1^n$ are some input and output sequences,
respectively, randomly generated according to
\begin{equation}\label{eq:fsmc:joint:xy:distribution}
\resizebox{.88\columnwidth}{!}{$\displaystyle\hspace{-12pt}
P_{\rv{X}_1^n,\rv{Y}_1^n}(\cvx_1^n,\cvy_1^n) =
    \Sum_{\vs_0^n} P_{\rv{S}_0}(s_0)\cdot Q^{(n)}(\cvx_1^n)
    \cdot W(\cvy_1^n,\vs_1^n|\cvx_1^n,s_0),$}\hspace{-10pt}
\end{equation}
where $W(\cvy_1^n,\vs_1^n|\cvx_1^n,s_0)$ is defined 
in~\eqref{eq:FSMC:n:conditional}.
Note that $\cvx_1^n$ can be obtained by simulating the input process, and 
$\cvy_1^n$ can be obtained by simulating the channel for the given input string 
$\cvx_1^n$. The latter can be done by keeping track of
    $P_{\rv{Y}_\ell|\rv{X}_1^\ell,\rv{Y}_1^{\ell\!-\!1}}
        (y_\ell|\cvx_1^\ell,\cvy_1^{\ell\!-\!1})$,
which is proportional to
    $P_{\rv{Y}_\ell,\rv{Y}_1^{\ell\!-\!1}|\rv{X}_1^\ell}
        (y_\ell,\cvy_1^{\ell\!-\!1}|\cvx_1^\ell)$,
and can be efficiently calculated by applying suitable closing-the-box
operations as in Fig.~\ref{fig:CFSM:channel:simulation:Y} in the Appendix.
\par 
We continue by showing how the three terms appearing on the right-hand side
of~\eqref{eq:FSMC:ir:estimate:1} can be computed
efficiently. We show it explicitly for the second term, and then outline it
for the first and the third term.
\par 
In order to efficiently compute the second term on the right-hand side
of~\eqref{eq:FSMC:ir:estimate:1}, \ie,
$-\frac{1}{n}\log{P_{\rv{Y}_1^n}(\cvy_1^n)}$, we consider the
\emph{state metric} defined in~\cite{arnold2006simulation} as
\begin{equation}\label{eq:def:classical:channel:state:metric:Y:1}
\muY_{\ell}(s_{\ell}) \defeq
    \sum_{\vx_1^{\ell}} \sum_{\vs_0^{\ell\!-\!1}}
    P_{\rv{S}_0}(s_0)\cdot Q^{(\ell)}(\vx_1^\ell)\cdot 
    W(\cvy_1^\ell, \vs_1^\ell | \vx_1^\ell,  s_0).
\end{equation}
In this case,
\begin{equation} \label{eq:calculate:p:y:1}
P_{\rv{Y}_1^n}(\cvy_1^n) = \sum_{s_n} \muY_n(s_n),
\end{equation}
and the calculation of $\muY_\ell(s_{\ell})$ can be done iteratively as
\begin{align}
\muY_\ell(s_{\ell})
  &= \sum_{x_{\ell}}\sum_{s_{\ell\!-\!1}}
    \muY_{\ell\!-\!1}(s_{\ell\!-\!1}) \!\cdot\! 
    Q(x_{\ell} | \vx_1^{\ell\!-\!1}) \!\cdot\!
    W(\cy_{\ell},s_{\ell} | x_{\ell},s_{\ell\!-\!1})\nonumber\\
  &= \sum_{x_{\ell}} \sum_{s_{\ell\!-\!1}}
    \muY_{\ell\!-\!1}(s_{\ell\!-\!1}) \!\cdot\! 
     Q(x_{\ell})\!\cdot\!
    W(\cy_{\ell},s_{\ell} |  x_{\ell},s_{\ell\!-\!1}).
    \label{eq:recursive:state:metric:Y:1}
\end{align}
Eq.~\eqref{eq:recursive:state:metric:Y:1} is visualized in 
Fig.~\ref{fig:CFSM:estimate:hY} as applying suitable closing-the-box operations
to the NFG in Fig.~\ref{fig:FMSC:high:level:1}.
\par 
However, since the value of $\muY_{\ell}(s_\ell)$ tends to zero as
$\ell$ grows, such recursive calculations are numerically inconvenient.
A solution is to \emph{normalize} $\muY_{\ell}(s_\ell)$ after each use of
\eqref{eq:recursive:state:metric:Y:1} and to keep track of the scaling
coefficients. Namely,
\begin{equation}\label{eq:recursive:state:metric:Y:2}
\bmuY_{\ell}(s_\ell) \defeq 
    \frac{1}{\lambdaY_\ell} \sum_{x_\ell}\sum_{s_{\ell\!-\!1}}
    \bmuY_{\ell\!-\!1}(s_\ell) \!\cdot\!  Q(x_\ell) \!\cdot\!
    W(\cy_\ell,s_\ell|x_\ell,s_{\ell\!-\!1}),\hspace{-1pt}
\end{equation}
where the scaling factor $\lambdaY_\ell > 0$ is chosen such that
$\sum_{s_\ell} \bmuY_\ell(s_\ell) = 1$. With this,
Eq.~\eqref{eq:calculate:p:y:1} can be rewritten as
\begin{equation} \label{eq:calculate:p:y:2}
    P_{\rv{Y}_1^n}(\cvy_1^n) = \Prod_{\ell=1}^n \lambdaY_{\ell}.
\end{equation}
Finally, we arrive at the following efficient procedure for computing
$-\frac{1}{n}\log{P_{\rv{Y}_1^n}(\vy_1^n)}$:
\begin{itemize}
\item For $\ell = 1,\ldots, n$, iteratively compute the normalized state
    metric and with that the scaling factors $\lambdaY_{\ell}$.
\item Conclude with the result
\begin{equation} \label{eq:calculate:p:y:3}
-\frac{1}{n}\log{P_{\rv{Y}_1^n}(\vy_1^n)} = \frac{1}{n} \sum_{\ell=1}^n
    \log(\lambdaY_\ell).
\end{equation}
\end{itemize}
\par 
The third term on the right-hand side
of~\eqref{eq:FSMC:ir:estimate:1} can be evaluated
by an analogous procedure, where the state metric $\muY_\ell(s_\ell)$ 
is replaced by the state metric
\begin{equation}\label{eq:def:classical:channel:state:metric:XY:1}
  \muXY_\ell(s_\ell) \defeq
  \sum_{\vs_0^{\ell\!-\!1}} P_{\rv{S}_0}(s_0) \cdot Q^{(\ell)}(\cvx_1^\ell)
  \cdot W(\cvy_1^\ell,\vs_1^\ell|\cvx_1^\ell).      
\end{equation}
The iterative calculation of $\muXY_{\ell}(s_{\ell})$ is visualized in
Fig.~\ref{fig:CFSM:estimate:hXY}.
\par 
Finally, the first term on the right-hand side of~\eqref{eq:FSMC:ir:estimate:1}
can be trivially evaluated if $\rv{X}$ is an i.i.d. process, and with a similar
approach as above if it is described by an FSMS.
\par 
The above discussion is summarized as Algorithm~\ref{alg:SPA}.
On the side, note that for each $\ell=2,\ldots,n$, the quantities
$\lambdaY_\ell$ and $\lambdaXY_\ell$ in the algorithm
are the conditional probabilities
$P_{\rv{Y}_\ell|\rv{Y}_{1}^{\ell\!-\!1}}(\cy_\ell|\cvy_1^{\ell\!-\!1})$ and 
$P_{\rv{X}_\ell\rv{Y}_\ell|\rv{X}_{1}^{\ell\!-\!1}\rv{Y}_{1}^{\ell\!-\!1}}(\cx_{\ell},\cy_\ell|\cvx_1^{\ell\!-\!1},\cvy_1^{\ell\!-\!1})$,
respectively.
\subsection{Auxiliary Channels and Bounds on the Information Rate}
\label{sec:aux}
As already mentioned in Section~\ref{sec:1:Introduction}, auxiliary 
channels\footnote{Technically speaking, an auxiliary channel can be 
defined as \emph{any} channel with the same input/output alphabet.
For example, an auxiliary channel for an 
FSMC can be just another FSMC with smaller state space; whereas in 
Section~\ref{sec:5:UBLB}, an auxiliary channel can also be a 
quantum-state channel.} are introduced when the state space of the FSMC
is too large, making the calculation in Algorithm~\ref{alg:SPA} (pratically)
intractable. More precisely, given an auxiliary forward FSMC (AF-FSMC)
$\hat{W}(y_\ell,\hat{s}_\ell|x_\ell,\hat{s}_{\ell\!-\!1})$
and an auxiliary backward FSMC (AB-FSMC)
$\hat{V}(x_\ell,\hat{s}_\ell|y_\ell,\hat{s}_{\ell\!-\!1})$, 
a pair of upper and lower bounds of the information rate
is given in~\cite{arnold2006simulation, sadeghi2009optimization} as
\begin{align}
\label{eq:def:IRUB}
\IRUB^{(n)}_{W}(\hat{W}) &\defeq 
    \frac{1}{n} \Sum_{\vx_1^n,\vy_1^n}
    Q(\vx_1^n) W(\vy_1^n|\vx_1^n)
    \log{\frac{W(\vy_1^n|\vx_1^n)}{\QWaux(\vy_1^n)}},\\
\label{eq:def:IRLB}
\IRLB_{W}^{(n)}(\hat{V})  &\defeq
    \frac{1}{n} \Sum_{\vx_1^n,\vy_1^n}
    Q(\vx_1^n) W(\vy_1^n|\vx_1^n)
    \log{\frac{\hat{V}(\vx_1^n|\vy_1^n)}{Q(\vx_1^n)}},
\end{align}
where $\QWaux(\vy_1^n)\defeq\sum_{\vx_1^n} Q(\vx_1^n)\cdot
       \hat{W}(\vy_1^n|\vx_1^n)$.
To see that~\eqref{eq:def:IRUB} and~\eqref{eq:def:IRLB} are, respectively,
upper and lower bounds, one can verify the following two equalities,
\begin{align}
\label{eq:IRUB:minus:IR}
&\IRUB_{W}(\hat{W})-\infoRate_{W} =\frac{1}{n} \infdiv{\QW(\rv{Y}_1^n)}
{\QWaux(\rv{y}_1^n)},\\
\label{eq:IR:minus:IRLB}
&\begin{aligned}
\infoRate_{W}-\IRLB_{W}(\hat{V}) = \frac{1}{n} \sum_{\vy_1^n} & \QW(\vy_1^n)\cdot\\
&\infdiv{V(\rv{X}_1^n|\vy_1^n)}{\hat{V}(\rv{X}_1^n|\vy_1^n)},
\end{aligned}
\end{align}
where $\infdiv{\cdot}{\cdot}$ stands for the Kullback–Leibler (KL) divergence,
and where the backward channel $V(\vx|\vy)$ is defined as
$V(\vx|\vy)\defeq Q(\vx)W(\vy|\vx)/\QW(\vy)$.
In particular, given an AF-FSMC $\hat{W}$, the 
paper~\cite{sadeghi2009optimization} considered the induced AB-FSMC
$\hat{V}(\vx|\vy)\defeq Q(\vx)\hat{W}(\vy|\vx)/\QWaux(\vy)$. In this case, 
\begin{equation}
\IRLB_{W}^{(n)}(\hat{V}) =
    \frac{1}{n} \Sum_{\vx_1^n,\vy_1^n}
    Q(\vx_1^n) W(\vy_1^n|\vx_1^n)
    \log\frac{\hat{W}(\vy_1^n|\vx_1^n)}{\QWaux(\vy_1^n)}.
\end{equation}
The difference function $\Delta_{W}^{(n)}(\hat{W})$ is defined as
\begin{equation}\label{eq:delta}
\hspace{0pt}
\begin{aligned}[t]
\Delta_{W}^{(n)}(\hat{W}) &\defeq \IRUB^{(n)}_W(\hat{W}) - \IRLB_{W}^{(n)}(\hat{V})\\
&=\frac{1}{n} \Sum_{\vx_1^n,\vy_1^n}Q(\vx_1^n)W(\vy_1^n|\vx_1^n)
    \log\left(\frac{W(\vy_1^n|\vx_1^n)}{\hat{W}(\vy_1^n|\vx_1^n)}\right)\\
&=\frac{1}{n} \infdiv{\!Q(\rv{X}_1^n) W(\rv{Y}_1^n|\rv{X}_1^n)\!}
    {\!Q(\rv{X}_1^n) \hat{W}(\rv{Y}_1^n|\rv{X}_1^n)\!}.
\end{aligned}\hspace{-30pt}
\end{equation}
Apparently, $\Delta_{W}^{(n)}(\hat{W})\geqslant 0$, and equality holds if and only
if $\hat{W}(\vy_1^n|\vx_1^n)=W(\vy_1^n|\vx_1^n)$ for all $\vx_1^n$ and $\vy_1^n$
with positive support w.r.t. $P_{\rv{X}_1^n,\rv{Y}_1^n}$ defined
in~\eqref{eq:joint:fsmc}. An efficient algorithm for finding a local minimum of
the difference function was proposed in~\cite{sadeghi2009optimization};
we refer to~\cite{sadeghi2009optimization} for further details.
\begin{algorithm}[t]
\caption{Estimating the information rate of an FSMC}
\begin{algorithmic}[1] 
\Require{indecomposable FSMC channel law $W$,
         input distribution $Q$,
         positive integer $n$ large enough.}
\Ensure{$\infoRate_W(Q) \approx \entropy(\rv{X}) +
         \hat\entropicRate(\rv{Y}) - \hat\entropicRate(\rv{X},\rv{Y})$.}
\State Initialize the channel state distribution $P_{\rv{S}_0}$ as a uniform
       distribution over $\set{S}$
\State Generate an input sequence $\cvx_1^n \sim Q^{\tensor n}$
\State Generate a corresponding output sequence $\cvy_1^n$
\State $\bmuY_{0}\gets P_{\rv{S}_0}$
\ForEach{$\ell=1,\ldots,n$}
\State $\muY_\ell(s_\ell) \!\gets\! \sum_{x_\ell,s_{\ell\!-\!1}}
    \bmuY_{\ell\!-\!1}(s_{\ell\!-\!1}) \!\cdot\! Q(x_\ell) \!\cdot\!
    W(\cy_\ell,s_\ell|x_\ell,s_{\ell\!-\!1})$
\State $\lambdaY_\ell \gets \sum_{s_\ell} \muY_\ell(s_\ell)$;
\State $\bmuY_\ell \gets \muY_\ell / \lambdaY_\ell$
\EndFor
\State $\hat\entropicRate(\rv{Y})\gets
        -\frac{1}{n} \sum_{\ell=1}^n \log(\lambdaY_\ell)$
\State $\bmuXY_{0}\gets P_{\rv{S}_0}$
\ForEach{$\ell=1,\ldots,n$}
\State $\muXY_\ell(s_\ell) \!\gets\! \sum_{s_{\ell\!-\!1}}
    \bmuXY_{\ell\!-\!1}(s_{\ell\!-\!1}) \!\cdot\! Q(\cx_\ell) \!\cdot\!
    W(\cy_\ell,s_\ell|\cx_\ell,s_{\ell\!-\!1})$
\State $\lambdaXY_\ell \gets \sum_{s_\ell} \muXY_\ell(s_\ell)$
\State $\bmuXY_\ell \gets \muXY_\ell/\lambdaXY_\ell$
\EndFor
\State $\hat\entropicRate(\rv{X},\rv{Y})\gets
        -\frac{1}{n} \sum_{\ell=1}^n \log(\lambdaXY_\ell)$
\State $\entropy(\rv{X}) \gets -\sum_{x} Q(x) \log{Q(x)}$
\State Estimate $\infoRate_W(Q)$ as $\entropy(\rv{X}) +
       \hat\entropicRate(\rv{Y}) - \hat\entropicRate(\rv{X},\rv{Y})$.
\end{algorithmic}
\label{alg:SPA}
\end{algorithm}
\section{Quantum Channel with Memory and
         their Graphical Representation}\label{sec:3:QCM}
In this section, we formalize our notations and modeling of quantum channels 
with memory~\cite{bowen2004quantum, kretschmann2005quantum, caruso2014quantum}
and of classical communications over such channels. In particular,
we will define a class of channels named \emph{quantum-state channels}, 
which is an alternative description of the classical communications
over quantum channels with memory.
In addition, we will introduce several NFGs for representing
these channels and processes.
\subsection{Classical Communication over a Quantum Channel with Memory}
\label{sec:classical:comm:quantum}
As already mentioned in Section~\ref{sec:1:Introduction}, a quantum channel
with memory is a completely positive trace-preserving (CPTP) map
\begin{equation}\label{eq:def:qcm}
\operator{N}:\DensOp(\hilbert_\system{A}\tensor\hilbert_\system{S})
\rightarrow \DensOp(\hilbert_\system{B}\tensor\hilbert_{\system{S}'}),
\end{equation}
where $\system{A}$ is the input system, $\system{B}$ is the output system,
$\system{S}$ and $\system{S}'$ are, respectively, the memory systems
before and after the channel use. The Hilbert spaces $\hilbert_\system{A}$, 
$\hilbert_\system{B}$, and $\hilbert_\system{S} \equiv 
\hilbert_{\system{S}'}$ are the state spaces corresponding to those
systems.
\par 
In the present paper, we consider classical communication over such
channels using some separable input ensemble and local output measurements;
namely, the encoder and decoder are, respectively, some classical-to-quantum and
quantum-to-classical channels involving a single input or output system.
In particular, given an ensemble $\{\rho^{(x)}_\system{A}\}_{x\in\set{X}}$
and a measurement $\{\Lambda_\system{B}^{(y)}\}_{y\in\set{Y}}$, we define
the encoding and decoding function, respectively, as
\begin{alignat}{2}
&\text{Encoding }\mathcal{E}\!:
    p_\rv{X} \mapsto \sum_{x\in\set{X}} p_\rv{X}(x)\rho^{(x)}_\system{A}
    && \forall\ p_\rv{X} \text{ over }\set{X},\\
&\text{Decoding }\mathcal{D}\!:
    \sigma_\system{B} \mapsto \left\{\tr(\Lambda_\system{B}^{(y)}
    \!\cdot\!\sigma_\system{B})\right\}_{y\in\set{Y}}
    && \forall\ \sigma_\system{B} \text{ over }
       \hilbert_\system{B}.
\end{alignat}
We emphasize that in our setup, the
ensemble $\{\rho^{(x)}_\system{A}\}_{x\in\set{X}}$ and measurements
$\{\Lambda_\system{B}^{(y)}\}_{y\in\set{Y}}$ are given and fixed.
Furthermore, we assume that one does not have access to the memory systems
of the channel.
For the case of i.i.d. inputs, the memory system $\system{S}$ before each
channel use shall be independent of the input system $\system{A}$, namely,
the joint memory-input operator shall take the form of
$\rho_\system{A}\tensor\rho_\system{S}$ at each channel input.\footnote{
More generally, for FSMSs, this statement also holds by conditioning on all
previous inputs.}
\par 
With this, the probability of receiving $y\in\set{Y}$, given that $x\in\set{X}$ 
was sent and given that the density operator of the memory system \emph{before}
the usage of the channel was $\rho_\system{S}$, equals
\begin{equation}\label{eq:channel:law:1}
P_{\rv{Y}|\rv{X};\system{S}}(y|x;\rho_\system{S}) = \tr\left(
    \Lambda_\system{B}^{(y)} \cdot
    \tr_{\system{S}'}\left(
        \operator{N}(\rho^{(x)}_\system{A}\tensor\rho_\system{S})
        \right)
    \right),
\end{equation}
which can also be written as
\begin{equation}\label{eq:channel:law:2}
P_{\rv{Y}|\rv{X};\system{S}}(y|x;\rho_\system{S}) = \tr\left(
    (\Lambda_\system{B}^{(y)}\tensor I_\system{S}) \cdot \operator{N}(
    \rho^{(x)}_\system{A}\tensor\rho_\system{S})\right),
\end{equation}
where $\tr_{\system{S}'}$ stands for the \emph{partial trace} operator
(see, \eg~\cite[Section~2.4.3]{nielsen2011quantum})
that extracts the subsystem $\system{B}$ from the joint system 
$(\system{B}\system{S}')$.
Moreover, assuming that $y$ was observed, the density operator of the
memory system \emph{after} the channel use is given by
\begin{equation}\label{eq:channel:evolution:1}
\rho_{\system{S}'} = \frac{
    \tr_\system{B}\left((\Lambda_\system{B}^{(y)}\tensor I_\system{S}) \cdot
    \operator{N}(\rho^{(x)}_\system{A}\tensor\rho_\system{S})\right)}{
    \tr\left((\Lambda_\system{B}^{(y)}\tensor I_\system{S}) \cdot
    \operator{N}(\rho^{(x)}_\system{A}\tensor\rho_\system{S})\right)}.
\end{equation}
Notice that the denominator in~\eqref{eq:channel:evolution:1} equals the
expressions in~\eqref{eq:channel:law:1} and~\eqref{eq:channel:law:2}.
One should note that, though the input and the memory systems are independent
before each channel use (given i.i.d. inputs), the output and the memory systems
after each channel use can be correlated or even entangled.
In particular, this translates to the fact that the measurement 
outcome $y$ can have an influence on the memory system as indicated
in~\eqref{eq:channel:evolution:1}.
\par 
Consider using the channel $n$ times consecutively with the above scheme.
The joint channel law, namely the conditional pmf of the channel
outputs $\rv{Y}_1^n$ given the channel inputs $\rv{X}_1^n$ and the initial
channel state $\rho_{\system{S}_0}$, can be computed iteratively 
using~\eqref{eq:channel:law:2} and~\eqref{eq:channel:evolution:1}.
In particular, the joint conditional pmf can be computed as
\begin{equation}\label{eq:joint:1}
P_{\rv{Y}_1^n|\rv{X}_1^n;\system{S}_0}(\vy_1^n|\vx_1^n;\rho_{\system{S}_0})
= \prod_{\ell=1}^n P_{\rv{Y}_\ell|\rv{X}_\ell;\system{S}_{\ell\!-\!1}}
                     (y_\ell|x_\ell;\rho_{\system{S}_{\ell\!-\!1}}),
\end{equation}
where we compute the density operators $\{\rho_{\system{S}_\ell}\}_{\ell=1}^n$
iteratively using~\eqref{eq:channel:evolution:1} as
\begin{equation}\label{eq:channel:evolution:2}
\rho_{\system{S}_\ell} = \frac{
    \tr_\system{B}\left((\Lambda_\system{B}^{(y_\ell)}\tensor I_\system{S}) \cdot
    \operator{N}(\rho^{(x_\ell)}_\system{A}\tensor\rho_{\system{S}_{\ell\!-\!1}})
    \right)}{
    \tr\left((\Lambda_\system{B}^{(y_\ell)}\tensor I_\system{S}) \cdot
    \operator{N}(\rho^{(x_\ell)}_\system{A}\tensor\rho_{\system{S}_{\ell\!-\!1}})
    \right)}.
\end{equation}
\subsection{Quantum-State Channels} \label{sec:quantum-state-channel}
For each channel-ensemble-measurement configuration ($\operator{N}$,
$\{\rho^{(x)}_\system{A}\}_{x\in\set{X}}$,
$\{\Lambda_\system{B}^{(y)}\}_{y\in\set{Y}}$) as introduced above,
one ends up with a joint conditional pmf, as 
in~\eqref{eq:joint:1}. However, this relationship is not bijective. In 
particular, consider unitary operators $U_\system{A}$ and $U_\system{B}$ acting
on $\hilbert_\system{A}$ and $\hilbert_\system{B}$, respectively. The following
setup induces exactly the same joint conditional pmf:
\begin{align*}
&\tilde{\operator{N}}: \tilde\rho_\system{AS} \mapsto 
    (U_\system{B}\!\tensor\! I_\system{S}) \cdot
    \operator{N}\left(\!(U_\system{A}\!\tensor\! I_\system{S})
    \tilde\rho_\system{AS}
    (U_\system{A}^\Herm\!\tensor\! I_\system{S})\!\right)
    \cdot (U_\system{B}^\Herm\!\tensor\! I_\system{S}),\\
&\tilde\rho^{(x)}_\system{A} \defeq 
    U_\system{A}^\Herm \cdot \rho^{(x)}_\system{A} \cdot U_\system{A}
    \quad\forall x\in\set{X},\\
&\tilde\Lambda_\system{B}^{(y)} \defeq
    U_\system{B}^\Herm \cdot \Lambda_\system{B}^{(y)} \cdot U_\system{B}
    \quad\forall y\in\set{Y}.
\end{align*}
Such redundancy is not only tedious, but also detrimental when we try to
compare different channels; in particular, when we try to introduce proper
auxiliary channels to approximate the original communication scheme.
\par 
In this subsection, we introduce a class of channels called~\emph{quantum-state
channels} to eliminate such redundancies. In particular, notice that the
statistical behavior of the aforementioned communication scheme is fully
specified via~\eqref{eq:channel:law:2} and~\eqref{eq:channel:evolution:1};
which are in turn determined by the set of completely positive mappings
$\{\operator{N}^{y|x}\}_{x\in\set{X},y\in\set{Y}}$ defined as
\begin{equation}\label{eq:def:qsc}
\operator{N}^{y|x}: \rho_\system{S} \mapsto 
    \tr_\system{B}\left((\Lambda_\system{B}^{(y)}\tensor I_\system{S})\cdot
    \operator{N}(\rho^{(x)}_\system{A}\tensor\rho_\system{S})\right).
\end{equation}
In this case,~\eqref{eq:channel:law:2},~\eqref{eq:channel:evolution:1},
and~\eqref{eq:joint:1} can be rewritten, respectively, as
\begin{align}
\label{eq:channel:law:3}
P_{\rv{Y}|\rv{X};\system{S}}(y|x;\rho_\system{S}) &= 
    \tr\left(\operator{N}^{y|x}(\rho_\system{S})\right),\\
\label{eq:channel:evolution:3}
\rho_{\system{S}'} &=\operator{N}^{y|x}(\rho_\system{S})\big/
    \tr\left(\operator{N}^{y|x}(\rho_\system{S})\right),\\
\label{eq:joint:2}\hspace{-10pt}
P_{\rv{Y}_1^n|\rv{X}_1^n;\system{S}_0}(\vy_1^n|\vx_1^n;\rho_{\system{S}_0})
&=\tr\!\left(\!
  \operator{N}^{y_n|x_n}\!\circ\cdots\circ\!\operator{N}^{y_1|x_1}
  (\rho_{\system{S}_0})\!\right)\!.\hspace{-10pt}
\end{align}
Thus, the operators $\{\operator{N}^{y|x}\}_{x\in\set{X},y\in\set{Y}}$ 
fully specify the joint conditional pmf as in~\eqref{eq:joint:2}.
Moreover, such specification is also \emph{unique}; namely, any two sets of
channel-ensemble-measurement configuration shall end up with the same joint
channel law if and only if the mappings defined in~\eqref{eq:def:qsc} are
identical. This inspires us to make the following definition.
\begin{definition}[Quantum-State Channel] A (finite indexed) set of 
completely positive operators $\{\operator{N}^{y|x}\}_{x\in\set{X},y\in\set{Y}}$
(acting on the same Hilbert space) is said to be a (classical-input
classical-output) \emph{quantum-state channel}
(CC-QSC) if $\sum_{y\in\set{Y}}\operator{N}^{y|x}$ is trace-preserving for each 
$x\in\set{X}$.
\end{definition}
Given any channel-ensemble-measurement configuration as described in
Section~\ref{sec:classical:comm:quantum}, one can always 
define a corresponding CC-QSC by~\eqref{eq:def:qsc}.
On the other hand, as stated in the proposition below, the converse is also
true.
\begin{proposition}\label{prop:quantum:state:channel}
For any CC-QSC
$\{\operator{N}^{y|x}\}_{x\in\set{X},y\in\set{Y}}$, 
there exists some quantum channel with memory $\operator{N}$ as
in~\eqref{eq:def:qcm} such that~\eqref{eq:def:qsc} holds with
ensemble $\{\rho_\system{A}^{(x)}=\braket{x}\}_{x\in\set{X}}$ and
measurement $\{\Lambda_\system{B}^{(y)}=\braket{y}\}_{y\in\set{Y}}$.
Here, $\hilbert_\system{A}$ and $\hilbert_\system{B}$ are defined such that
$\{\bra{x}\}_x$ and $\{\bra{y}\}_y$ are orthonormal bases of 
$\hilbert_\system{A}$ and $\hilbert_\system{B}$, respectively.
\end{proposition}
\begin{proof}
It suffices to show that there exists a CPTP map $\operator{N}:
\DensOp(\hilbert_\system{A}\tensor\hilbert_\system{S})\rightarrow
\DensOp(\hilbert_\system{B}\tensor\hilbert_\system{S})$ such that
for all $\rho_\system{S}\in\DensOp(\hilbert_\system{S})$, and $x\in\set{X}$,
\[
\operator{N}: \braket{x}\tensor\rho_\system{S} \mapsto
    \sum_{y\in\set{Y}} \braket{y}\tensor\operator{N}^{y|x}(\rho_\system{S}).
\]
Such an $\operator{N}$ can be constructed as
\[
\operator{N}: \rho \mapsto
\sum_{x,y,k}\left(\bra{y}\!\ket{x}\otimes E^{y|x}_k\right)\cdot\rho\cdot
\left(\bra{y}\!\ket{x}\otimes E^{y|x}_k\right)^{\Herm},
\]
where $\left\{E^{y|x}_k\right\}_{k}$ is a Kraus representation of
$\operator{N}^{y|x}$, namely,
\[
\operator{N}^{y|x}(\rho_{\system{S}}) \equiv
\sum_{k} E^{y|x}_k\cdot\rho_\system{S}\cdot(E^{y|x}_k)^\Herm
\qquad \forall \rho_\system{S}\in\DensOp(\hilbert_\system{S}).
\]
It remains to check if $\operator{N}$ is a CPTP, which is indeed the case:
\[\begin{aligned}
&\sum_{x,y,k} \left(\bra{y}\!\ket{x}\otimes E^{y|x}_k\right)^{\Herm}\cdot
\left(\bra{y}\!\ket{x}\otimes E^{y|x}_k\right)\\
= &\sum_x\sum_{y,k}
\braket{x} \otimes (E^{y|x}_k)^\Herm E^{y|x}_k = \sum_x \braket{x}\otimes I
= I. \qedhere
\end{aligned}\]
\end{proof}
\subsection{Visualization using Normal Factor Graphs} \label{sec:NFGs}
In this subsection, we focus on the computations of \eqref{eq:channel:law:3},
\eqref{eq:channel:evolution:3}, and~\eqref{eq:joint:2} for the situation where
the involved channel $\operator{N}$ is of finite dimension.
In analogy to the FSMCs, we demonstrate how to use NFGs
to facilitate and visualize the relevant computations. Our use of NFGs for
describing quantum systems follows~\cite{loeliger2017factor}.
\par 
By Proposition~\ref{prop:quantum:state:channel}, let us consider a CC-QSC
$\{\operator{N}^{y|x}\}_{x\in\set{X},y\in\set{Y}}$ acting on
$\hilbert_\system{S}$, where $d=\dim(\hilbert_\system{S})$ is finite, and
$\{\bra{s}\}_{s\in\set{S}}$ is an orthonormal basis of $\hilbert_\system{S}$.
(Apparently, $\size{\set{S}}=d$.) Since for each $x$ and $y$,
$\operator{N}^{y|x}$ is a completely positive map, there must exist finitely
many (\emph{not} necessarily unique) matrices
$\{F_k^{y|x}\in\mathbb{C}^{\set{S}\times\set{S}}\}_{k}$ such that
\begin{equation} \label{eq:Kraus:quantum:state:channel}
    \bigl[\operator{N}^{y|x}(\rho_\system{S})\bigr] \equiv 
    \sum_{k} F_k^{y|x}\cdot [\rho_\system{S}] \cdot (F_k^{y|x})^\Herm
        \quad \forall \rho_\system{S}\in\DensOp(\hilbert_\system{S}),
\end{equation}
where $\bigl[\operator{N}^{y|x}(\rho_\system{S})\bigr]$ and $[\rho_\system{S}]$
are, respectively, the matrix representation of the operator 
$\operator{N}^{y|x}(\rho_\system{S})$ and $\rho_\system{S}$ under
$\{\bra{s}\}_{s\in\set{S}}$. The reason for such matrices $\{F_k^{y|x}\}_{k}$
to exist is the same as for the Kraus operators of CPTP maps
(see~\cite[Theorems~8.1 and~8.3]{nielsen2011quantum}).
Also note that $\sum_{y\in\set{Y}} \mathcal{E}^{y|x}$ is
trace-preserving, thus it must hold that
\begin{equation} \label{eq:operator:sum:representation:condition:1}
    \sum_{y\in\set{Y}}\sum_{k} (F_k^{y|x})^\Herm F_k^{y|x} = I
        \quad \forall x \in\set{X}.
\end{equation}
Now, define a set of functions
$\{W^{y|x}\}_{x\in\set{X},y\in\set{Y}}$ as
\begin{equation}\label{eq:def:channel:function:representation}
    W^{y|x}:(s',s,\tilde{s}',\tilde{s}) \mapsto
    \sum_k F_k^{y|x}(s',s)\conj{F_k^{y|x}(\tilde{s}',\tilde{s})},
\end{equation}
where $s',s,\tilde{s}',\tilde{s}\in\set{S}$ are indices of the corresponding
matrices, namely, $F_k^{y|x}(s',s)$ is the $(s',s)$-th entry of matrix $F_k^{y|x}$.
In this case, one can rewrite 
\eqref{eq:channel:law:3},~\eqref{eq:channel:evolution:3}
and~\eqref{eq:joint:2}, respectively, into
\begin{equation}
\label{eq:channel:law:4}
P_{\rv{Y}|\rv{X};\system{S}}(y|x;\rho_\system{S}) =
    \sum_{s',\tilde{s}':\atop s'=\tilde{s}'} \sum_{s,\tilde{s}}
    W^{y|x}(s',s,\tilde{s}',\tilde{s})\cdot[\rho_\system{S}]_{s,\tilde{s}},	
\end{equation}
\begin{equation}
\label{eq:channel:evolution:4}
[\rho_{\system{S}'}]_{s',\tilde{s}'} =
    \frac{\sum_{s,\tilde{s}} W^{y|x}(s',s,\tilde{s}',\tilde{s})
          \cdot [\rho_\system{S}]_{s,\tilde{s}}}
         {\sum_{s',\tilde{s}':\atop s'=\tilde{s}'}\sum_{s,\tilde{s}}
          W^{y|x}(s',s,\tilde{s}',\tilde{s})
          \cdot [\rho_\system{S}]_{s,\tilde{s}}},
\end{equation}
\begin{equation}
\label{eq:joint:3}
\begin{aligned}
P_{\rv{Y}_1^n|\rv{X}_1^n;\system{S}_0}(\vy_1^n|\vx_1^n;\rho_{\system{S}_0})
    =&\sum_{s_n,\tilde{s}_n:\atop s_n=\tilde{s}_n}
    \sum_{\vs_0^{n\!-\!1},\tilde{\vs}_0^{n\!-\!1}}
    [\rho_{\system{S}_0}]_{s_0,\tilde{s}_0}\cdot \\
    &\prod_{\ell=1}^n W^{y_\ell|x_\ell}
    (s_\ell,s_{\ell\!-\!1},\tilde{s}_\ell,\tilde{s}_{\ell\!-\!1}).
\end{aligned}
\end{equation}
By rearranging the entries of $W^{y|x}$ (for each $x,y$) into a matrix
$[W^{y|x}]\in\mathbb{C}^{\set{S}^2\times\set{S}^2}$ as
\begin{equation}\label{eq:qsc:matrix:1}
    [W^{y|x}]_{(s',\tilde{s}'),(s,\tilde{s})} \defeq
    W^{y|x}(s',s,\tilde{s}',\tilde{s}),
\end{equation}
where $(s',\tilde{s}')\in\set{S}^2$ is the first index, and 
$(s,\tilde{s})\in\set{S}^2$ is the second index of $[W^{y|x}]$,
we can simplify~\eqref{eq:channel:law:4},
\eqref{eq:channel:evolution:4}, and~\eqref{eq:joint:3} as
\begin{align}
\label{eq:channel:law:5}
P_{\rv{Y}|\rv{X};\system{S}}(y|x;\rho_\system{S}) &=
    \tr([W^{y|x}] \cdot [\rho_\system{S}]),\\
\label{eq:channel:evolution:5}
[\rho_{\system{S}'}] &=
    \frac{[W^{y|x}] \cdot [\rho_\system{S}]}
         {\tr([W^{y|x}] \cdot [\rho_\system{S}])},\\
\label{eq:joint:4}\hspace{-10pt}
P_{\rv{Y}_1^n|\rv{X}_1^n;\system{S}_0}(\vy_1^n|\vx_1^n;\rho_{\system{S}_0})&=
    \tr\!\left(\![W^{y_{n}|x_{n}}]\cdots[W^{y_{1}|x_{1}}]
             \!\cdot\![\rho_{\system{S}_0}]\!\right)\!,\hspace{-10pt}
\end{align}
respectively. Here we treat $[\rho_\system{S}]$ as a length-$d^2$ vector
indexed by $(s,\tilde{s})\in\set{S}^2$ in the above equations.
\par 
By considering $\{W^{y|x}\}_{x,y}$ as a function of six variables, we can
represent it using a factor node of degree six in an NFG as in
Fig.~\ref{fig:NFG:QSC:single}. In this case, Eqs.~\eqref{eq:channel:law:4}
and~\eqref{eq:channel:law:5} can be visualized as ``closing the
{\color{magenta} outer} box''
in the NFG, \ie, summing over all the variables represented by the edges
interior to the box. Similarly,~\eqref{eq:channel:evolution:4} 
and~\eqref{eq:channel:evolution:5} can be visualized as ``closing the
{\color{cyan}inner} box''.
The NFG corresponding to using the channel $n$~times consecutively
is depicted in Fig.~\ref{fig:NFG:QSC:multiple}, where~\eqref{eq:joint:3}
and~\eqref{eq:joint:4} are visualized as closing the {\color{magenta}outermost}
box.
Interestingly, this closing-the-box operation can be carried out by a sequence
of simpler closing-the-box operations as shown in the figure.
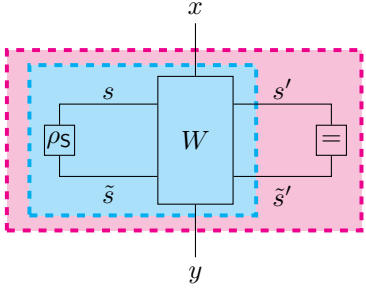
\begin{figure}
\centering
    \begin{tikzpicture}[node/.style={draw=none},
                    factor/.style={rectangle, minimum width=1cm, minimum height=1.7cm, draw},
                    sfactor/.style={rectangle, minimum size=.4cm, draw}]
    \node[factor] (W) {$W$};
    \node[above left=-.5cm and 1.2cm of W] (sp) {};
    \node[below left=-.5cm and 1.2cm of W] (sp') {};
    \node[above right=-.5cm and 1.2cm of W] (s) {};
    \node[below right=-.5cm and 1.2cm of W] (s') {};
    \node[above=.7cm of W] (X) {$x$};
    \node[below=.7cm of W] (Y) {$y$};
    \node[sfactor, inner sep=0pt, midway] (EW) [right=1.6cm of W] {$=$};
    \node[sfactor, inner sep=0pt, midway] (rho) [left=1.6cm of W] {};
    \node at (rho) {$\rho_\system{S}$};
    
    \draw (sp.east-|W.west) -| (rho) node[above=-0.05cm,pos=0.25] {$s$};
    \draw (sp'.east-|W.west) -| (rho) node[below,pos=0.25] {$\tilde{s}$};
    \draw (sp.east-|W.east) -| (EW) node[above=-0.05cm,pos=0.25] {$s'$};
    \draw (sp'.east-|W.east) -| (EW) node[below,pos=0.25] {$\tilde{s}'$};
    \draw (X) -- (W);
    \draw (Y) -- (W);
    
    \begin{pgfonlayer}{bg}
    \draw[dashed, magenta, line width=1.5pt,fill= magenta!30]
        ([xshift=-.7cm,yshift=1.2cm]rho) rectangle
        ([xshift=.4cm,yshift=-1.2cm]EW);
    \draw[dashed, cyan, line width=1.5pt,fill= cyan!30]
        ([xshift=-.4cm,yshift=1cm]rho) rectangle
        ([xshift=-1cm,yshift=-1cm]EW);
    \end{pgfonlayer}
\end{tikzpicture}
\caption{Representation of $\{W^{y|x}\}_{x,y}$ using an NFG.}
  \label{fig:NFG:QSC:single}
\end{figure}
\begin{figure}
\centering
    \begin{tikzpicture}[
    factor/.style ={rectangle, minimum width=1cm, minimum height=1.7cm, draw},
    sfactor/.style={rectangle, minimum size=.4cm, draw},
    label/.style={anchor=south east, circle, draw, inner sep=.5pt, 
        outer sep=5pt, font=\scriptsize}]
\node[sfactor] (S) {}; \node at (S) {\resizebox{.375cm}{!}{$\rho_{\system{S}_0}$}};
\node[factor] (E1) [right=.5cm of S] {$W$};
\draw (S.north) |- ([yshift=.6cm]E1.west)
    node[above=-0.05cm,pos=.75] {$s_0$};
\draw (S.south) |- ([yshift=-.6cm]E1.west) node[below,pos=.75] {$\tilde{s}_0$};
\draw (E1.north) -- ([yshift=1cm]E1.north) node[right] {$x_1$};
\draw (E1.south) -- ([yshift=-1cm]E1.south) node[right] {$y_1$};

\node[factor] (E2) [right=.6cm of E1] {$W$};
\draw ([yshift=.6cm]E1.east) |- ([yshift=.6cm]E2.west)
    node[above=-0.05cm,pos=.75] {$s_1$};
\draw ([yshift=-.6cm]E1.east) |- ([yshift=-.6cm]E2.west)
    node[below,pos=.75] {$\tilde{s}_1$};
\draw (E2.north) -- ([yshift=1cm]E2.north) node[right] {$x_2$};
\draw (E2.south) -- ([yshift=-1cm]E2.south) node[right] {$y_2$};

\node[factor, draw=none] (Edummy) [right=.6cm of E2] {$\cdots$};
\draw ([yshift=.6cm]E2.east) |- ([yshift=.6cm]Edummy.west) 
    node[above=-0.05cm,pos=.75] {$s_2$};
\draw ([yshift=-.6cm]E2.east) |- ([yshift=-.6cm]Edummy.west)
    node[below,pos=.75] {$\tilde{s}_2$};

\node[factor] (En) [right=.6cm of Edummy] {$W$};
\draw ([yshift=.6cm]Edummy.east) |- ([yshift=.6cm]En.west)
    node[above=-0.05cm,pos=.75] {$s_{\!n-\!1}$};
\draw ([yshift=-.6cm]Edummy.east) |- ([yshift=-.6cm]En.west)
    node[below,pos=.75] {$\tilde{s}_{n\!-\!1}$};
\draw (En.north) -- ([yshift=1cm]En.north) node[right] {$x_n$};
\draw (En.south) -- ([yshift=-1cm]En.south) node[right] {$y_n$};

\node[sfactor, draw=none] (Edummy1) [above=-0.50cm of Edummy] {$\cdots$};
\node[sfactor, draw=none] (Edummy2) [below=-0.45cm of Edummy] {$\cdots$};

\node[sfactor,draw=none] (Xdummy) at ([yshift=1cm]E1.north-|Edummy) {$\cdots$};
\node[sfactor,draw=none] (Ydummy) at ([yshift=-1cm]E1.south-|Edummy) {$\cdots$};

\node[sfactor, inner sep=0pt] (ee) [right=.5cm of En] {$=$};
\draw ([yshift=.6cm]En.east) -| (ee.north) node[above=-0.05cm,pos=.25] {$s_n$};
\draw ([yshift=-.6cm]En.east) -| (ee.south) node[below,pos=.25] {$\tilde{s}_n$};

\begin{pgfonlayer}{bg}
\draw[dashed, magenta,line width=1.5pt,fill= magenta!10]
    ([xshift=-.7cm,yshift=1.5cm]S) rectangle
    ([xshift=1.5cm,yshift=-1.6cm]En);
\draw[dashed, magenta,line width=1.5pt,fill= magenta!20]
    ([xshift=-.6cm,yshift=1.4cm]S) rectangle
    ([xshift=.6cm,yshift=-1.5cm]En);
\node[label] at ([xshift=.6cm,yshift=-1.5cm]En) {$n$};
\draw[dashed, magenta,line width=1.5pt,fill= magenta!30]
    ([xshift=-.5cm,yshift=1.3cm]S) rectangle
    ([xshift=.6cm,yshift=-1.4cm]E2);
\node[label] at ([xshift=.6cm,yshift=-1.4cm]E2) {2};
\draw[dashed, magenta,line width=1.5pt,fill= magenta!50] 
    ([xshift=-.4cm,yshift=1.2cm]S) rectangle
    ([xshift=.6cm,yshift=-1.3cm]E1);
\node[label] at ([xshift=.6cm,yshift=-1.3cm]E1) {1};
\end{pgfonlayer}
\end{tikzpicture}
\caption{The joint channel law~\eqref{eq:joint:3} and~\eqref{eq:joint:4} can be
    visualized as the result of the ``closing of the {\color{cyan}outermost}
    box'' above, which can in turn be carried out by a sequence of
    ``closing-the-box'' operations as indicated.}
\label{fig:NFG:QSC:multiple}
\end{figure}
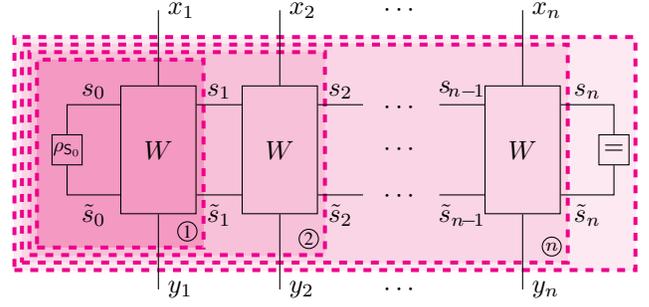
\par 
A number of statistical quantities and density operators of interest can be 
computed and visualized as closing-the-box operations on suitable NFGs similar
to that of Fig.~\ref{fig:NFG:QSC:multiple}.
The following example highlights how quantities of this kind can be computed
in such a manner.
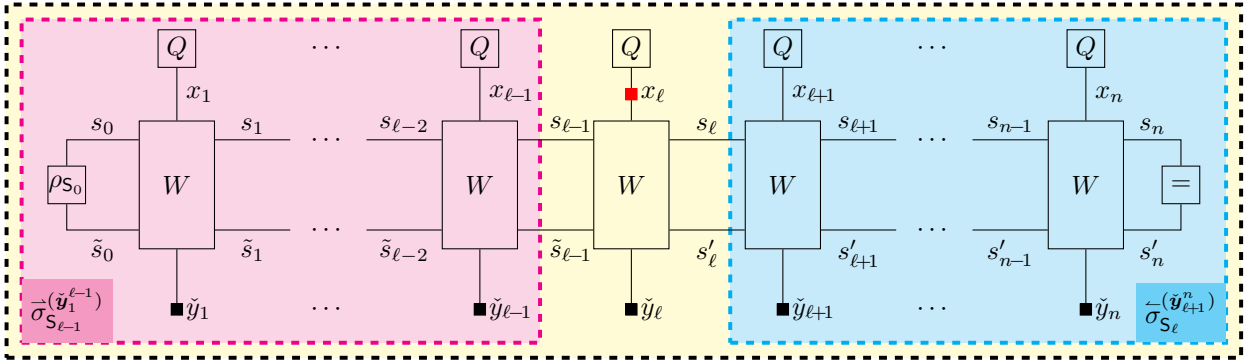
\begin{figure*}[t]
\centering
    \begin{tikzpicture}[
    node/.style={draw=none},
    factor/.style={rectangle, minimum width=1cm, minimum height=1.7cm, draw},
    sfactor/.style={rectangle, minimum size=.5cm, draw},
    darksolid/.style={rectangle, minimum size=.15cm, draw, fill = black,
    inner sep=0pt, outer sep = 0pt}]
\node[sfactor] (S) {}; \node at (S) {$\rho_{\system{S}_0}$};
\node[factor] (E1) [right=.7cm of S] {$W$};
\draw (S.north) |- ([yshift=.6cm]E1.west) node[above=-0.05cm,pos=.75] {$s_0$};
\draw (S.south) |- ([yshift=-.6cm]E1.west) node[below,pos=.75] {$\tilde{s}_0$};
\node[sfactor] (X1) [above=.7cm of E1] {}; \node at (X1) {$Q$};
\draw (X1) -- (E1) node[right, midway] {$x_1$};
\draw (E1.south) -- ([yshift=-.8cm]E1.south) node[darksolid]{} 
    node (Y1) [right] {$\cy_1$};

\node[factor, draw=none] (Edummy1) [right=1cm of E1] {};
\node at ([yshift=.6cm]E1.east-|Edummy1) {$\cdots$};
\node at ([yshift=-.6cm]E1.east-|Edummy1) {$\cdots$};
\draw ([yshift=.6cm]E1.east) |- ([yshift=.6cm]Edummy1.west)
    node[above=-0.05cm,pos=.75] {$s_1$};
\draw ([yshift=-.6cm]E1.east) |- ([yshift=-.6cm]Edummy1.west)
    node[below,pos=.75] {$\tilde{s}_1$};
\node at (X1-|Edummy1) {$\cdots$};
\node at (Y1-|Edummy1) {$\cdots$};

\node[factor] (El0) [right=1cm of Edummy1] {$W$};
\draw ([yshift=.6cm]Edummy1.east) |- ([yshift=.6cm]El0.west)
    node[above=-0.05cm,pos=.75] {$s_{\ell-2}$};
\draw ([yshift=-.6cm]Edummy1.east) |- ([yshift=-.6cm]El0.west)
    node[below,pos=.75] {$\tilde{s}_{\ell-2}$};
\node[sfactor] (Xl0) [above=.7cm of El0] {}; \node at (Xl0) {$Q$};
\draw (Xl0) -- (El0) node[right, midway] {$x_{\ell\!-\!1}$};
\draw (El0.south) -- ([yshift=-.8cm]El0.south) node[darksolid]{}
    node[right] {$\cy_{\ell\!-\!1}$};

\node[factor] (El) [right=1cm of El0] {$W$};
\draw ([yshift=.6cm]El0.east) |- ([yshift=.6cm]El.west) 
    node[above=-0.05cm,pos=.85] {$s_{\ell\!-\!1}$};
\draw ([yshift=-.6cm]El0.east) |- ([yshift=-.6cm]El.west) 
    node[below,pos=.85] {$\tilde{s}_{\ell\!-\!1}$};
\node[sfactor] (Xl) [above=.7cm of El] {}; \node at (Xl) {$Q$};
\draw (El.north) -- (Xl.south) node[darksolid,midway,red,fill=red]{}
    node[right,midway] {$x_{\ell}$};
\draw (El.south) -- ([yshift=-.8cm]El.south) node[darksolid]{}
    node[right] {$\cy_{\ell}$};

\node[factor] (El2) [right=1cm of El] {$W$};
\draw ([yshift=.6cm]El.east) |- ([yshift=.6cm]El2.west)
    node[above=-0.05cm,pos=.75] {$s_\ell$};
\draw ([yshift=-.6cm]El.east) |- ([yshift=-.6cm]El2.west)
    node[below,pos=.75] {$s_\ell'$};
\node[sfactor] (Xl2) [above=.7cm of El2] {}; \node at (Xl2) {$Q$};
\draw (Xl2) -- (El2) node[right, midway] {$x_{\ell\!+\!1}$};
\draw (El2.south) -- ([yshift=-.8cm]El2.south) node[darksolid]{}
    node[right] {$\cy_{\ell\!+\!1}$};
    
\node[factor, draw=none] (Edummy2) [right=1cm of El2] {};
\node at ([yshift=.6cm]El2.east-|Edummy2) {$\cdots$};
\node at ([yshift=-.6cm]El2.east-|Edummy2) {$\cdots$};
\draw ([yshift=.6cm]El2.east) |- ([yshift=.6cm]Edummy2.west)
    node[above=-0.05cm,pos=.75] {$s_{\ell\!+\!1}$};
\draw ([yshift=-.6cm]El2.east) |- ([yshift=-.6cm]Edummy2.west)
    node[below,pos=.75] {$s_{\ell\!+\!1}'$};
\node at (X1-|Edummy2) {$\cdots$};
\node at (Y1-|Edummy2) {$\cdots$};
    
\node[factor] (En) [right=1cm of Edummy2] {$W$};
\draw ([yshift=.6cm]Edummy2.east) |- ([yshift=.6cm]En.west)
    node[above=-0.05cm,pos=.75] {$s_{n\!-\!1}$};
\draw ([yshift=-.6cm]Edummy2.east) |- ([yshift=-.6cm]En.west)
    node[below,pos=.75] {$s_{n\!-\!1}'$};
\node[sfactor] (Xn) [above=.7cm of En] {}; \node at (Xn) {$Q$};
\draw (Xn) -- (En) node[right, midway] {$x_n$};
\draw (En.south) -- ([yshift=-.8cm]En.south) node[darksolid]{}
    node[right] {$\cy_n$};
    
\node[sfactor, inner sep=0pt] (ee) [right=.5cm of En] {$=$};
\draw ([yshift=.6cm]En.east) -| (ee.north) node[above=-0.05cm,pos=.25] {$s_n$};
\draw ([yshift=-.6cm]En.east) -| (ee.south) node[below,pos=.25] {$s_n'$};
    
\begin{pgfonlayer}{bg}
    \draw[dashed, black, line width=1.5pt, fill=yellow!20]
        ([xshift=-.8cm,yshift=2.4cm]S) rectangle ([xshift=.8cm,yshift=-2.3cm]ee);
    \draw[dashed, cyan, line width=1.5pt, fill=cyan!20]
        ([xshift=-.7cm,yshift=2.2cm]El2) rectangle 
        ([xshift=.6cm,yshift=-2.1cm]ee);
    \node[anchor=south east, fill= cyan!50] at ([xshift=.6cm,yshift=-2.1cm]ee)
        {$\lvec{\sigma}_{\system{S}_{\ell}}^{(\cvy_{\ell\!+\!1}^n)}$};
    \draw[dashed, magenta, line width=1.5pt, fill=magenta!20]
        ([xshift=-.6cm,yshift=2.2cm]S) rectangle 
        ([xshift=.8cm,yshift=-2.1cm]El0);
    \node at ([xshift=-.6cm,yshift=2.2cm]S) (tempA) {};
    \node at ([xshift=.8cm,yshift=-2.1cm]El0) (tempB) {};
    \node[anchor=south west,fill=magenta!50] at (tempA |- tempB) 
        {$\rvec{\sigma}_{\system{S}_{\ell\!-\!1}}^{(\cvy_1^{\ell\!-\!1})}$};
\end{pgfonlayer}
\end{tikzpicture}
\caption{Computation of the marginal pmf
$P_{\rv{X}_\ell|\rv{Y}_1^n;\system{S}_0}$ using a sequence of closing-the-box
operations.}
\label{fig:qsc:bcjr}
\end{figure*}
\begin{example}[BCJR~\cite{bahl1974optimal} decoding for CC-QSCs]
For fixed $\cvy_1^n\in\set{Y}^n$ and a given initial density operator
$\rho_{\system{S}_0}$, the conditional probability
$P_{\rv{X}_\ell|\rv{Y}_1^n;\system{S}_0}(x_\ell|\cvy_1^n;\rho_{\system{S}_0})$
can be computed via
\begin{equation}\label{eq:qsc:bcjr:1}
P_{\rv{X}_\ell|\rv{Y}_1^n;\system{S}_0}(\cdot|\cvy_1^n;\rho_{\system{S}_0})
\propto
P_{\rv{X}_\ell,\rv{Y}_1^n|\system{S}_0}(\cdot,\cvy_1^n|\rho_{\system{S}_0}),
\end{equation}
where the right-hand side of~\eqref{eq:qsc:bcjr:1} is a marginal pmf defined as
\begin{equation}\label{eq:qsc:bcjr:2}
\begin{aligned}
P_{\rv{X}_\ell,\rv{Y}_1^n|\system{S}_0}(x_\ell,\cvy_1^n|\rho_{\system{S}_0})
= \sum_{\vx_{1}^{\ell\!-\!1},\vx_{\ell\!+\!1}^{n}}
  \sum_{\vs_0^n,\tilde{\vs}_0^n} [\rho_{\system{S}_0}]_{s_0,\tilde{s}_0}\cdot\\
  \prod_{i=1}^n Q(x_i) \cdot\prod_{j=1}^n W^{\cy_j|x_j}
    (s_j,s_{j\!-\!1},\tilde{s}_j,\tilde{s}_{j\!-\!1}),
\end{aligned}
\end{equation}
where we have assumed that the input process $\rv{X}_1^n$ is i.i.d.
characterized by some pmf $Q$. The evaluation of~\eqref{eq:qsc:bcjr:2} can
be carried out efficiently using a sequence of closing-the-box operations as
visualized in Fig.~\ref{fig:qsc:bcjr}. These operations can be roughly divided
into the following three steps.
\begin{enumerate}
\item Closing the {\color{magenta}left inner} box: this results in an operator 
  $\rvec{\sigma}_{\system{S}_{\ell\!-\!1}}^{(\cvy_1^{\ell\!-\!1})}$ on 
  $\hilbert_{\system{S}_{\ell\!-\!1}}$.
\item Closing the {\color{cyan}right inner} box: this results in another
  operator
  $\lvec{\sigma}_{\system{S}_{\ell}}^{(\cvy_{\ell\!+\!1}^n)}$ on 
  $\hilbert_{\system{S}_{\ell}}$.
\item Applying the closing-the-box operation to the yellow box: the result is
  the marginal pmf $P_{\rv{X}_\ell,\rv{Y}_1^n|\system{S}_0}
  (x_\ell,\cvy_1^n|\rho_{\system{S}_0})$, from which the desired conditional
  probability $P_{\rv{X}_\ell|\rv{Y}_1^n;\system{S}_0}
  (x_\ell|\cvy_1^n;\rho_{\system{S}_0})$ can be easily obtained by
  normalization.
\end{enumerate}
The operators mentioned in 1) and 2) can be computed recursively, using 
a sequence of closing-the-box operations. Namely, one can carry out the
computations in 1) consecutively with $\ell=1,2,\ldots,n$; or the computations
in 2) consecutively with $\ell=n,n\!-\!1,\ldots,1$. This provides an efficient
way to evaluate $P_{\rv{X}_\ell|\rv{Y}_1^n;\system{S}_0}(x_\ell|\cvy_1^n;
\rho_{\system{S}_0})$ for each $\ell=1,\ldots,n$; and thus provides an 
efficient symbol-wise decoding algorithm. The idea in this example is 
conceptually identical to that of the BCJR decoding algorithm for an FSMC.
\end{example}
As shown in the above example, very often the desired functions or quantities
are based on the same partial results. The NFG framework is very helpful to
visualize these partial results and to show how they can be combined to obtain
the desired functions and quantities.
\par 
We emphasize that the functions
$\{W^{y|x}\}_{x,y}$ defined in~\eqref{eq:def:channel:function:representation} 
are unique for a given finite-dimensional CC-QSC $\{\operator{N}^{y|x}\}_{x,y}$; 
even though such uniqueness does not apply to the Kraus operators $
\{F^{y|x}\}_k$ being used to define $\{W^{y|x}\}_{x,y}$. This can be proven
by making the identification that
\begin{equation}\label{eq:qsc:matrix:identification}
\bigl[\operator{N}^{y|x}(\rho_\system{S})\bigr]\equiv
[W^{y|x}]\cdot[\rho_\system{S}]
\qquad \forall \rho_\system{S}\in\DensOp(\hilbert_\system{S}),
\end{equation}
for all $x$ and $y$. Moreover, we argue that the functions $\{W^{y|x}\}_{x,y}$,
are an \emph{equivalent} way to specify a CC-QSC, or classical communication
over a quantum channel with memory as described at the beginning of this
section. Namely, for any set of complex-valued functions $\{W^{y|x}\}_{x,y}$ 
on $\set{S}^{4}$ satisfying some constraints to be clarified later,
there must exist a unique CC-QSC $\{\operator{N}^{y|x}\}_{x,y}$
such that~\eqref{eq:qsc:matrix:identification} holds; and thus, there must
exist some corresponding channel-ensemble-measurement configuration,
unique up to its channel law. As for such constraints, we rearrange the entries
of $W^{y|x}$ (for each $x,y$) into another matrix
$\llbracket W^{y|x}\rrbracket \in\mathbb{C}^{\set{S}^2\times\set{S}^2}$
(a.k.a. Choi--Jamio{\l}kowski matrix~\cite{jamiolkowski1972linear}),
whose entries are defined as
\begin{equation}\label{eq:qsc:matrix:2}
\llbracket W^{y|x}\rrbracket_{(s',s),(\tilde{s}',\tilde{s})} \defeq 
W^{y|x}(s',s,\tilde{s}',\tilde{s}),
\end{equation}
where $(s',s)\in\set{S}^2$ is the first index, and 
$(\tilde{s}',\tilde{s})\in\set{S}^2$ is the second index of
$\llbracket W^{y|x}\rrbracket$. Notice that, $\llbracket W^{y|x}\rrbracket$
is a positive semi-definite (p.s.d.) matrix, and it satisfies the following
equation
\begin{equation}\label{eq:operator:sum:representation:condition:2}
    \sum_{y\in\set{Y}} \sum_{s',\tilde{s}':\: s'=\tilde{s}'}
        \llbracket W^{y|x}\rrbracket_{(s',s),(\tilde{s}',\tilde{s})}
    = \delta_{s,\tilde{s}} \quad \forall x\in\set{X},
\end{equation}
where $\delta_{s,\tilde{s}}$ is the Kronecker-delta function.
In this case, the ``equivalence'' can be shown by the following proposition.
\begin{proposition}\label{prop:CFR}
Let $\set{X}$, $\set{Y}$ be finite sets, and $\hilbert_\system{S}$ be a 
finite-dimensional Hilbert space with an orthonormal basis 
$\{\bra{s}\}_{s\in\set{S}}$. For any set of functions
\[
\{W^{y|x}:\set{S}\times\set{S}\times\set{S}\times\set{S}
          \rightarrow \mathbb{C}\}_{x\in\set{X},y\in\set{Y}}
\]
such that their matrix form $\{\llbracket W^{y|x}\rrbracket\}_{x,y}$ 
consists of p.s.d. matrices and 
satisfies~\eqref{eq:operator:sum:representation:condition:2},
there must exist a unique CC-QSC $\{\operator{N}^{y|x}\}_{x,y}$ acting
on $\hilbert_\system{S}$ such that~\eqref{eq:qsc:matrix:identification} holds.
\end{proposition}
\begin{proof}
The idea of the proof is to consider the eigenvalue decomposition of
$\llbracket W^{y|x}\rrbracket$, and reconstruct $\operator{N}^{y|x}$ by
following the equations~\eqref{eq:def:channel:function:representation}
and~\eqref{eq:Kraus:quantum:state:channel} backwardly.
We omit the details here.
\end{proof}
\par 
Let us conclude this section by pointing out that the functions 
$\{W^{y|x}\}_{x,y}$, particularly the corresponding NFG, can be constructed
from the channel-ensemble-measurement configuration ($\operator{N}$,
$\{\rho^{(x)}_\system{A}\}_{x\in\set{X}}$,
$\{\Lambda_\system{B}^{(y)}\}_{y\in\set{Y}}$) as in Fig.~\ref{fig:CFR}.
This can be justified by checking~\eqref{eq:def:qsc}
and~\eqref{eq:qsc:matrix:identification}.
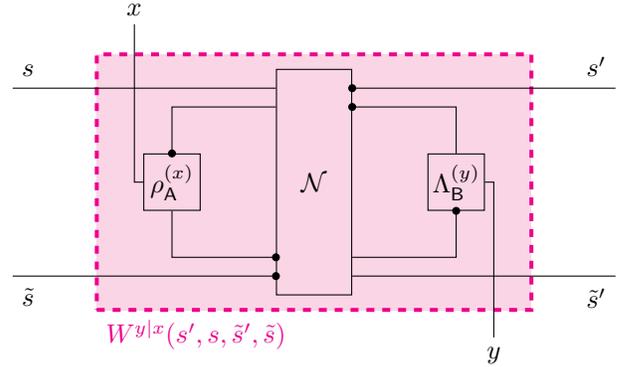
\begin{figure}
  \centering
  
    \begin{tikzpicture}[
    node/.style={draw=none},
    factor/.style={rectangle, minimum width=1cm, minimum height=3cm, draw},
    sfactor/.style={rectangle, minimum size=.75cm, draw, inner sep=1pt}]
\node[sfactor] (rhoX) {};
\node at (rhoX) {$\rho_\system{A}^{(x)}$};
\node[factor, right=1cm of rhoX] (N) {$\operator{N}$};
\draw[*-] (rhoX.north) |- ([yshift=1cm]N.west);
\draw[-*] (rhoX.south) |- ([yshift=-1cm]N.west);
\node at ([xshift=-.5cm,yshift=2.3cm]rhoX) (X) {$x$};
\draw (X) |- (rhoX);
    
\draw ([xshift=-3.50cm,yshift=1.25cm]N.west) -- ([yshift=1.25cm]N.west)
    node[above=1pt,pos=0,anchor=south west] {$s$};
\draw[-*] ([xshift=-3.50cm,yshift=-1.25cm]N.west) -- ([yshift=-1.25cm]N.west)
    node[below=1pt,pos=0,anchor=north west] {$\tilde{s}$};
\draw[-*] ([xshift=3.50cm,yshift=1.25cm]N.east) -- ([yshift=1.25cm]N.east)
    node[above=1pt,pos=0,anchor=south east] {$s'$};
\draw ([xshift=3.50cm,yshift=-1.25cm]N.east) -- ([yshift=-1.25cm]N.east)
    node[below=1pt,pos=0,anchor=north east] {$\tilde{s}'$};

\node[sfactor, right=1cm of N] (M) {};
\node at (M) {$\Lambda^{(y)}_{\system{B}}$};
\draw[-*] (M.north) |- ([yshift=1cm]N.east);
\draw[*-] (M.south) |- ([yshift=-1cm]N.east);
\node at ([xshift=.5cm,yshift=-2.3cm]M) (Y) {$y$};
\draw (Y) |- (M);
    
\begin{pgfonlayer}{bg}
    \draw[dashed, magenta, line width=1.5pt, fill= magenta!20]
        ([xshift=-1.0cm,yshift=1.7cm]rhoX) rectangle
        ([xshift=1cm,yshift=-1.7cm]M);
        \node[magenta,anchor=north west] at ([xshift=-1.0cm,yshift=-1.7cm]rhoX)
        {$W^{y|x}(s',s,\tilde{s}',\tilde{s})$};
\end{pgfonlayer}
\end{tikzpicture}
  \caption{NFG representation of the channel-ensemble-measurement 
    configuration ($\operator{N}$, $\{\rho^{(x)}_\system{A}\}_{x\in\set{X}}$,
    $\{\Lambda_\system{B}^{(y)}\}_{y\in\set{Y}}$).}
  \label{fig:CFR}
\end{figure}
\section{Information Rate and its Estimation}\label{sec:4:IR}
In this section, we focus on the information rate of the communication scheme
described in Section~\ref{sec:3:QCM}. As defined in~\eqref{eq:capacity:jointk},
the information rate is the limit superior of the average mutual information
$\frac{1}{n}\mutualInfo\left(\rv{X}_1^n;\rv{Y}_1^n\right)$ between the input
and output processes $\rv{X}_1^n$ and $\rv{Y}_1^n$ as $n$
tends to infinity. We assume that $\rv{X}_1^n$ is distributed
according to some i.i.d. process\footnote{For more general type
of sources, like a finite-state-machine source (FSMS), one can consider
``merging'' the memory of the source into that of the channel, and thus
obtaining an equivalent memoryless input process.
} characterized by the pmf $Q$, \ie,
$Q^{(n)}(\vx_1^n) = \prod_{\ell=1}^{n} Q(x_{\ell})$. In this case, 
the joint distribution of $(\rv{X}_1^n,\rv{Y}_1^n)$ is given by
\begin{equation}\label{eq:joint:5}\hspace{0pt}
P_{\rv{X}_1^n,\rv{Y}_1^n|\system{S}_0}(\vx_1^n,\vy_1^n|\rho_{\system{S}_0}\!)
\!=\! \prod_{\ell=1}^n Q(x_\ell) \!\cdot\!
  P_{\rv{Y}_1^n|\rv{X}_1^n;\system{S}_0}(\vy_1^n|\vx_1^n;\rho_{\system{S}_0}\!),
\hspace{-8pt}
\end{equation}
where $P_{\rv{Y}_1^n|\rv{X}_1^n;\system{S}_0}$ is specified
in~\eqref{eq:joint:1},~\eqref{eq:joint:2},~\eqref{eq:joint:3}
or~\eqref{eq:joint:4}, depending on which notation we use to specify the channel
(see Propositions~\ref{prop:quantum:state:channel} and~\ref{prop:CFR}).
It is obvious that the value of~\eqref{eq:joint:5}, and thus the information
rate, depends on the initial density operator $\rho_{\system{S}_0}$. In this
sense, we denote the information rate as a function of the input pmf $Q$,
the CC-QSC $\{\operator{N}^{y|x}\}_{x,y}$ describing the channel, and the
initial density operator $\rho_{\system{S}_0}$, namely
\begin{align}
\label{eq:def:qsc:IR:2}
&\infoRate(Q,\{\operator{N}^{y|x}\}_{x,y},\rho_{\system{S}_0})
\defeq \limsup_{n\rightarrow\infty}
    \infoRate^{(n)}(Q,\{\operator{N}^{y|x}\}_{x,y},\rho_{\system{S}_0}),
    \hspace{-5pt}\\
\label{eq:def:qsc:IR:1}
&\infoRate^{(n)}(Q,\{\operator{N}^{y|x}\}_{x,y},\rho_{\system{S}_0})
\defeq \frac{1}{n}\mutualInfo(\rv{X}_1^n;\rv{Y}_1^n)(\rho_{\system{S}_0}).
\end{align}
Here, $\mutualInfo(\rv{X}_1^n;\rv{Y}_1^n)(\rho_{\system{S}_0})$ is the mutual 
information between $\rv{X}_1^n$ and $\rv{Y}_1^n$; and the latter are jointly 
distributed according to~\eqref{eq:joint:5}. 
The argument $\rho_{\system{S}_0}$ emphasizes the dependency of 
$\frac{1}{n}\mutualInfo(\rv{X}_1^n;\rv{Y}_1^n)$ on $\rho_{\system{S}_0}$.
\par 
Similar to the case of an FSMC, the dependency of the information rate on the
initial density operator usually cannot be ignored. However, as already
mentioned in~Section~\ref{sec:FSMC:IR}, for a class of FSMCs, namely the
indecomposable FSMCs, it is known that the information rate is independent
from the initial channel state~\cite{gallager1968information}.
An indecomposable FSMC, intuitively speaking, is an FSMC whose state
distribution, given different initial states, tends to be indistinguishable as
$n\rightarrow\infty$, independently of the input sequence realized.
A quantum analogy was proposed by Bowen, Devetak, and
Mancini~\cite{bowen2005bounds}, where they defined the indecomposable quantum
channels with memory, and proved that the quantum entropic bound for such
channels is independent from the initial density operator.
\par 
In the remainder of this section we firstly define the indecomposability 
of CC-QSCs, and prove the independence of the information rate as
in~\eqref{eq:def:qsc:IR:2} from the initial density
operator. Secondly, we generalize the methods in Algorithm~\ref{alg:SPA}
for estimating such information rates efficiently.
\par 
The definition of an indecomposable CC-QSC in our paper is similar (but 
different) and closely related to that of an indecomposable (quantum)
channel with memory in~\cite{bowen2005bounds}. Namely, an indecomposable
channel with memory equipped with separable input ensemble and local output
measurement will always induce an indecomposable CC-QSC, but not necessarily
vice versa. Moreover, in~\cite{bowen2005bounds} the classical capacity of 
quantum channels with finite memory was considered, where the capacity is
essentially the Holevo bound, and where the latter was proven to be
achievable~\cite{bowen2004quantum}. However, in our work, we focus on the
situation where the ensemble and the measurement are fixed.
\subsection{Indecomposable Quantum-State Channel}
\begin{definition}A CC-QSC $\{\operator{N}^{y|x}\}_{x,y}$ is said to be
\emph{indecomposable} if for any initial density operators
$\alpha_{\system{S}_0}$ and $\beta_{\system{S}_0}$, the following statement
holds: for any $\epsilon>0$, there exists some positive integer $N$ s.t.
\begin{equation}
\norm{\alpha_{\system{S}_n}^{(\vx_1^n)} - \beta_{\system{S}_n}^{(\vx_1^n)}}_1
< \varepsilon \quad \forall n\geqslant N,\ \forall \vx_1^n\in\set{X}^n,
\end{equation}
where
\begin{align}
\alpha_{\system{S}_n}^{(\vx_1^n)} &\defeq 
    \sum_{\vy_1^n} \operator{N}^{y_n|x_n} \circ \cdots \circ 
    \operator{N}^{y_1|x_1}(\alpha_{\system{S}_0}),\\
\beta_{\system{S}_n}^{(\vx_1^n)} &\defeq 
    \sum_{\vy_1^n} \operator{N}^{y_n|x_n} \circ \cdots \circ
    \operator{N}^{y_1|x_1}(\beta_{\system{S}_0}),
\end{align}
and where $\norm{A}_1$ is the trace distance for an operator $A$ on 
$\hilbert_\system{S}$, \ie, $\norm{A}_1\defeq\frac{1}{2}\tr\sqrt{A^\Herm A}$.
\end{definition}
\begin{theorem}\label{thm:quantum:indecom}\hspace{-5pt}\footnote{
    A similar result regarding indecomposable/forgetful quantum channel
    with memory can be found in~\cite{kretschmann2005quantum}
    and~\cite{bowen2005bounds}.}
The information rate of an indecomposable CC-QSC with an i.i.d. input process
is independent of the initial density operator. Namely, if 
$\{\operator{N}^{y|x}\}_{x,y}$ is indecomposable, then
\[
  \infoRate^{(n)}(Q,\{\operator{N}^{y|x}\}_{x,y},\alpha_{\system{S}_0})
- \infoRate^{(n)}(Q,\{\operator{N}^{y|x}\}_{x,y},\beta_{\system{S}_0})
\stackrel{n\rightarrow\infty}{\longrightarrow} 0,
\]
for any initial density operators
$\alpha_{\system{S}_0},\ \beta_{\system{S}_0}
\in\DensOp(\hilbert_{\system{S}_0})$.
\end{theorem}
In the proof below, we follow a similar idea as
in~\cite{gallager1968information} for indecomposable FSMCs, and as that
in~\cite{bowen2005bounds} for indecomposable quantum channels with memory.
\begin{proof}
Let $\system{A}$ and $\system{B}$ be quantum systems described by
Hilbert spaces $\hilbert_\system{A}$ and $\hilbert_\system{B}$, respectively,
where $\{\bra{x}\}_{x\in\set{X}}$ and $\{\bra{y}\}_{y\in\set{Y}}$ are
orthonormal bases of $\hilbert_\system{A}$ and $\hilbert_\system{B}$,
respectively.
Let $\system{A}_1^n$ and $\system{B}_1^n$ be $n$ copies of
$\system{A}$ and $\system{B}$, respectively. Let $\rho_{\system{S}_0}$ be some
initial density operator; and let the joint density operator on system 
$\system{A}_1^n\system{B}_1^n$ be 
\[
\rho_{\system{A}_1^n\system{B}_1^n} \!\defeq\!
\sum_{\vx_1^n} \!Q(\vx_1^n)\cdot\braket{\vx_1^n} \tensor
\sum_{\vy_1^n} \!\tr\!\left(\!
    \operator{N}^{\vy_1^n|\vx_1^n}(\rho_{\system{S}_0})\!\right)
    \cdot \braket{\vy_1^n},
\]
where $\operator{N}^{\vy_1^n|\vx_1^n}\defeq\operator{N}^{y_n|x_n}\circ
\cdots\circ\operator{N}^{y_1|x_1}$. In this case, it is not hard to see that
\[
\mutualInfo(\rv{X}_1^n;\rv{Y}_1^n)[\rho_{\system{S}_0}] = 
    \qmutualInfo(\system{A}_1^n;\system{B}_1^n)[\rho_{\system{S}_0}].
\]
In fact, one can easily check that
\begin{align*}
\qEntropy(\system{A}_1^n) &= \entropy(\rv{X}_1^n),\\
\qEntropy(\system{B}_1^n) &= \entropy(\rv{Y}_1^n),\\
\qEntropy(\system{A}_1^n,\system{B}_1^n) &= \entropy(\rv{X}_1^n,\rv{Y}_1^n).
\end{align*}
In particular, $\qEntropy(\system{A}_1^n)$ is independent of the initial
density operator $\rho_{\system{S}_0}$. We also claim that, for each 
$\rho_{\system{S}_0}\in\DensOp(\hilbert_{\system{S}_0})$ and 
positive integer $N < n$, 
\begin{align}
\label{eq:indecomposable:AB}
\qmutualInfo(\system{A}_1^N\system{B}_1^N;\system{A}_{N+1}^n\system{B}_{N+1}^n) 
    &\leqslant 2 \qEntropy(\system{S}_N),\\
\label{eq:indecomposable:B}
\qmutualInfo(\system{B}_1^N;\system{B}_{N+1}^n)
    &\leqslant 2 \qEntropy(\system{S}_N),
\end{align}
where the density operator for $\system{S}_N$ is defined as
(depending on $\rho_{\system{S}_0}$)
\[
\rho_{\system{S}_N} \defeq \sum_{\vx_1^N} Q(\vx_1^N)\cdot
\sum_{\vy_1^N} \operator{N}^{\vy_1^N|\vx_1^N}(\rho_{\system{S}_0}).
\]
\par
\textbf{Proof of~\eqref{eq:indecomposable:AB}: }
We define a class of CPTP maps 
$\{\Phi_a^b: \DensOp(\hilbert_{\system{S}_a})\rightarrow
\DensOp(\hilbert_{\system{A}_a^b \system{B}_a^b})\}_{a<b\in\mathbb{N}}$ as
\[
\Phi_a^b\!\!: \rho_{\system{S}_a} \!\!\mapsto\!\!
    \sum_{\vx_a^b}\! Q(\vx_a^b)\cdot \bra{\vx_a^b}\hspace{-3pt}\ket{\vx_a^b}
    \tensor\sum_{\vy_a^b} \!\tr\!\left(\!
        \operator{N}^{\vy_a^b|\vx_a^b}(\rho_{\system{S}_a})\!\right)
    \cdot \bra{\vy_a^b}\hspace{-3pt}\ket{\vy_a^b}\!.
\]
Since the input process $Q$ is i.i.d., we can rewrite
$\rho_{\system{A}_1^n\system{B}_1^n}$, for each positive integer $N<n$, as
\[
\rho_{\system{A}_1^n\system{B}_1^n} = 
\left(I_{\system{A}_1^N\system{B}_1^N}\otimes\Phi_{N+1}^n\right)
\left(\rho_{\system{A}_1^N\system{B}_1^N\system{S}_N}\right),
\]
where
\[
\rho_{\system{A}_1^N\!\system{B}_1^N\!\system{S}_N} \!\!\!\defeq\!
\sum_{\vx_1^N}\! Q(\vx_1^N) 
\bra{\vx_1^N\!}\hspace{-3pt}\ket{\vx_1^N\!}
\tensor
\sum_{\vy_1^N} \operator{N}^{\vy_1^N|\vx_1^N}(\rho_{\system{S}_0}\!)
               \bra{\vy_1^N\!}\hspace{-3pt}\ket{\vy_1^N\!}\!.
\]
Hence, by data processing inequality for quantum mutual
information (see \eg,~\cite[Theorem~11.9.4]{wilde2017quantum}),
one must have 
\[
\qmutualInfo(\system{A}_1^N\system{B}_1^N;\system{A}_{N+1}^n\system{B}_{N+1}^n)
\leqslant \qmutualInfo(\system{A}_1^N\system{B}_1^N;\system{S}_N).
\]
Additionally, by subadditivity of joint entropy, we have
\begin{align*}
\qmutualInfo(\system{A}_1^N\system{B}_1^N;\system{S}_N)
\defeq \qEntropy(\system{A}_1^N\system{B}_1^N) + \qEntropy(\system{S}_N) -
        \qEntropy(\system{A}_1^N\system{B}_1^N\system{S}_N)\\
\begin{aligned}
&\leqslant \qEntropy(\system{A}_1^N\system{B}_1^N) \!+\! 
           \qEntropy(\system{S}_N) \!-\!
           \abs{\qEntropy(\system{A}_1^N\system{B}_1^N) \!-\!
                \qEntropy(\system{S}_N)}\\
&\leqslant 2\qEntropy(\system{S}_N).
\end{aligned}
\end{align*}
Combining the above two inequalities, we have
proven~\eqref{eq:indecomposable:AB}.
\par
\textbf{Proof of~\eqref{eq:indecomposable:B}: }
We follow the same approach as above for proving~\eqref{eq:indecomposable:AB}
by defining another class of CPTP maps
$\{\Psi_a^b: \DensOp(\hilbert_{\system{S}_a})\rightarrow
\DensOp(\hilbert_{\system{B}_a^b})\}_{a<b\in\mathbb{N}}$ as
\[
\Psi_a^b: \rho_{\system{S}_a} \mapsto 
\sum_{\vx_a^b} Q(\vx_a^b)\cdot 
\sum_{\vy_a^b} \tr\left(
    \operator{N}^{\vy_a^b|\vx_a^b}(\rho_{\system{S}_a})\right)
\cdot \braket{\vy_a^b}.
\]
We omit the detailed derivation of~\eqref{eq:indecomposable:B}.
\par
We now return to the main proof. 
Given the initial density operators $\alpha_{\system{S}_0}$, and 
$\beta_{\system{S}_0}$, we define $\alpha_{\system{A}_1^n\system{B}_1^n}$,
$\beta_{\system{A}_1^n\system{B}_1^n}$ and $\alpha_{\system{S}_N}$,
$\beta_{\system{S}_N}$ in a similar fashion as we have defined 
$\rho_{\system{A}_1^n\system{B}_1^n}$ and $\rho_{\system{S}_N}$ based on
$\rho_{\system{S}_0}$. In this case, one obtains 
\begin{align}
\nonumber
&\bigabs{\qEntropy(\alpha_{\system{A}_1^n\system{B}_1^n})
\!-\!\qEntropy(\beta_{\system{A}_1^n\system{B}_1^n})}
\!-\!\bigabs{
\qEntropy(\alpha_{\system{A}_{N\!+\!1}^n\system{B}_{N\!+\!1}^n})
\!-\!\qEntropy(\beta_{\system{A}_{N\!+\!1}^n\system{B}_{N\!+\!1}^n})}\\
\nonumber
&\overset{(\mathrm{a})}{\leqslant}\ 
\begin{aligned}[t]
&\bigabs{\qEntropy(\alpha_{\system{A}_1^N\system{B}_1^N})
       -\qEntropy(\beta_{\system{A}_1^N\system{B}_1^N})} + \\
&\bigabs{\qmutualInfo(\system{A}_1^N\system{B}_1^N;
    \system{A}_{N\!+\!1}^n\system{B}_{N\!+\!1}^n)
    [\alpha_{\system{A}_1^n\system{B}_1^n}]\\
&\hspace{90pt}-\qmutualInfo(\system{A}_1^N\system{B}_1^N;
    \system{A}_{N\!+\!1}^n\system{B}_{N\!+\!1}^n)
    [\beta_{\system{A}_1^n\system{B}_1^n}]}
\end{aligned}\\
\label{eq:tail:1}
&\overset{(\mathrm{b})}{\leqslant}\ 
N\cdot\log(\dim{\hilbert_\system{AB}}) +
    2\cdot\max\left\{\qEntropy(\alpha_{\system{S}_N}),
    \qEntropy(\beta_{\system{S}_N})\right\},
\end{align}
where we have used the triangle inequality in step (a), and a basic property
of von Neumann entropy~\cite[Theorem~11.8]{nielsen2011quantum}
and~\eqref{eq:indecomposable:AB} in step (b).
Similarly, using~\eqref{eq:indecomposable:B}, one can prove
\begin{equation}\label{eq:tail:2}
\begin{aligned}
&\bigabs{\qEntropy(\alpha_{\system{B}_1^n})
    -\qEntropy(\beta_{\system{B}_1^n})}
-\bigabs{\qEntropy(\alpha_{\system{B}_{N+1}^n})
    -\qEntropy(\beta_{\system{B}_{N+1}^n})}\\
&\leqslant N\cdot\log(\dim{\hilbert_\system{B}}) +
    2\cdot\max\left\{\qEntropy(\alpha_{\system{S}_N}),
        \qEntropy(\beta_{\system{S}_N})\right\}.
\end{aligned}
\end{equation}
By assumption, there exists some positive integer $d$ such that
$\max\{\dim{\hilbert_\system{A}},\dim{\hilbert_\system{B}},
\dim{\hilbert_\system{S}}\} \leqslant d$. Thus, we have
\begin{align*}
&\frac{1}{n}\bigabs{
    \mutualInfo(\rv{X}_1^n;\rv{Y}_1^n)[\alpha_{\system{S}_0}]
   -\mutualInfo(\rv{X}_1^n;\rv{Y}_1^n)[\beta_{\system{S}_0}]}\\
=\ &\frac{1}{n}\bigabs{
    \qmutualInfo(\system{A}_1^n;\system{B}_1^n)[\alpha_{\system{S}_0}]
   -\qmutualInfo(\system{A}_1^n;\system{B}_1^n)[\beta_{\system{S}_0}]}\\
=\ &\frac{1}{n}\bigabs{
    \left(\qEntropy(\alpha_{\system{B}_1^n})
         -\qEntropy(\alpha_{\system{A}_1^n\system{B}_1^n})\right)
   -\left(\qEntropy(\beta_{\system{B}_1^n})
         -\qEntropy(\beta_{\system{A}_1^n\system{B}_1^n})\right)}\\
\overset{(\mathrm{c})}{\leqslant}\ & 
    \frac{1}{n} \bigabs{\qEntropy(\alpha_{\system{B}_1^n})
                       -\qEntropy(\beta_{\system{B}_1^n})}
   +\frac{1}{n}\bigabs{\qEntropy(\alpha_{\system{A}_1^n\system{B}_1^n})
                      -\qEntropy(\beta_{\system{A}_1^n\system{B}_1^n})}\\
\overset{(\mathrm{d})}{\leqslant}\ &
    \frac{3N+4}{n}\cdot\log{d}
    \begin{aligned}[t]
     &+\frac{1}{n}\bigabs{\qEntropy(\alpha_{\system{B}_{N+1}^n})
                      -\qEntropy(\beta_{\system{B}_{N+1}^n})}\\
     &+\frac{1}{n}\bigabs{\qEntropy(\alpha_{\system{A}_{N+1}^n\system{B}_{N+1}^n})
                      -\qEntropy(\beta_{\system{A}_{N+1}^n\system{B}_{N+1}^n})}
    \end{aligned}\\
=\ & \frac{3N+4}{n}\cdot\log{d}
    \begin{aligned}[t]
     &+\frac{1}{n}\bigabs{\qEntropy(\Psi_{N+1}^n(\alpha_{\system{S}_N}))
                      -\qEntropy(\Psi_{N+1}^n(\beta_{\system{S}_N}))}\\
     &+\frac{1}{n}\bigabs{\qEntropy(\Phi_{N+1}^n(\alpha_{\system{S}_N}))
                      -\qEntropy(\Phi_{N+1}^n(\beta_{\system{S}_N}))},
    \end{aligned}
\end{align*}
where we have used the triangle inequality in step (c), 
and~\cite[Theorem~11.8]{nielsen2011quantum},~\eqref{eq:tail:1},~\eqref{eq:tail:2} in step (d).
Using a loose variant of Fannes' inequality~\cite{fannes1973continuity}\footnote{
Namely, we used the inequality $\abs{\qEntropy(\rho)-\qEntropy(\sigma)}\leqslant
\log{\dim}\cdot\norm{\rho-\sigma}_1+e^{-1}$. Note that tighter variants of
Fannes' inequality exist, but the above inequality is good enough to prove
the desired result.}, we have
\begin{align*}
&\begin{aligned}
\bigabs{\qEntropy(\Psi_{N+1}^n(\alpha_{\system{S}_N}))
       -\qEntropy(\Psi_{N+1}^n(\beta_{\system{S}_N}))} \leqslant
    (n-N)\cdot\log{d}\ \cdot\\
    \norm{\Psi_{N+1}^n(\alpha_{\system{S}_N})
                                -\Psi_{N+1}^n(\beta_{\system{S}_N})}_1
    +e^{-1},
\end{aligned}\\
&\begin{aligned}
\bigabs{\qEntropy(\Phi_{N+1}^n(\alpha_{\system{S}_N}))
       -\qEntropy(\Phi_{N+1}^n(\beta_{\system{S}_N}))} \leqslant
    2\cdot(n-N)\cdot\log{d}\ \cdot\\
    \norm{\Phi_{N+1}^n(\alpha_{\system{S}_N})
                                      -\Phi_{N+1}^n(\beta_{\system{S}_N})}_1
   +e^{-1}.
\end{aligned}
\end{align*}
Moreover, by the contractivity of the trace distance, we have,
\begin{align*}
\norm{\Psi_{N+1}^n(\alpha_{\system{S}_N})-\Psi_{N+1}^n(\beta_{\system{S}_N})}_1
&\leqslant \norm{\alpha_{\system{S}_N}-\beta_{\system{S}_N}}_1,\\
\norm{\Phi_{N+1}^n(\alpha_{\system{S}_N})-\Phi_{N+1}^n(\beta_{\system{S}_N})}_1
&\leqslant \norm{\alpha_{\system{S}_N}-\beta_{\system{S}_N}}_1.
\end{align*}
This allows us to bound the difference of the information rates by
\[
\begin{aligned}
\frac{1}{n}\bigabs{\mutualInfo(\rv{X}_1^n;\rv{Y}_1^n)[\alpha_{\system{S}_0}]
                  -\mutualInfo(\rv{X}_1^n;\rv{Y}_1^n)[\beta_{\system{S}_0}]}
\leqslant \frac{3N+4}{n}\cdot\log{d}\\
+\frac{3(n-N)}{n}\cdot\log{d}
    \cdot\norm{\alpha_{\system{S}_n}-\beta_{\system{S}_n}}_1
+\frac{2}{n\cdot e}.
\end{aligned}
\]
Finally, because the CC-QSC is indecomposable, for any $\varepsilon>0$, 
we can choose $N$ large enough such that
\[
\norm{\alpha_{\system{S}_N}-\beta_{\system{S}_N}}_1 < 
\frac{\varepsilon}{6\cdot\log{d}},
\]
and then choose an integer $M>N$ such that
\[
\frac{3N+4}{M}\cdot\log{d} + \frac{2}{M\cdot e}< \frac{\varepsilon}{2}.
\]
This will ensure that for any $n>M$, we have
\[
\frac{3N+4}{n}\cdot\log{d}
         +\frac{3(n-N)}{n}\cdot\log{d}
          \cdot\norm{\alpha_{\system{S}_n}-\beta_{\system{S}_n}}_1
         +\frac{2}{n\cdot e}
< \varepsilon,
\]
which concludes the proof of the theorem.
\end{proof}
\subsection{Estimation of the Information Rate}
The development in this section is very similar to the development in
Section~\ref{sec:FSMC:IR}. In particular, we follow the same approach as in 
Eqs.~\eqref{eq:def:fsmc:ir:2}--\eqref{eq:FSMC:ir:estimate:1}. This similarity
stems from the similarity of the NFGs in Figs.~\ref{fig:FMSC:high:level:1}
and~\ref{fig:qsc:bcjr}, and highlights one of the benefits of the
factor-graph approach that we take to estimate information rates of
quantum channels with memory.
\par 
We make the following assumptions.
\begin{itemize}
\item As already mentioned, the derivations in this paper are for the case
  where the input process $\rv{X}_1^n=(\rv{X}_1,\ldots,\rv{X}_n)$ is an i.i.d.
  process. The results can be generalized to other stationary ergodic input
  processes that can be represented by a finite-state-machine source
  (FSMS). Technically, this is done by defining a new state that combines the
  FSMS state and the channel state.
\item We assume that the corresponding quantum-state channel 
  $\{\operator{N}^{y|x}\}_{x\in\set{X},y\in\set{Y}}$ is finite-dimensional
  and indecomposable. We also assume it can be represented by some functions
  $\{W^{y|x}\}_{x,y}$ as defined
  in~\eqref{eq:def:channel:function:representation}.
\end{itemize}
The major difference compared with Section~\ref{sec:FSMC:IR} is the conditional
pmf $P_{\rv{Y}_1^n|\rv{X}_1^n;\system{S}_0}$, and thus the joint 
pmf $P_{\rv{Y}_1^n,\rv{X}_1^n|\system{S}_0}$ as specified
in~\eqref{eq:joint:4} and~\eqref{eq:joint:5}, respectively.
In this case, in order to compute $-\frac{1}{n}\log{P_{\rv{Y}_1^n}(\cvy_1^n)}$
and $-\frac{1}{n}\log{P_{\rv{X}_1^n\rv{Y}_1^n}(\cvx_1^n,\cvy_1^n)}$ using
a similar method as in Section~\ref{sec:FSMC:IR}, we consider the state metrics
$\{\sigmaY_\ell\}_{\ell=1}^n$ and $\{\sigmaXY_\ell\}_{\ell=1}^n$
(which are operators on $\hilbert_{\system{S}_\ell}$ for each $\ell$) 
defined w.r.t. $\cvy_1^n$ and w.r.t. $\cvx_1^n$ and $\cvy_1^n$, respectively, as
\begin{align}
\sigmaY_\ell &\defeq\sum_{\vx_1^\ell} Q^{(\ell)}(\vx_1^\ell)\cdot 
\operator{N}^{\cy_n|x_n}\circ\cdots\circ\operator{N}^{\cy_1|x_1}
(\rho_{\system{S}_0}),\\
\sigmaXY_{\ell} &\defeq 
\operator{N}^{\cy_n|\cx_n}\circ\cdots\circ\operator{N}^{\cy_1|\cx_1}
(\rho_{\system{S}_0}).
\end{align}
In this case, we have
$P_{\rv{Y}_1^n}(\cvy_1^n) = \tr(\sigmaY_n)$, and
$P_{\rv{X}_1^n\rv{Y}_1^n}(\cvx_1^n,\cvy_1^n) = \tr(\sigmaXY_n)$.
Notice that $\{\sigmaY_\ell\}_{\ell}$ and 
$\{\sigmaXY_\ell\}_{\ell}$ can be computed iteratively as
\begin{align}
\label{eq:recursive:quantum:state:metric:Y:1}
[\sigmaY_{\ell}] &= \sum_{x_\ell} Q(x_\ell)\cdot 
    [W^{y|x}] \cdot [\sigmaY_{\ell-1}],\\
\label{eq:def:quantum:channel:state:metric:XY:1}
[\sigmaXY_{\ell}] &= [W^{y|x}] \cdot [\sigmaXY_{\ell-1}],
\end{align}
where we treat $[\sigma^\rv{Y}_{\ell}]$ and $[\sigma^\rv{XY}_{\ell}]$
as length-$d^2$ vectors indexed by $(s,\tilde{s})\in\set{S}^2$ in the
above two equations. (See~\eqref{eq:qsc:matrix:1} 
and~\eqref{eq:joint:4} for notations.) Moreover, we can also 
introduce normalizing coefficients $\{\lambda_\ell^\rv{Y}\}_{\ell}$ and
$\{\lambda_\ell^\rv{XY}\}_{\ell}$, similar 
to~\eqref{eq:recursive:state:metric:Y:2}, for the sake of numerical
stability. In the latter case, we have iterative updating rules
\begin{align}
\label{eq:recursive:quantum:state:metric:Y:2}
[\bsigmaY_\ell] &= \frac{1}{\lambdaY_\ell}\cdot\sum_{x_\ell} Q(x_\ell)\cdot 
    [W^{y|x}] \cdot [\bsigmaY_{\ell-1}],\\
\label{eq:def:quantum:channel:state:metric:XY:2}
[\bsigmaXY_\ell] &= \frac{1}{\lambdaXY_\ell}\cdot [W^{y|x}] 
    \cdot [\bsigmaXY_{\ell-1}],
\end{align}
where the scaling factors $\lambdaY_\ell>0$ and $\lambdaXY_\ell>0$ are chosen
such that $\tr(\bsigmaY_\ell)=1$ and $\tr(\bsigmaXY_\ell)=1$, respectively.
In addition, one can verify that 
$P_{\rv{Y}_1^n}(\cvy_1^n) = \prod_{\ell=1}^n \lambdaY_\ell$, and
$P_{\rv{X}_1^n\rv{Y}_1^n}(\cvx_1^n,\cvy_1^n) = \prod_{\ell=1}^n\lambdaXY_\ell$.
\par 
The above discussion is summarized as Algorithm~\ref{alg:SPA:2}.
The computations corresponding to Line 3, 5--9 and 12--16 are
visualized in
Figs.~\ref{fig:QFSM:channel:simulation:Y},~\ref{fig:QFSM:estimate:hY},
and~\ref{fig:QFSM:estimate:hXY} in the Appendix, respectively.
\begin{algorithm}[t]
\caption{Estimating the information rate of a CC-QSC}
\begin{algorithmic}[1] 
\Require{indecomposable CC-QSC 
         $\{\operator{N}^{y|x}\}_{x\in\set{X},y\in\set{Y}}$, which
         can be represented by functions $\{W^{y|x}\}_{x,y}$,
         input~distribution~$Q$,
         positive~integer~$n$ large enough.}
\Ensure{$\infoRate^{(n)}(Q,\{\operator{N}^{y|x}\}_{x,y}) \approx
         \entropy(\rv{X})+\hat\entropicRate(\rv{Y})
         -\hat\entropicRate(\rv{X},\rv{Y})$.}
\State Initialize the memory density operator
       $\rho_{\system{S}_0}\gets\braket{0_\system{S}}$
\State Generate an input sequence $\cvx_1^n \sim Q^{\tensor n}$
\State Generate a corresponding output sequence $\cvy_1^n$
\State $\bsigmaY_0\gets\rho_{\system{S}_0}$
\ForEach{$\ell=1,\ldots,n$}
\State $[\sigmaY_\ell] \gets \sum_{x_\ell} Q(x_\ell)\cdot 
       [W^{\cy_{\ell}|x}] \cdot [\bsigmaY_{\ell-1}]$
\State $\lambdaY_\ell \gets \tr(\sigmaY_\ell)$
\State $\bsigmaY_{\ell} \gets \sigmaY_\ell/\lambdaY_\ell$
\EndFor
\State $\hat\entropicRate(\rv{Y})\gets
        -\frac{1}{n} \sum_{\ell=1}^n \log(\lambdaY_{\ell})$
\State $\bsigmaXY_0\gets\rho_{\system{S}_0}$
\ForEach{$\ell=1,\ldots,n$}
\State $[\sigmaXY_\ell] \gets [W^{\cy_\ell|\cx_\ell}] \cdot
        [\bsigmaXY_{\ell-1}]$
\State $\lambdaXY_\ell \gets \tr(\sigmaXY_\ell)$
\State $\bsigmaXY_{\ell} \gets \sigmaXY_\ell/\lambdaXY_\ell$
\EndFor
\State $\hat\entropicRate(\rv{X},\rv{Y})\gets
        -\frac{1}{n} \sum_{\ell=1}^n \log(\lambdaXY_{\ell})$
\State $\entropy(\rv{X}) \gets -\sum_{x} Q(x) \log{Q(x)}$
\State Estimate $\infoRate^{(n)}(Q,\{\operator{N}^{y|x}\}_{x,y})$ as 
       $\entropy(\rv{X}) +\hat\entropicRate(\rv{Y}) - 
       \hat\entropicRate(\rv{X},\rv{Y})$.
\end{algorithmic}
\label{alg:SPA:2}
\end{algorithm}
\section{Information rate upper/lower bounds\\and their Optimization}
\label{sec:5:UBLB}
\bigformulatop{77}{
\begin{align}
\label{eq:IRUB:explicit}
\IRUB^{(n)}_W(\hat{W}) &= \frac{1}{n}\left\langle
\log\frac{\tr\left([W^{\rv{Y}_n|\rv{X}_n}]\cdots[W^{\rv{Y}_1|\rv{X}_1}]
             \cdot[\rho_{\system{S}_0}]\right)}
         {\sum_{\vx_1^n} Q^{(n)}(\vx_1^n) \cdot
          \tr\left([\hat{W}^{\rv{Y}_n|x_n}]\cdots[\hat{W}^{\rv{Y}_1|x_1}]
             \cdot[\rho_{\hat{\system{S}}_0}]\right)}
\right\rangle_{\!\!\!\rv{X}_1^n\rv{Y}_1^n},\\
\label{eq:IRLB:explicit}
\IRLB_{W}^{(n)}(\hat{W}) &= \frac{1}{n}\left\langle
\log\frac{\tr\left([\hat{W}^{\rv{Y}_n|\rv{X}_n}]\cdots
         [\hat{W}^{\rv{Y}_1|\rv{X}_1}]\cdot[\rho_{\hat{\system{S}}_0}]\right)}
         {\sum_{\vx_1^n} Q^{(n)}(\vx_1^n)\cdot
          \tr\left([\hat{W}^{\rv{Y}_n|x_n}]\cdots[\hat{W}^{\rv{Y}_1|x_1}]
             \cdot[\rho_{\hat{\system{S}}_0}]\right)}
\right\rangle_{\!\!\!\rv{X}_1^n\rv{Y}_1^n},\\
\label{eq:DIFF:explicit}
\Delta_{W}^{(n)}(\hat{W}) &= \frac{1}{n}\left\langle
\log\frac{\tr\left([W^{\rv{Y}_n|\rv{X}_n}]\cdots[W^{\rv{Y}_1|\rv{X}_1}]
             \cdot[\rho_{\system{S}_0}]\right)}
         {\tr\left([\hat{W}^{\rv{Y}_n|\rv{X}_n}]\cdots
            [\hat{W}^{\rv{Y}_1|\rv{X}_1}]
            \cdot[\rho_{\hat{\system{S}}_0}]\right)}
\right\rangle_{\!\!\!\rv{X}_1^n\rv{Y}_1^n},
\end{align}
\begin{align}
\left.\frac{\D}{\D t}\right|_{t=0} \hspace{-15pt} \IRUB_{W}^{(n)}(\hat{W}+tH)
&\propto-\frac{1}{n}\left\langle
    \sum_{k=1}^n \sum_{\vx_1^n} Q^{(n)}(\vx_1^n)\cdot
    \tr\left([\hat{W}^{\rv{Y}_n|x_n}]\cdots
             [\hat{W}^{\rv{Y}_{k\!+\!1}|x_{k\!+\!1}}]
             [H^{\rv{Y}_k|x_k}]
             [\hat{W}^{\rv{Y}_{k\!-\!1}|x_{k\!-\!1}}]\cdots
             [\hat{W}^{\rv{Y}_1|x_1}]\cdot
             [\rho_{\hat{\system{S}}_0}]\right)
    \right\rangle_{\!\!\!\rv{Y}_1^n} \nonumber\\
\label{eq:IRUB:dev:2}
&=-\frac{1}{n}
\sum_{\vx_1^n,\vy_1^n} P_{\rv{X}_1^n,\rv{Y}_1^n|\system{S}_0}
                         (\vx_1^n,\vy_1^n|\rho_{\system{S}_0})
\cdot\sum_{k}\sum_{s',s,\tilde{s}',\tilde{s}}
    \rvec{\varrho}_{\hat{\system{S}}_{k\!-\!1}}^{(\vy_1^{k\!-\!1})}(s,\tilde{s})
    \cdot H^{y_k|x_k}(s',s,\tilde{s}',\tilde{s})\cdot
    \lvec{\varrho}_{\hat{\system{S}}_k}^{(\vy_{k\!+\!1}^n)}(s',\tilde{s}'),\\
\left.\frac{\D}{\D t}\right|_{t=0} \hspace{-15pt} \IRLB_{W}^{(n)}(\hat{W}+tH)
&\propto\!\!\begin{aligned}[t] 
    &-\frac{1}{n}\left\langle
    \sum_{k=1}^n \sum_{\vx_1^n} Q^{(n)}(\vx_1^n)\cdot
    \tr\left([\hat{W}^{\rv{Y}_n|x_n}]\cdots
             [\hat{W}^{\rv{Y}_{k\!+\!1}|x_{k\!+\!1}}]
             [H^{\rv{Y}_k|x_k}]
             [\hat{W}^{\rv{Y}_{k\!-\!1}|x_{k\!-\!1}}]
             \cdots[\hat{W}^{\rv{Y}_1|x_1}]\cdot
             [\rho_{\hat{\system{S}}_0}]\right)
    \right\rangle_{\!\!\!\rv{Y}_1^n}\\
    &+\frac{1}{n}\left\langle
    \sum_{k=1}^n \tr\left([\hat{W}^{\rv{Y}_n|\rv{X}_n}]\cdots
                          [\hat{W}^{\rv{Y}_{k\!+\!1}|\rv{X}_{k\!+\!1}}]
                          [H^{\rv{Y}_k|\rv{X}_k}]
                          [\hat{W}^{\rv{Y}_{k\!-\!1}|\rv{X}_{k\!-\!1}}]\cdots
                          [\hat{W}^{\rv{Y}_1|\rv{X}_1}]\cdot
                          [\rho_{\hat{\system{S}}_0}]\right)
    \right\rangle_{\!\!\!\rv{X}_1^n\rv{Y}_1^n}
  \end{aligned}\nonumber\\
\label{eq:IRLB:dev:2}
&=\!\!\begin{aligned}[t] 
    &-\frac{1}{n}
    \sum_{\vx_1^n,\vy_1^n} P_{\rv{X}_1^n,\rv{Y}_1^n|\system{S}_0}
                           (\vx_1^n,\vy_1^n|\rho_{\system{S}_0})
    \cdot\sum_{k}\sum_{s',s,\tilde{s}',\tilde{s}}
         \rvec{\varrho}_{\hat{\system{S}}_{k\!-\!1}}^{\,(\vy_1^{k\!-\!1})}
         (s,\tilde{s})\cdot H^{y_k|x_k}(s',s,\tilde{s}',\tilde{s})\cdot
         \lvec{\varrho}_{\hat{\system{S}}_k}^{\,(\vy_{k\!+\!1}^n)}
         (s',\tilde{s}')\\
    &+\frac{1}{n}
    \sum_{\vx_1^n,\vy_1^n} P_{\rv{X}_1^n,\rv{Y}_1^n|\system{S}_0}
                           (\vx_1^n,\vy_1^n|\rho_{\system{S}_0})
    \cdot\sum_{k}\sum_{s',s,\tilde{s}',\tilde{s}}
         \rvec{\varrho}_{\hat{\system{S}}_{k\!-\!1}}
         ^{(\vx_1^{k\!-\!1},\vy_1^{k\!-\!1})}(s,\tilde{s})\cdot
         H^{y_k|x_k}(s',s,\tilde{s}',\tilde{s})\cdot
         \lvec{\varrho}_{\hat{\system{S}}_k}^{(\vx_{k\!+\!1}^n,\vy_{k\!+\!1}^n)}
         (s'\!,\tilde{s}'),
  \end{aligned}\hspace{-24pt}\\
\left.\frac{\D}{\D t}\right|_{t=0} \hspace{-15pt} \Delta_{W}^{(n)}(\hat{W}+tH)
&\propto-\frac{1}{n}\left\langle
    \sum_{k=1}^n 
    \tr\left([\hat{W}^{\rv{Y}_n|\rv{X}_n}]\cdots
             [\hat{W}^{\rv{Y}_{k\!+\!1}|\rv{X}_{k\!+\!1}}]
             [H^{\rv{Y}_k|\rv{X}_k}]
             [\hat{W}^{\rv{Y}_{k\!-\!1}|\rv{X}_{k\!-\!1}}]\cdots
             [\hat{W}^{\rv{Y}_1|\rv{X}_1}]\cdot
             [\rho_{\hat{\system{S}}_0}]\right)
    \right\rangle_{\!\!\!\rv{X}_1^n\rv{Y}_1^n}\nonumber\\
\label{eq:DIFF:dev:2}
&\hspace{-26pt}
=-\frac{1}{n}\sum_{\vx_1^n,\vy_1^n} P_{\rv{X}_1^n,\rv{Y}_1^n|\system{S}_0}
                                     (\vx_1^n,\vy_1^n|\rho_{\system{S}_0})
    \cdot\sum_{k}\sum_{s',s,\tilde{s}',\tilde{s}}
         \rvec{\varrho}_{\hat{\system{S}}_{k\!-\!1}}
         ^{\,(\vx_1^{k\!-\!1},\vy_1^{k\!-\!1})}(s,\tilde{s})\cdot
         H^{y_k|x_k}(s',s,\tilde{s}',\tilde{s})\cdot
         \lvec{\varrho}_{\hat{\system{S}}_k}
         ^{\,(\vx_{k\!+\!1}^n,\vy_{k\!+\!1}^n)}(s',\tilde{s}'),
\end{align}
\begin{align}\label{eq:IRUB:grad:1}
\left(\grad\IRUB_{W,\mathrm{ext}}^{(n)}(\hat{W})\right)^{y|x} &\propto -\frac{1}{n}
    \left\langle\sum_{k=1}^{n}\delta_{\rv{X_k},x}\cdot\delta_{\rv{Y_k},y}\cdot
    \rvec{\varrho}_{\hat{\system{S}}_{k\!-\!1}}^{(\rv{Y}_1^{k\!-\!1})}
    \tensor\lvec{\varrho}_{\hat{\system{S}}_k}^{(\rv{Y}_{k\!+\!1}^n)}
    \right\rangle_{\!\!\!\rv{X}_1^n\rv{Y}_1^n},\\
\label{eq:IRLB:grad:1}
\left(\grad \IRLB_{W,\mathrm{ext}}^{(n)}(\hat{W})\right)^{y|x} &\propto-\frac{1}{n}
    \left\langle\sum_{k=1}^{n}\delta_{\rv{X_k},x}\cdot\delta_{\rv{Y_k},y}\cdot
    \left(\rvec{\varrho}_{\hat{\system{S}}_{k\!-\!1}}^{(\rv{Y}_1^{k\!-\!1})}
    \tensor\lvec{\varrho}_{\hat{\system{S}}_k}^{(\rv{Y}_{k\!+\!1}^n)}
    -
    \rvec{\varrho}_{\hat{\system{S}}_{k\!-\!1}}
        ^{(\rv{X}_1^{k\!-\!1},\rv{Y}_1^{k\!-\!1})}
    \tensor\lvec{\varrho}_{\hat{\system{S}}_k}
        ^{(\rv{X}_{k\!+\!1}^n,\rv{Y}_{k\!+\!1}^n)}\right)
    \right\rangle_{\!\!\!\rv{X}_1^n\rv{Y}_1^n},\\
\label{eq:DIFF:grad:1}
\left(\grad \Delta_{W,\mathrm{ext}}^{(n)}(\hat{W})\right)^{y|x} &\propto -\frac{1}{n}
    \left\langle\sum_{k=1}^{n}\delta_{\rv{X_k},x}\cdot\delta_{\rv{Y_k},y}\cdot
    \rvec{\varrho}_{\hat{\system{S}}_{k\!-\!1}}
        ^{(\rv{X}_1^{k\!-\!1},\rv{Y}_1^{k\!-\!1})}
    \tensor\lvec{\varrho}_{\hat{\system{S}}_k}
        ^{(\rv{X}_{k\!+\!1}^n,\rv{Y}_{k\!+\!1}^n)}
    \right\rangle_{\!\!\!\rv{X}_1^n\rv{Y}_1^n}.
\end{align}}
In this section, we consider auxiliary channels and their induced upper and
lower bounds on the information rate. As already mentioned in the introduction,
auxiliary channels are often introduced as a low-complexity approximation of
the original channel, which are useful in mismatch decoding.
The techniques developed in this section only require the channel input/output
data, but not the channel model itself. This is particularly useful when the channel is only made physically, but not mathematically, available.
In this case, the task of minimizing the difference between the upper and
lower bound is equivalent to finding the channel model (within a specified class 
of channel models) best fitting the \emph{empirical} channel law. Similarly,
minimizing the upper bound corresponds to finding the channel model best
fitting the \emph{empirical} channel output distribution, and maximizing the
lower bound corresponds to finding the channel model best fitting the
\emph{empirical} reverse channel law.
Motivated by the above scenarios, we particularly consider the
auxiliary channels chosen from the domain of all CC-QSCs with the same input and
output alphabet as the original channel, and acting on a memory system of a
certain dimension (which can be different from the memory dimension of the 
original channel). 
Throughout this section, we assume the original channel as described
in Section~\ref{sec:3:QCM} is indecomposable, and that all the involved Hilbert 
spaces are of finite dimension, and that the alphabets $\set{X}$ and
$\set{Y}$ are finite.
\par 
Suppose we have some auxiliary CC-QSC $\{\hat{\operator{N}}^{y|x}\}_{x,y}$,
describable by some functions $\{\hat{W}^{y|x}\}_{x,y}$ as 
in~\eqref{eq:def:channel:function:representation}. Let
$\hat{P}_{\rv{Y}_1^n|\rv{X}_1^n,\hat{\system{S}}_0}$ denote its joint channel
law, similar to~\eqref{eq:joint:2},~\eqref{eq:joint:3}, or~\eqref{eq:joint:4}. 
Namely,
\begin{equation}\hspace{0pt}
\hat{P}_{\rv{Y}_1^n|\rv{X}_1^n,\hat{\system{S}}_0}
    (\vy_1^n|\vx_1^n;\hat{\rho}_{\system{S}_0}) \defeq
\tr\!\left(\![\What^{y_{n}|x_{n}}]\cdots[\What^{y_{1}|x_{1}}]
    \!\cdot\![\hat{\rho}_{\system{S}_0}]\!\right)\!.\hspace{-11pt}
\end{equation}
We follow a similar approach as in~\cite{arnold2006simulation,
sadeghi2009optimization}, and define the quantities
\begin{align}
\label{eq:IRUB}
\IRUB_{W}^{(n)}(\hat{W}) & \defeq \begin{aligned}[t]
    &\frac{1}{n} \sum_{\vx_1^n,\vy_1^n}
    Q^{(n)}(\vx_1^n)\cdot P_{\rv{Y}_1^n|\rv{X}_1^n;\system{S}_0}
                          (\vy_1^n|\vx_1^n;\rho_{\system{S}_0})\\
    &\cdot\log\frac{P_{\rv{Y}_1^n|\rv{X}_1^n;\system{S}_0}
                   (\vy_1^n|\vx_1^n;\rho_{\system{S}_0})}
                  {\sum_{\cvx_1^n} Q^{(n)}(\cvx_1^n)
                   \hat{P}_{\rv{Y}_1^n|\rv{X}_1^n,\hat{\system{S}}_0}
                   (\vy_1^n|\cvx_1^n;\rho_{\hat{\system{S}}_0})},\hspace{-10pt}
 \end{aligned}\\
\label{eq:IRLB}
\IRLB_{W}^{(n)}(\hat{W}) & \defeq \begin{aligned}[t]
    &\frac{1}{n} \sum_{\vx_1^n,\vy_1^n}
    Q^{(n)}(\vx_1^n)\cdot P_{\rv{Y}_1^n|\rv{X}_1^n;\system{S}_0}
                          (\vy_1^n|\vx_1^n;\rho_{\system{S}_0})\\
    &\cdot\log\frac{\hat P_{\rv{Y}_1^n|\rv{X}_1^n;\system{S}_0}
                   (\vy_1^n|\vx_1^n;\rho_{\system{S}_0})}
                  {\sum_{\cvx_1^n} Q^{(n)}(\cvx_1^n)
                   \hat{P}_{\rv{Y}_1^n|\rv{X}_1^n,\hat{\system{S}}_0}
                   (\vy_1^n|\cvx_1^n;\rho_{\hat{\system{S}}_0})},\hspace{-10pt}
	\end{aligned}
\end{align}
where $P_{\rv{Y}_1^n|\rv{X}_1^n;\system{S}_0}$ is defined
in~\eqref{eq:joint:2},~\eqref{eq:joint:3} or~\eqref{eq:joint:4}.
By following similar arguments like those
in~\eqref{eq:IRUB:minus:IR} and~\eqref{eq:IR:minus:IRLB},
one can verify that 
\begin{equation}
\IRLB_{W}^{(n)}(\hat{W})\leqslant\infoRate^{(n)}_{W}\leqslant\IRUB_{W}^{(n)}(\hat{W}),
\end{equation}\addtocounter{equation}{9}
where the first inequality holds with equality if and only if 
$\hat{P}_{\rv{Y}_1^n|\rv{X}_1^n,\hat{\system{S}}_0}
    (\vy_1^n|\vx_1^n;\rho_{\hat{\system{S}}_0})$
and
$P_{\rv{Y}_1^n|\rv{X}_1^n;\system{S}_0}
    (\vy_1^n|\vx_1^n;\rho_{\system{S}_0})$
coincide for all $\vx_1^n$ and $\vy_1^n$ with positive support of
$P_{\rv{Y}_1^n|\rv{X}_1^n;\system{S}_0}$, 
and where the second inequalities holds with equality if and only if 
$\hat{P}_{\rv{Y}_1^n|\hat{\system{S}}_0}
    (\vy_1^n|\rho_{\hat{\system{S}}_0})$
and
$P_{\rv{Y}_1^n|\system{S}_0}
    (\vy_1^n|\rho_{\system{S}_0})$
coincide for all $\vy_1^n$ with positive support of
$P_{\rv{Y}_1^n|\system{S}_0}$.
Another quantity of interest is the \emph{difference function} defined as 
\begin{equation}\label{eq:DIFF}
\Delta_{W}^{(n)}(\hat{W}) \defeq \IRUB_{W}^{(n)}(\hat{W}) - \IRLB_{W}^{(n)}(\hat{W}).
\end{equation}
Explicit expressions of~\eqref{eq:IRUB},~\eqref{eq:IRLB}, and~\eqref{eq:DIFF}
are given by~\eqref{eq:IRUB:explicit},~\eqref{eq:IRLB:explicit},
and~\eqref{eq:DIFF:explicit}, respectively, at the top of this page, 
where $\rv{X}_1^n$ and $\rv{Y}_1^n$ are random variables
distributed according to the joint distribution 
$Q^{(n)}(\vx_1^n)\cdot P_{\rv{Y}_1^n|\rv{X}_1^n;\system{S}_0}
                       (\vy_1^n|\vx_1^n;\rho_{\system{S}_0})$,
and where $\left\langle\cdot\right\rangle$ stands for the expectation function.
\par
In the remainder of this section, we propose an algorithm based on the
gradient-descent method and the techniques described in Section~\ref{sec:NFGs}
and~\ref{sec:4:IR} for optimizing the quantities
in~\eqref{eq:IRUB},~\eqref{eq:IRLB}, and~\eqref{eq:DIFF}.
In particular, we consider $\{\hat{W}^{y|x}\}_{x,y}$
to be an \emph{interior} point in the domain of CC-QSCs, namely
\begin{itemize}
\item The Choi--Jamio{\l}kowski matrices
      $\llbracket \hat{W}^{y|x}\rrbracket$,
      defined similarly as \eqref{eq:qsc:matrix:2}, are strictly positive
      definite for each $x$ and $y$,
\item Eq.~\eqref{eq:operator:sum:representation:condition:2} holds
  by replacing $W^{y|x}$ with $\hat{W}^{y|x}$, namely 
  $\sum_{y\in\set{Y}} \sum_{s',\tilde{s}':\: s'=\tilde{s}'}
   \llbracket \hat{W}^{y|x}\rrbracket_{(s',s),(\tilde{s}',\tilde{s})}
   =\delta_{s,\tilde{s}}$ for all $x\in\set{X}$.
\end{itemize}
For any set of functions $\{H^{y|x}:\set{S}^4\rightarrow \mathbb{C}\}_{x,y}$ 
such that $\llbracket H^{y|x}\rrbracket$ (again, defined similarly
as~\eqref{eq:qsc:matrix:2}) is Hermitian for each $x$ and $y$, and such that
\begin{equation}\label{eq:tangent:linear:constraint}
  \sum_{y\in\set{Y}}\sum_{s',\tilde{s}':\: s'=\tilde{s}'}
  \llbracket H^{y|x}\rrbracket_{(s',s),(\tilde{s}',\tilde{s})}
  = 0 \quad \forall x\in\set{X},
\end{equation}
the functions $\{\hat{W}^{y|x}+t\cdot H^{y|x}\}_{x,y}$
describe a valid CC-QSC, for all $t$ in some neighborhood of $0$.
In this case, the directional derivatives of functions $\IRLB_{W}^{(n)}$,
$\IRUB_{W}^{(n)}$, and $\Delta_{W}^{(n)}$ at $\{\hat{W}^{y|x}\}_{x,y}$ \emph{along}
$\{H^{y|x}\}_{x,y}$ is well defined, and can be expressed
as~\eqref{eq:IRUB:dev:2},~\eqref{eq:IRLB:dev:2},
and~\eqref{eq:DIFF:dev:2} at the top of the last page, 
where we define the messages
$\{\rvec{\varrho}_{\system{S}_\ell}^{(\cvy_1^\ell)}\}_\ell$,
$\{\lvec{\varrho}_{\system{S}_\ell}^{(\cvy_{\ell\!+\!1}^n)}\}_\ell$,
$\{\rvec{\varrho}_{\system{S}_\ell}^{(\cvx_1^\ell,\cvy_1^\ell)}\}_\ell$,
and $\{\lvec{\varrho}_{\system{S}_\ell}
    ^{(\cvx_{\ell\!+\!1}^n,\cvy_{\ell\!+\!1}^n)}\}_\ell$
in a recursive manner as 
\begin{align}
&[\rvec{\varrho}_{\system{S}_\ell}^{(\cvy_1^\ell)}] \defeq
    \sum_{\vx_1^\ell} Q(\vx_1^\ell)\cdot
    [\hat{W}^{\cy_\ell|x_\ell}]\cdots
    [\hat{W}^{\cy_1|x_1}]\cdot
    [\rho_{\system{S}_0}],\\
&[\lvec{\varrho}_{\system{S}_\ell}^{(\cvy_{\ell\!+\!1}^n)}] \defeq 
    \sum_{\vx_{\ell\!+\!1}^n} Q(\vx_{\ell+1}^n) \!\cdot\!
    [I_{\system{S}_n}]\!\cdot\!
    [\hat{W}^{\cy_n|x_n}]\cdots
    [\hat{W}^{\cy_{\ell\!+\!1}|x_{\ell\!+\!1}}],\hspace{-3pt}\\
&[\rvec{\varrho}_{\system{S}_\ell}^{(\cvx_1^\ell,\cvy_1^\ell)}] \defeq
    [\hat{W}^{\cy_\ell|\cx_\ell}]\cdots
    [\hat{W}^{\cy_1|\cx_1}]\cdot
    [\rho_{\system{S}_0}],\\
&[\lvec{\varrho}_{\system{S}_\ell}
    ^{(\cvx_{\ell\!+\!1}^n,\cvy_{\ell\!+\!1}^n)}] \defeq
    [I_{\system{S}_n}]\cdot
    [\hat{W}^{\cy_n|\cx_n}]\cdots
    [\hat{W}^{\cy_{\ell\!+\!1}|\cx_{\ell\!+\!1}}].
\end{align}
Recall that, in above equations, $[I_{\system{S}_n}]$ is a row vector, whereas
$[\rho_{\system{S}_0}]$ is a column vector.
\par
By extending the domain of the functions $\IRLB_{W}^{(n)}$, $\IRUB_{W}^{(n)}$, and
$\Delta_{W}^{(n)}$ to include \emph{all} p.s.d. matrices
$\llbracket \hat{W}^{y|x}\rrbracket$, one can omit the linear
constraint~\eqref{eq:tangent:linear:constraint}. Namely, the ``direction''
$\{\llbracket H^{y|x}\rrbracket\}_{x,y}$ can take any Hermitian matrices.
Using some linear algebra, the gradient w.r.t. $\hat{W}$ of these functions on
this \emph{extended} domain can be expressed
as~\eqref{eq:IRUB:grad:1},~\eqref{eq:IRLB:grad:1}, and~\eqref{eq:DIFF:grad:1},
respectively, at the top of the last page.
For stationary and ergodic input and output processes $(\rv{X}_1^n,\rv{Y}_1^n)$,
we can \emph{estimate}~\eqref{eq:IRUB:grad:1} and~\eqref{eq:DIFF:grad:1},
respectively, as
\begin{align}
\label{eq:IRUB:grad:2}
\left(\grad\IRUB^{(n)}_{W,\mathrm{ext}}(\hat{W})\right)^{y|x}\hspace{-3pt} 
    &\overset{\cdot}{\propto}
    -\frac{1}{n} \sum_{k:{\cx_k=x \atop \cy_k=y}}\!\!
    \rvec{\varrho}_{\hat{\system{S}}_{k\!-\!1}}^{(\cvy_1^{k\!-\!1}\!)}
    \tensor\lvec{\varrho}_{\hat{\system{S}}_k}^{(\cvy_{k\!+\!1}^n\!)},\\
\label{eq:DIFF:grad:2}
\left(\grad\Delta_{W,\mathrm{ext}}^{(n)}(\hat{W})\right)^{y|x}\hspace{-3pt}
    &\overset{\cdot}{\propto}
    -\frac{1}{n} \sum_{k:{\cx_k=x \atop \cy_k=y}}\!\!
    \rvec{\varrho}_{\hat{\system{S}}_{k\!-\!1}}
        ^{(\cvx_1^{k\!-\!1}\!\!,\cvy_1^{k\!-\!1}\!)}
    \!\tensor\lvec{\varrho}_{\hat{\system{S}}_k}
        ^{(\cvx_{k\!+\!1}^n,\cvy_{k\!+\!1}^n\!)}\!,\hspace{-5pt}
\end{align}
where $(\cvx_1^n,\cvy_1^n)$ is a realization of the channel
input/output processes generated by the original channel model.
The dot in~\eqref{eq:IRUB:grad:2} and~\eqref{eq:DIFF:grad:2} stands for
``approximation''.
Notice that the messages 
$\rvec{\varrho}_{\system{S}_{k\!-\!1}}^{(\cvy_1^{k\!-\!1})}$,
$\lvec{\varrho}_{\system{S}_{k}}^{(\cvy_{k\!+\!1}^n)}$,
$\rvec{\varrho}_{\system{S}_{k\!-\!1}}
    ^{(\cvx_1^{k\!-\!1},\cvy_1^{k\!-\!1})}$, and 
$\lvec{\varrho}_{\system{S}_k}
    ^{(\cvx_{k\!+\!1}^n,\cvy_{k\!+\!1}^n)}$ can be computed
iteratively. Thus,~\eqref{eq:IRUB:grad:2} and \eqref{eq:DIFF:grad:2}
provide efficient means to estimate the gradient. However, due to the extension
of the domain, the gradients computed above may not satisfy 
constraint~\eqref{eq:tangent:linear:constraint}. This can be compensated using 
a projection w.r.t. the linear constraint, which can be solved using linear 
programming. On the other hand, the above gradient method may lead to a
violation of the p.s.d. condition required by CC-QSCs.
However, since the feasible domain of CC-QSCs is convex and bounded, this can be
corrected using convex programming at each step.
\par
\begin{algorithm}[h!]
\caption{Optimizing the difference function}
\begin{algorithmic}[1] 
\Require{indecomposable~CC-QSC,
         input~distribution~$Q$,
         positive~integer~$n$ large enough,
         initial~auxiliary~CC-QSC~$\{\hat{W}^{y|x}\}_{x,y}$,
         step~size~$\gamma>0$.}
\Ensure{$\{\hat{W}^{y|x}\}_{x,y}$, an estimated local minimum point of
        $\Delta_{W}^{(n)}$.}
\State Initialize the memory density operator
       $\rho_{\hat{\system{S}_0}}\gets\braket{0}$
\State Generate an input sequence $\cvx_1^n \sim Q^{\tensor n}$
\State Generate a corresponding output sequence $\cvy_1^n$
\Repeat
\State $\rvec{\varrho}_{\hat{\system{S}}_0}\gets\rho_{\hat{\system{S}}_0}$
\ForEach{$\ell=1,\ldots,n$}
\State $[\rvec{\varrho}_{\hat{\system{S}}_\ell}] \gets
        [\hat{W}^{\cy_\ell|\cx_\ell}] \cdot
        [\rvec{\varrho}_{\hat{\system{S}}_{\ell\! -\!1}}]$
\State $\lambda_\ell \gets \tr(\rvec{\varrho}_{\hat{\system{S}}_\ell})$
\State $\rvec{\varrho}_{\hat{\system{S}}_\ell} \gets \lambda_\ell^{-1}\cdot
        \rvec{\varrho}_{\hat{\system{S}}_\ell}$
\EndFor
\State $\lvec{\varrho}_{\hat{\system{S}}_n}\gets I_{\hat{\system{S}}_n}$
\ForEach{$\ell=n,\ldots,1$}
\State $[\lvec{\varrho}_{\hat{\system{S}}_{\ell\!-\!1}}] \gets
    [\lvec{\varrho}_{\hat{\system{S}}_\ell}]\cdot[\hat{W}^{\cy_\ell|\cx_\ell}]$
\State $\lvec{\varrho}_{\hat{\system{S}}_{\ell\!-\!1}} \gets
        \left(\tr(
        \lvec{\varrho}_{\hat{\system{S}}_{\ell\!-\!1}})\right)^{-1}\cdot
        \lvec{\varrho}_{\hat{\system{S}}_{\ell\!-\!1}}$
\EndFor
\State for each $x,y$, let
       $\left(\grad\Delta_{W,\mathrm{ext}}^{(n)}(\hat{W})\right)^{y|x}\gets \mathbf{0}$ 
\ForEach{$k=1,\ldots,n$}
    \State $\left(\grad\Delta_{W,\mathrm{ext}}^{(n)}(\hat{W})\right)^{\cy_k|\cx_k} \mathrel{+}=
            \frac{1}{n}\cdot
            \frac{\rvec{\varrho}_{\hat{\system{S}}_{k\!-\!1}}
            \tensor\lvec{\varrho}_{\hat{\system{S}}_k}}{
            \lambda_k\cdot \tr(\rvec{\varrho}_{\hat{\system{S}}_k}
            \cdot\lvec{\varrho}_{\hat{\system{S}}_k})}$
\EndFor
\State Project $\{\left(\grad\Delta_{W,\mathrm{ext}}^{(n)}(\hat{W})\right)^{y|x}\}_{x,y}$
    onto the subspace satisfying~\eqref{eq:tangent:linear:constraint};
    denoting the result by 
    $\left\{\left(\grad\Delta_{W}^{(n)}(\hat{W})\right)^{y|x}\right\}_{x,y}$
\State $\{\hat{W}^{y|x}\}_{x,y}\gets\{\hat{W}^{y|x}\}_{x,y}-\gamma\cdot
    \left\{\left(\grad\Delta_{W}^{(n)}(\hat{W})\right)^{y|x}\right\}_{x,y}$
\State Solve the following convex program w.r.t. $\{\tilde{W}^{y|x}\}_{x,y}$:
\begin{center}\begin{tabular}{rl}
min & $\Sum_{x,y}\tr\!\left(\!
       (\llbracket\tilde{W}^{y|x}\rrbracket-
       \llbracket\hat{W}^{y|x}\rrbracket)\cdot
       (\llbracket\tilde{W}^{y|x}\rrbracket
       -\llbracket\hat{W}^{y|x}\rrbracket)^\Herm\!\right) $\\
s.t. & $\llbracket \tilde{W}^{y|x}\rrbracket \in
       \mathbb{C}^{\set{S}^2\times\set{S}^2}$ is p.s.d. for each $x,y$ \\
& $\sum_{y\in\set{Y}}\sum_{s',\tilde{s}':\: s'=\tilde{s}'}
   \llbracket \tilde{W}^{y|x}\rrbracket_{(s',s),(\tilde{s}',\tilde{s})}
   =\delta_{s,\tilde{s}}\quad \forall x$
\end{tabular}\end{center}
\State $\{\hat{W}^{y|x}\}\gets\{\tilde{W}^{y|x}\}$
\Until{$\{\hat{W}^{y|x}\}_{x,y}$ has converged.}
\end{algorithmic}
\label{alg:grad:1}
\end{algorithm}
We summarize the above discussion as Algorithm~\ref{alg:grad:1}, which is an
iterative gradient-descent method for minimizing $\Delta_{W}^{(n)}$.
Notice that the quantity $\lambda_\ell$ in this case is the conditional
probability $P_{\rv{X}_\ell\rv{Y}_\ell|\rv{X}_{1}^{\ell\!-\!1}\rv{Y}_{1}^{\ell\!-\!1}}(\cx_{\ell},\cy_\ell|\cvx_1^{\ell\!-\!1},\cvy_1^{\ell\!-\!1})$.
The algorithm for minimizing the upper and lower bounds are similar, and we
omit the details.
\section{Example: Quantum Gilbert--Elliott Channels}\label{sec:6:example}
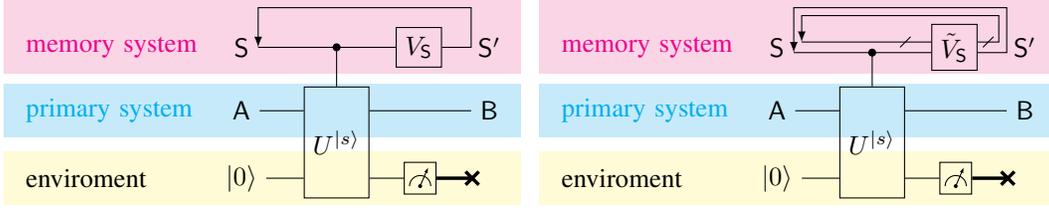
\begin{figure*}
    \centering
    
    \begin{tikzpicture}[every node/.style={transform shape}]
\node (S) {$\system{S}$};
\node[below = 10pt of S] (A) {$\system{A}$};
\node[below = 10pt of A] (E) {$\bra{0}$};
\path[draw=none] (A) edge[draw=none] node[midway] (AE) {} (E);
\node[draw, minimum width = 20pt, minimum height = 42pt, right = 20pt of AE] (CBF)
    {$U^{\bra{s}}$};
\node[right = 80pt of A] (B) {$\system{B}$};
\node[right = 52pt of E, draw, minimum size = 12pt] (M) {};

\node[minimum size = 1.5pt, fill = black, inner sep = 0pt, outer sep = 0pt,
    below = 3pt of M.center, circle, anchor = center] (m) {};
\draw [decoration={markings,mark=at position 1 with
    {\arrow[scale=0.5,>=latex]{>}}},postaction={decorate}, draw = none]
    (m.center) -- ([yshift=7pt, xshift=3pt]m.center);
\draw (m.center) -- ([yshift=6.3pt, xshift=2.7pt]m.center);
\draw ([xshift=4pt]m.center) arc (0:180:4pt);
\draw[line width = 1.3pt] (M.east) -- ([xshift=13.5pt]M.east) 
    node[pos=1, minimum size = 5pt, inner sep = 0pt, outer sep = 0pt] (Eend) {};
\draw[line width = 1.3pt] (Eend.south west) -- (Eend.north east);
\draw[line width = 1.3pt] (Eend.north west) -- (Eend.south east);

\node[right = 52pt of S, draw, minimum size = 12pt] (Us) {$V_\system{S}$};
\node at (Us-|B) {$\system{S}'$};

\draw (A) -- (A-|CBF.west);
\draw (E) -- (E-|CBF.west);
\draw (M-|CBF.east) -- (M);
\draw (B-|CBF.east) -- (B);
\draw (S) -- (Us);
\draw[-latex] (Us) -- (Us-|B.west) -- ([yshift=15pt]Us-|B.west) -- ([yshift=15pt]S.east)
              -- (S.east);
\draw[-*] (CBF) -- (S-|CBF);

\begin{pgfonlayer}{bg}
\draw [fill=magenta!20,draw=none] ([xshift=-90pt, yshift=18pt]S) rectangle 
      ([xshift=5pt,yshift=-10pt]Us-|B.east);
\node [left=85pt of S.center, anchor=west, color = magenta] {memory system};
\draw [fill=cyan!20,draw=none] ([xshift=-90pt, yshift=10pt]A) rectangle 
      ([xshift=5pt,yshift=-10pt]B.east);
\node [left=85pt of A.center, anchor=west, color = cyan] {primary system};
\draw [fill=yellow!20,draw=none] ([xshift=-90pt, yshift=10pt]E) rectangle 
      ([xshift=5pt,yshift=-10pt]E-|B.east);
\node [left=85pt of E.center, anchor=west] {enviroment};
\end{pgfonlayer}
\end{tikzpicture}
    
    \begin{tikzpicture}[every node/.style={transform shape}]
\node (S) {$\system{S}$};
\node[below = 10pt of S] (A) {$\system{A}$};
\node[below = 10pt of A] (E) {$\bra{0}$};
\path[draw=none] (A) edge[draw=none] node[midway] (AE) {} (E);
\node[draw, minimum width = 20pt, minimum height = 42pt, right = 20pt of AE] (CBF)
    {$U^{\bra{s}}$};
\node[right = 80pt of A] (B) {$\system{B}$};
\node[right = 52pt of E, draw, minimum size = 12pt] (M) {};

\node[minimum size = 1.5pt, fill = black, inner sep = 0pt, outer sep = 0pt,
    below = 3pt of M.center, circle, anchor = center] (m) {};
\draw [decoration={markings,mark=at position 1 with
    {\arrow[scale=0.5,>=latex]{>}}},postaction={decorate}, draw = none]
    (m.center) -- ([yshift=7pt, xshift=3pt]m.center);
\draw (m.center) -- ([yshift=6.3pt, xshift=2.7pt]m.center);
\draw ([xshift=4pt]m.center) arc (0:180:4pt);
\draw[line width = 1.3pt] (M.east) -- ([xshift=13.5pt]M.east) 
    node[pos=1, minimum size = 5pt, inner sep = 0pt, outer sep = 0pt] (Eend) {};
\draw[line width = 1.3pt] (Eend.south west) -- (Eend.north east);
\draw[line width = 1.3pt] (Eend.north west) -- (Eend.south east);

\node[right = 52pt of S, draw, minimum size = 12pt] (Us) {$\tilde{V}_\system{S}$};
\node at (Us-|B) {$\system{S}'$};

\draw (A) -- (A-|CBF.west);
\draw (E) -- (E-|CBF.west);
\draw (M-|CBF.east) -- (M);
\draw (B-|CBF.east) -- (B);
\draw ([yshift=-2pt]S.east) -- ([yshift=-2pt]Us.west);
\draw[-latex] ([yshift=-2pt]Us.east) -- ([yshift=-2pt]Us.east-|B.west)
              -- ([yshift=15pt]Us-|B.west) -- ([yshift=15pt]S.east)
              -- ([yshift=-2pt]S.east);
\draw[-*] (CBF) -- ([yshift=-2pt]S-|CBF);
\draw ([xshift=3pt, yshift=2pt]S.east) -- node[pos=0.8] (marker1) {} ([yshift=2pt]Us.west);
\draw[-latex] ([yshift=2pt]Us.east) -- node[midway] (marker2) {} ([yshift=2pt,xshift=-3pt]Us.east-|B.west)
              -- ([yshift=12pt,xshift=-3pt]Us-|B.west) -- ([xshift=3pt,yshift=12pt]S.east)
              -- ([xshift=3pt,yshift=2pt]S.east);
\draw ([xshift=-2pt, yshift=-2pt]marker1.center) -- ([xshift=2pt, yshift=2pt]marker1.center);
\draw ([xshift=-2pt, yshift=-2pt]marker2.center) -- ([xshift=2pt, yshift=2pt]marker2.center);

\begin{pgfonlayer}{bg}
\draw [fill=magenta!20,draw=none] ([xshift=-90pt, yshift=18pt]S) rectangle 
      ([xshift=5pt,yshift=-10pt]Us-|B.east);
\node [left=85pt of S.center, anchor=west, color = magenta] {memory system};
\draw [fill=cyan!20,draw=none] ([xshift=-90pt, yshift=10pt]A) rectangle 
      ([xshift=5pt,yshift=-10pt]B.east);
\node [left=85pt of A.center, anchor=west, color = cyan] {primary system};
\draw [fill=yellow!20,draw=none] ([xshift=-90pt, yshift=10pt]E) rectangle 
      ([xshift=5pt,yshift=-10pt]E-|B.east);
\node [left=85pt of E.center, anchor=west] {enviroment};
\end{pgfonlayer}
\end{tikzpicture}
    \caption{A quantum Gilbert--Elliott channel (LHS),
             and a variant where the memory system consists of multiple
             qubits with only one of them controlling $U^{\bra{s}}$ (RHS).}
    \label{fig:QGEC}
\end{figure*}
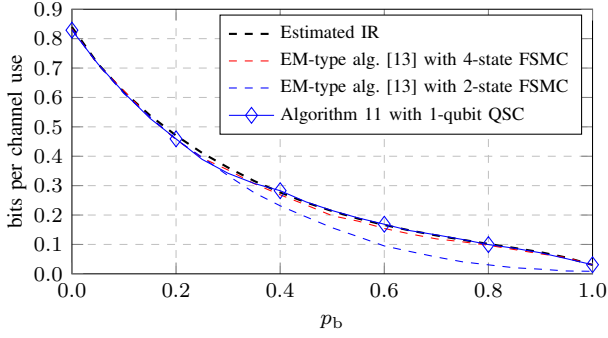
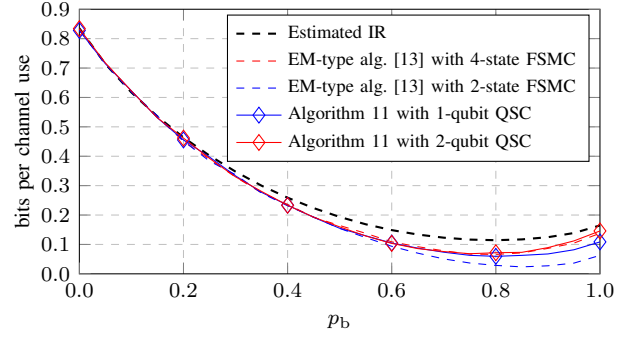
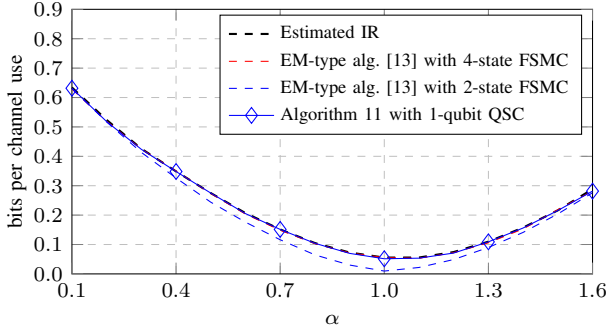
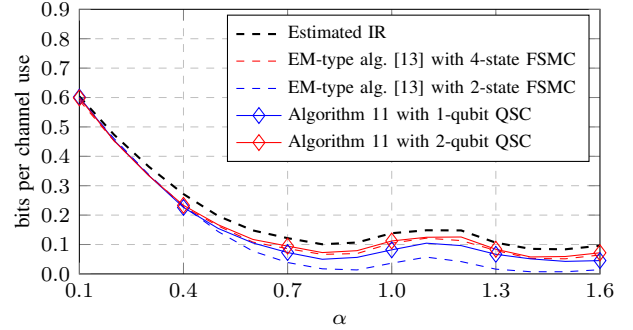
\begin{figure*}
\centering
\begin{subfigure}[t]{0.96\columnwidth}
    
    \begin{tikzpicture}[every axis/.append style={font=\footnotesize},
                    every mark/.append style={scale=1.5}]
    \begin{axis}[xlabel={$\pbad$}, 
                 ylabel={bits per channel use},
                 ylabel shift=-4 pt,
                 xmin=0, xmax=1, ymin=0,  ymax=0.9,
                 legend style={font=\scriptsize}, 
                 mark repeat={4}, 
                 width=\columnwidth, height=.6\columnwidth,
                 ytick={0.1,0.2,0.3,0.4,0.5,0.6,0.7,0.8},
                 extra y ticks={0,0.9},
                 extra y tick style={grid=none},
                 y tick label style={
                     /pgf/number format/fixed,
                     /pgf/number format/fixed zerofill,
                     /pgf/number format/precision=1},
                 xtick={0.2,0.4,0.6,0.8},
                 extra x ticks={0.0,1.0},
                 extra x tick style={grid=none},
                 x tick label style={
                     /pgf/number format/fixed,
                     /pgf/number format/fixed zerofill,
                     /pgf/number format/precision=1},
                 grid=both,
                 grid style={line width=0.1pt, dashed},
                 legend cell align={left}]
        \addplot[dashed, thick, black] coordinates {
            (0.00, 0.839601745174)
            (0.05, 0.714120015688)
            (0.10, 0.618261781187)
            (0.15, 0.538893269110)
            (0.20, 0.470418506076)
            (0.25, 0.411550042780)
            (0.30, 0.361325752184)
            (0.35, 0.315858686126)
            (0.40, 0.276693640181)
            (0.45, 0.242634877108)
            (0.50, 0.213364467275)
            (0.55, 0.187646703615)
            (0.60, 0.166111113269)
            (0.65, 0.147285664160)
            (0.70, 0.130630670279)
            (0.75, 0.115773844628)
            (0.80, 0.101787384142)
            (0.85, 0.087948744743)
            (0.90, 0.073629956869)
            (0.95, 0.055286167568)
            (1.00, 0.030194299618)
        };
        \addlegendentry{Estimated IR}
        \addplot[solid, red, dashed] coordinates {
            (0.00, 0.833686348679)
            (0.05, 0.711496512010)
            (0.10, 0.622781653924)
            (0.15, 0.531787273503)
            (0.20, 0.456691524228)
            (0.25, 0.394254900201)
            (0.30, 0.353083453813)
            (0.35, 0.306942832207)
            (0.40, 0.268742612068)
            (0.45, 0.234172015381)
            (0.50, 0.197625745231)
            (0.55, 0.179086497776)
            (0.60, 0.154340415942)
            (0.65, 0.135508136380)
            (0.70, 0.118932943912)
            (0.75, 0.108928054432)
            (0.80, 0.097114178557)
            (0.85, 0.082541871698)
            (0.90, 0.067974760037)
            (0.95, 0.055364118464)
            (1.00, 0.030712333955)
        };
        \addlegendentry{EM-type alg.~[13] with 4-state FSMC}
        \addplot[solid, blue, dashed] coordinates {
            (0.00, 0.831741905571)
            (0.05, 0.712427595410)
            (0.10, 0.613164794327)
            (0.15, 0.533958257063)
            (0.20, 0.457099591967)
            (0.25, 0.396516510418)
            (0.30, 0.333931200180)
            (0.35, 0.276208896384)
            (0.40, 0.231405341707)
            (0.45, 0.193564225477)
            (0.50, 0.156714715230)
            (0.55, 0.126192277424)
            (0.60, 0.094889516489)
            (0.65, 0.075748267428)
            (0.70, 0.057208456686)
            (0.75, 0.040925933725)
            (0.80, 0.030531953773)
            (0.85, 0.020645659910)
            (0.90, 0.015760096465)
            (0.95, 0.009658573661)
            (1.00, 0.008579667796)
        };
        \addlegendentry{EM-type alg.~[13] with 2-state FSMC}
        \addplot[solid, mark = diamond, blue] coordinates {
            (0.00, 0.828760809658)
            (0.05, 0.713985929442)
            (0.10, 0.615994378803)
            (0.15, 0.528854930308)
            (0.20, 0.459058469386)
            (0.25, 0.389432982447)
            (0.30, 0.341934887183)
            (0.35, 0.307475093096)
            (0.40, 0.282804615465)
            (0.45, 0.244481977767)
            (0.50, 0.215103229847)
            (0.55, 0.190116309030)
            (0.60, 0.168708522517)
            (0.65, 0.146737015048)
            (0.70, 0.132491368761)
            (0.75, 0.116971911450)
            (0.80, 0.099518348532)
            (0.85, 0.086967374828)
            (0.90, 0.071632787144)
            (0.95, 0.051482497374)
            (1.00, 0.030962259381)
        };
        \addlegendentry{Algorithm~11 with 1-qubit QSC}
    \end{axis}
\end{tikzpicture}
	\caption{
    Quantum Gilbert–Elliott Channel: $\pgood = 0.05$ is fixed; $\pbad$
    varies from $0$ to $1$; $V_\system{S} = \exp(-j \alpha H)$, where
    $H$ is some fixed $2$-by-$2$
    Hermitian matrix and where $\alpha = 1$ is fixed;
    $n=10^5$.
    } \label{fig:QGEC:plot:1}
\end{subfigure}
~
\begin{subfigure}[t]{0.96\columnwidth}
    
    \begin{tikzpicture}[every axis/.append style={font=\footnotesize},
                    every mark/.append style={scale=1.5}]
    \begin{axis}[xlabel={$\pbad$}, 
                 ylabel={bits per channel use},
                 ylabel shift=-4 pt,
                 xmin=0, xmax=1, ymin=0,  ymax=0.9,
                 legend style={font=\scriptsize}, 
                 mark repeat={4}, 
                 width=\columnwidth, height=.6\columnwidth,
                 y tick label style={
                     /pgf/number format/fixed,
                     /pgf/number format/fixed zerofill,
                     /pgf/number format/precision=1},
                 ytick={0.1,0.2,0.3,0.4,0.5,0.6,0.7,0.8},
                 extra y ticks={0,0.9},
                 extra y tick style={grid=none},
                 xtick={0.2,0.4,0.6,0.8},
                 extra x ticks={0.0,1.0},
                 extra x tick style={grid=none},
                 x tick label style={
                     /pgf/number format/fixed,
                     /pgf/number format/fixed zerofill,
                     /pgf/number format/precision=1},
                 grid=both,
                 grid style={line width=0.1pt, dashed},
                 legend cell align={left}]
        \addplot[dashed, black, thick] coordinates {
            (0.00, 0.835334547072)
            (0.05, 0.713817138457)
            (0.10, 0.616995721117)
            (0.15, 0.535170417981)
            (0.20, 0.463871511452)
            (0.25, 0.402218381517)
            (0.30, 0.347600395229)
            (0.35, 0.299817809855)
            (0.40, 0.258942093992)
            (0.45, 0.223665276283)
            (0.50, 0.193521758116)
            (0.55, 0.168561414250)
            (0.60, 0.148733205093)
            (0.65, 0.133686946797)
            (0.70, 0.122912164107)
            (0.75, 0.116431195033)
            (0.80, 0.114371410381)
            (0.85, 0.117102579296)
            (0.90, 0.124549747387)
            (0.95, 0.138826478090)
            (1.00, 0.164582211924)
        };
        \addlegendentry{Estimated IR}
        \addplot[solid, red, dashed] coordinates {
            (0.00, 0.831517212717)
            (0.05, 0.714764648111)
            (0.10, 0.620863015847)
            (0.15, 0.532861848575)
            (0.20, 0.450692752931)
            (0.25, 0.389546432757)
            (0.30, 0.329290008573)
            (0.35, 0.275553175585)
            (0.40, 0.236190228422)
            (0.45, 0.192090727923)
            (0.50, 0.164389649731)
            (0.55, 0.137826355554)
            (0.60, 0.108796192389)
            (0.65, 0.092290544995)
            (0.70, 0.078342868664)
            (0.75, 0.067132930525)
            (0.80, 0.065980300478)
            (0.85, 0.071230799275)
            (0.90, 0.086533592418)
            (0.95, 0.104559960603)
            (1.00, 0.137889776752)
        };
        \addlegendentry{EM-type alg.~[13] with 4-state FSMC}
        \addplot[solid, blue, dashed] coordinates {
            (0.00, 0.828833977289)
            (0.05, 0.713401902069)
            (0.10, 0.618383090362)
            (0.15, 0.536796068619)
            (0.20, 0.452103146643)
            (0.25, 0.383414537377)
            (0.30, 0.337204267691)
            (0.35, 0.274681469889)
            (0.40, 0.231386008521)
            (0.45, 0.193295725224)
            (0.50, 0.154478319411)
            (0.55, 0.123336899762)
            (0.60, 0.093653236031)
            (0.65, 0.070335085534)
            (0.70, 0.050775825078)
            (0.75, 0.036778413585)
            (0.80, 0.028658913432)
            (0.85, 0.023617608576)
            (0.90, 0.027313636512)
            (0.95, 0.036456654565)
            (1.00, 0.062739939978)
        };
        \addlegendentry{EM-type alg.~[13] with 2-state FSMC}
        \addplot[solid, mark = diamond, blue] coordinates {
            (0.00, 0.827985156616)
            (0.05, 0.714668236552)
            (0.10, 0.622321223551)
            (0.15, 0.530413033978)
            (0.20, 0.455413779164)
            (0.25, 0.397241693779)
            (0.30, 0.333206279519)
            (0.35, 0.277076762749)
            (0.40, 0.233492782705)
            (0.45, 0.192237937528)
            (0.50, 0.154956559171)
            (0.55, 0.129387917312)
            (0.60, 0.104722394072)
            (0.65, 0.087434544657)
            (0.70, 0.073503966393)
            (0.75, 0.062853525376)
            (0.80, 0.059483982086)
            (0.85, 0.062788063958)
            (0.90, 0.067646229180)
            (0.95, 0.081847657751)
            (1.00, 0.108539715332)
        };
        \addlegendentry{Algorithm~11 with 1-qubit QSC}
        \addplot[solid, mark = diamond, red] coordinates {
            (0.00, 0.834042100608)
            (0.05, 0.717419247603)
            (0.10, 0.621893971509)
            (0.15, 0.530616289996)
            (0.20, 0.461705781004)
            (0.25, 0.392711338005)
            (0.30, 0.329821195930)
            (0.35, 0.280192615343)
            (0.40, 0.233190489326)
            (0.45, 0.192814793929)
            (0.50, 0.159635194920)
            (0.55, 0.128742999094)
            (0.60, 0.103480124610)
            (0.65, 0.085811168978)
            (0.70, 0.075800104426)
            (0.75, 0.068367278168)
            (0.80, 0.071753040441)
            (0.85, 0.073104108032)
            (0.90, 0.089161880328)
            (0.95, 0.112025283761)
            (1.00, 0.145801283377)
        };
        \addlegendentry{Algorithm~11 with 2-qubit QSC}
    \end{axis}
\end{tikzpicture}
    \caption{
    Variant of the Quantum Gilbert–Elliott Channel described in the
    RHS of Fig.~\ref{fig:QGEC}.
    Here, the memory system $\system{S}$ consists of two qubits, with only the 
    first one interacting with the primary system by serving as the controlling
    qubit of the controlled bit-flip channel.
    Parameters: $\pgood = 0.05$; $\pbad\in[0,1]$;
    $\tilde{V}_\system{S} = \exp(-j\alpha H)$, where $H$ is some fixed $4$-by-$4$
    Hermitian matrix and where $\alpha = 1$ is fixed; $n = 10^5$.
    } \label{fig:QGEC:plot:2}
\end{subfigure}
~
\begin{subfigure}[t]{0.96\columnwidth}
    
    \begin{tikzpicture}[every axis/.append style={font=\footnotesize},
                    every mark/.append style={scale=1.5}]
    \begin{axis}[xlabel={$\alpha$}, 
                 ylabel={bits per channel use},
                 ylabel shift=-4 pt,
                 xmin=0.1, xmax=1.6, ymin=0,  ymax=0.9,
                 legend style={font=\scriptsize}, 
                 mark repeat={3}, 
                 width=\columnwidth, height=.6\columnwidth,
                 y tick label style={
                     /pgf/number format/fixed,
                     /pgf/number format/fixed zerofill,
                     /pgf/number format/precision=1},
                 ytick={0.1,0.2,0.3,0.4,0.5,0.6,0.7,0.8},
                 extra y ticks={0,0.9},
                 extra y tick style={grid=none},
                 xtick={0.4,0.7,1.0,1.3},
                 extra x ticks={0.1,1.6},
                 extra x tick style={grid=none},
                 x tick label style={
                     /pgf/number format/fixed,
                     /pgf/number format/fixed zerofill,
                     /pgf/number format/precision=1},
                 grid=both,
                 grid style={line width=0.1pt, dashed},
                 legend cell align={left}]
        \addplot[dashed, black, thick] coordinates {
            (0.1, 0.634246436353)
            (0.2, 0.523724519704)
            (0.3, 0.425843507895)
            (0.4, 0.348767614102)
            (0.5, 0.274949201240)
            (0.6, 0.206807177630)
            (0.7, 0.151334883495)
            (0.8, 0.104972113295)
            (0.9, 0.072702348555)
            (1.0, 0.055474107518)
            (1.1, 0.055705182488)
            (1.2, 0.074364078064)
            (1.3, 0.109701367088)
            (1.4, 0.157408739672)
            (1.5, 0.216248463821)
            (1.6, 0.288383375526)
        };
        \addlegendentry{Estimated IR}
        \addplot[solid, red, dashed] coordinates {
            (0.1, 0.628757790182)
            (0.2, 0.519154991495)
            (0.3, 0.424800128488)
            (0.4, 0.349505951177)
            (0.5, 0.273691145329)
            (0.6, 0.203689311260)
            (0.7, 0.148728055607)
            (0.8, 0.103202914342)
            (0.9, 0.071644182670)
            (1.0, 0.055364118464)
            (1.1, 0.053273311931)
            (1.2, 0.072222309436)
            (1.3, 0.105810164407)
            (1.4, 0.153058635133)
            (1.5, 0.218322703723)
            (1.6, 0.281969975885)
        };
        \addlegendentry{EM-type alg.~[13] with 4-state FSMC}
        \addplot[solid, blue, dashed] coordinates {
            (0.1, 0.630537416883)
            (0.2, 0.516839131244)
            (0.3, 0.413357758219)
            (0.4, 0.326340217404)
            (0.5, 0.244019998069)
            (0.6, 0.176033806634)
            (0.7, 0.116161198638)
            (0.8, 0.064796428188)
            (0.9, 0.030125167285)
            (1.0, 0.009658573661)
            (1.1, 0.021889961949)
            (1.2, 0.048929762310)
            (1.3, 0.089895723650)
            (1.4, 0.141175364024)
            (1.5, 0.203941640222)
            (1.6, 0.275966557336)
        };
        \addlegendentry{EM-type alg.~[13] with 2-state FSMC}
        \addplot[solid, mark = diamond, blue] coordinates {
            (0.100, 0.631573943140)
            (0.200, 0.518053760695)
            (0.300, 0.423761299113)
            (0.400, 0.348255291277)
            (0.500, 0.274996844752)
            (0.600, 0.205089306142)
            (0.700, 0.151321209016)
            (0.800, 0.104187042637)
            (0.900, 0.069819436449)
            (1.000, 0.051482497374)
            (1.100, 0.053512560728)
            (1.200, 0.070980855572)
            (1.300, 0.109099922099)
            (1.400, 0.156097904004)
            (1.500, 0.215592931815)
            (1.600, 0.281598938534)
        };
        \addlegendentry{Algorithm~11 with 1-qubit QSC}
    \end{axis}
\end{tikzpicture}
    \caption{
    Quantum Gilbert–Elliott Channel: $\pgood = 0.05$ is fixed; $\pbad =
    0.95$ is fixed; $V_\system{S} = \exp(-j \alpha H)$, where $H$ is the same
    $2$-by-$2$ Hermitian matrix as in Fig.~\ref{fig:QGEC:plot:1} and where
    $\alpha$ varies from $0.1$ to $+1.5$; $n = 10^5$.
    } \label{fig:QGEC:plot:3}
\end{subfigure}
~
\begin{subfigure}[t]{0.96\columnwidth}
    
    \begin{tikzpicture}[every axis/.append style={font=\footnotesize},
                    every mark/.append style={scale=1.5}]
    \begin{axis}[xlabel={$\alpha$}, 
                 ylabel={bits per channel use},
                 ylabel shift=-4 pt,
                 xmin=0.1, xmax=1.6, ymin=0,  ymax=0.9,
                 legend style={font=\scriptsize}, 
                 mark repeat={3}, 
                 width=\columnwidth, height=.6\columnwidth,
                 y tick label style={
                     /pgf/number format/fixed,
                     /pgf/number format/fixed zerofill,
                     /pgf/number format/precision=1},
                 ytick={0.1,0.2,0.3,0.4,0.5,0.6,0.7,0.8},
                 extra y ticks={0,0.9},
                 extra y tick style={grid=none},
                 xtick={0.4,0.7,1.0,1.3},
                 extra x ticks={0.1,1.6},
                 extra x tick style={grid=none},
                 x tick label style={
                     /pgf/number format/fixed,
                     /pgf/number format/fixed zerofill,
                     /pgf/number format/precision=1},
                 grid=both,
                 grid style={line width=0.1pt, dashed},
                 legend cell align={left}]
        \addplot[dashed, black, thick] coordinates {
            (0.1, 0.604100063356)
            (0.2, 0.473401110905)
            (0.3, 0.364883790848)
            (0.4, 0.271535196296)
            (0.5, 0.196980525949)
            (0.6, 0.147892471850)
            (0.7, 0.121579631965)
            (0.8, 0.100800736650)
            (0.9, 0.106741942766)
            (1.0, 0.138425721984)
            (1.1, 0.148411159003)
            (1.2, 0.147873920267)
            (1.3, 0.106213495216)
            (1.4, 0.085179900573)
            (1.5, 0.083684790791)
            (1.6, 0.095801652946)
        };
        \addlegendentry{Estimated IR}
        \addplot[solid, red, dashed] coordinates {
            (0.1, 0.590030613376)
            (0.2, 0.452373967863)
            (0.3, 0.336164446674)
            (0.4, 0.228386828122)
            (0.5, 0.163171954985)
            (0.6, 0.108116007087)
            (0.7, 0.085836992514)
            (0.8, 0.066406308630)
            (0.9, 0.069213426994)
            (1.0, 0.104559960603)
            (1.1, 0.121506799903)
            (1.2, 0.113347442020)
            (1.3, 0.079694559896)
            (1.4, 0.052635475909)
            (1.5, 0.051345016335)
            (1.6, 0.064265969362)
        };
        \addlegendentry{EM-type alg.~[13] with 4-state FSMC}
        \addplot[solid, blue, dashed] coordinates {
            (0.1, 0.596576347658)
            (0.2, 0.460288337401)
            (0.3, 0.337629300284)
            (0.4, 0.230703408080)
            (0.5, 0.143804485696)
            (0.6, 0.077428614807)
            (0.7, 0.038734579522)
            (0.8, 0.017294618485)
            (0.9, 0.013494816112)
            (1.0, 0.036456654565)
            (1.1, 0.057035136631)
            (1.2, 0.041939227789)
            (1.3, 0.016095530905)
            (1.4, 0.007397882922)
            (1.5, 0.007319105337)
            (1.6, 0.014187593031)
        };
        \addlegendentry{EM-type alg.~[13] with 2-state FSMC}
        \addplot[solid, mark = diamond, blue] coordinates {
            (0.1, 0.601528340475)
            (0.2, 0.454463594618)
            (0.3, 0.335020097620)
            (0.4, 0.226052428000)
            (0.5, 0.154639172891)
            (0.6, 0.105263825537)
            (0.7, 0.072710389970)
            (0.8, 0.049759820290)
            (0.9, 0.056187338246)
            (1.0, 0.081847657751)
            (1.1, 0.104479497919)
            (1.2, 0.096208892501)
            (1.3, 0.066664735539)
            (1.4, 0.051219739998)
            (1.5, 0.042557259501)
            (1.6, 0.045233504158)
        };
        \addlegendentry{Algorithm~11 with 1-qubit QSC}
        \addplot[solid, mark = diamond, red] coordinates {
            (0.1, 0.598825525120)
            (0.2, 0.452793632504)
            (0.3, 0.333983755168)
            (0.4, 0.233252042999)
            (0.5, 0.167063892604)
            (0.6, 0.117258612189)
            (0.7, 0.094584955365)
            (0.8, 0.072521235449)
            (0.9, 0.079008564509)
            (1.0, 0.112025283761)
            (1.1, 0.124629280264)
            (1.2, 0.125535074107)
            (1.3, 0.083770542149)
            (1.4, 0.057645455501)
            (1.5, 0.059133475287)
            (1.6, 0.071769388212)
        };
        \addlegendentry{Algorithm~11 with 2-qubit QSC}
    \end{axis}
\end{tikzpicture}
    \caption{
    Same variant of the Quantum Gilbert–Elliott Channel as in
    Fig.~\ref{fig:QGEC:plot:2} with different parameters: $\pgood = 0.05$;
    $\pbad = 0.95$; $V_\system{S} = \exp(-j \alpha H)$, where $H$ is the same 
    $4$-by-$4$ Hermitian matrix as in Fig.~\ref{fig:QGEC:plot:2} and where
    $\alpha$ varies from $0.1$ to $+1.5$; $n = 10^5$.
    } \label{fig:QGEC:plot:4}
\end{subfigure}
\caption{Some numerical information rate lower bounds estimated for a QGEC
    and a variant of a QGEC, equipped with ``trivial'' orthonormal ensemble
    and projective measurements.
    The estimated information rates were obtained using Alg.~\ref{alg:SPA:2}.}
\label{fig:QGEC:plots}
\end{figure*}
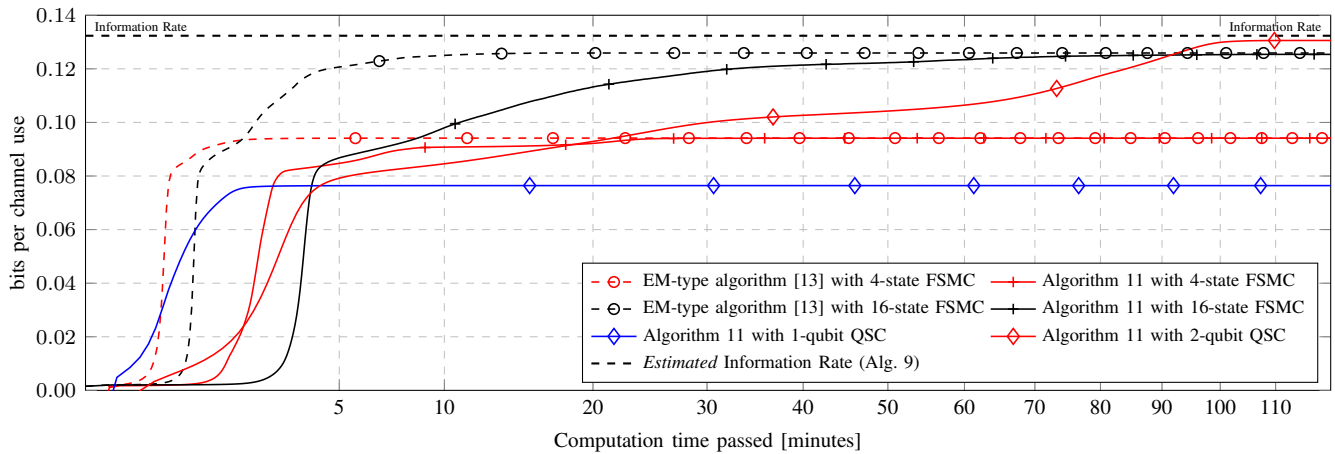
\begin{figure*}
    \centering
    
    \begin{tikzpicture}[every axis/.append style={font=\footnotesize},
                    every mark/.append style={scale=1.5}]
    \begin{axis}[xlabel={Computation time passed [minutes]}, 
                 ylabel={bits per channel use},
                 ylabel shift=-4 pt,
                 xmin=0, xmax=85, ymin=0,  ymax=0.14,
                 legend style={font=\scriptsize},
                 legend columns=2,
                 mark repeat=200, 
                 mark phase = 200,
                 width=\textwidth, height=187pt,
                 y tick label style={
                     /pgf/number format/fixed,
                     /pgf/number format/fixed zerofill,
                     /pgf/number format/precision=2},
                 ytick={0.02,0.04,0.06,0.08,0.1,0.12},,
                 extra y ticks={0,0.14},
                 extra y tick style={grid=none},
                 xticklabels={5,10,20,30,40,50,60,70,80,90,100,110},
                 xtick={17.3205,24.4949,34.6410,42.4264,48.9898,54.7723,
                     60.0000,64.8074,69.2820,73.4847,77.4597,81.2404},
                 x tick label style={
                     /pgf/number format/fixed,
                     /pgf/number format/fixed zerofill,
                     /pgf/number format/precision=0},
                 grid=both,
                 grid style={line width=0.1pt, dashed},
                 legend cell align={left},
                 legend style={at={(0.99,0.02)},anchor=south east},
                 every axis plot/.append style={semithick}]
        \addplot[red, dashed, mark = o, mark options={solid}]
            table[x expr=((\thisrow{step}-1)*0.848424200)^(0.5),y=IR_L]
            {data/cem_4_half.data};
        \addlegendentry{EM-type algorithm~[13] with 4-state FSMC} 
        \addplot[red, solid, mark = +, mark options={solid}]
            table[x expr=((\thisrow{step}-1)*1.343300420)^(0.5),y=IR_L]
            {data/cgd_4_half.data};
        \addlegendentry{Algorithm~11 with 4-state FSMC}
        \addplot[black, dashed, mark = o, mark options={solid}]
            table[x expr=((\thisrow{step}-1)*1.010961860)^(0.5),y=IR_L]
            {data/cem_16_half.data};
        \addlegendentry{EM-type algorithm~[13] with 16-state FSMC} 
        \addplot[black, solid, mark = +, mark options={solid}]
            table[x expr=((\thisrow{step}-1)*1.599697330)^(0.5),y=IR_L]
            {data/cgd_16_half.data};
        \addlegendentry{Algorithm~11 with 16-state FSMC}
        \addplot[blue, mark=diamond]
            table[x expr=((\thisrow{step}-1)*2.298483867)^(0.5),y=IR_L]
            {data/qgd_1_half.data};
        \addlegendentry{Algorithm~11 with 1-qubit QSC}
        \addplot[red, mark=diamond]
            table[x expr=((\thisrow{step}-1)*5.480697869)^(0.5),y=IR_L]
            {data/qgd_2_half.data};
        \addlegendentry{Algorithm~11 with 2-qubit QSC}
        \addplot[dashed,black,thick] coordinates{(0000,0.132352)(90,0.132352)};
        \addlegendentry{\emph{Estimated} Information Rate (Alg.~9)} 
        \node[anchor=north west, font=\tiny, yshift=1pt] at (rel axis cs:0,1)
             {Information Rate};
        \node[anchor=north east, font=\tiny, yshift=1pt] at (rel axis cs:1,1)
             {Information Rate};
    \end{axis}
\end{tikzpicture}
    \caption{Minimizing the difference function $\Delta_{W}^{(n)}$ using 
    different methods. The markers appear after every 400 updates.}
    \label{fig:QGEC:plot:5}
    \vspace{-10pt}
\end{figure*}
In this section we present some numerical results as a demonstration of the 
    algorithms introduced in this paper.
In particular, as a generalization of Example~\ref{example:GEC},
    we consider a class of quantum channels with memory named the
    quantum Gilbert--Elliott channels (QGECs), which were introduced
    in~\cite{cao2017estimating}, and consider their information
    rates using some separable input ensemble and local output measurement.
\par
A QGEC is a quantum channel with memory defined by\footnote{We put the system
$\system{S}$ ahead of $\system{A}$ and $\system{B}$ in this example to
emphasize the role of $\system{S}$ as a \emph{control} qubit, and also for
simplicity reasons.}
\begin{alignat*}{2}
&\operator{N}:\:& \DensOp(\hilbert_\system{S}\tensor\hilbert_\system{A}) 
& \rightarrow \DensOp(\hilbert_{\system{S}'}\tensor\hilbert_\system{B})\\
& & \rho_{\system{SA}} & \mapsto (V_\system{S}\tensor I_\system{B})\cdot
\Phi^{\mathrm{CBF}}(\rho_{\system{SA}})\cdot
(V_\system{S}^\Herm\tensor I_\system{B}),
\end{alignat*}
where $\hilbert_\system{A}$, $\hilbert_\system{B}$, and 
$\hilbert_\system{S}=\hilbert_{\system{S}'}$ are of dimension 2,
namely each of them is made up of one qubit; and where
$\Phi^{\mathrm{CBF}}$ is the \emph{controlled bit-flip channel} defined by
$\Phi^{\mathrm{CBF}}(\rho_\system{SA}) \defeq
E_0\rho^{\mathrm{SA}}E_0^\Herm + E_1\rho^{\mathrm{SA}}E_1^\Herm $
with
\[
E_0 \defeq\left[
\begin{smallmatrix}
\!\!\sqrt{1\!-\!\pgood}\!\! & 0                & 0   & 0 \\
0                & \!\!\sqrt{1\!-\!\pgood}\!\! & 0   & 0 \\
0                & 0                & \!\!\sqrt{1\!-\!\pbad}\!\! & 0 \\
0                & 0                & 0   & \!\!\sqrt{1\!-\!\pbad}\!\!
\end{smallmatrix}\right],\ 
E_1 \defeq \left[
\begin{smallmatrix}
0                  & \sqrt{\pgood} & 0     & 0 \\
\sqrt{\pgood} & 0            & 0           & 0 \\
0                  & 0            & 0   & \sqrt{\pbad} \\
0                  & 0            & \sqrt{\pbad} & 0
\end{smallmatrix}\right];
\]
and where $V_\system{S}$ is some unitary operator on $\hilbert_\system{S}$
to be specified later.
The controlled bit-flip channel $\Phi^{\mathrm{CBF}}$ applies a quantum
bit-flip channel on system $\system{A}$ with flipping probability $\pgood$
when the system $\system{S}$ is in the state of $\bra{0}$, and
with flipping probability $\pbad$ when the system $\system{S}$ is in
the state of $\bra{1}$.
The action of a QGEC is the combined effect of a controlled bit-flip channel
and a unitary evolution on $\system{S}$; as depicted in the following circuit
diagram in Fig.~\ref{fig:QGEC}, where $U^{\bra{s}}$ is a Stinespring
representation of $\Phi^{\mathrm{CBF}}$:
\[
\begin{aligned}
&U^{\bra{0}} \defeq\left[
\begin{smallmatrix}
\sqrt{1\!-\!\pgood}  & 0
& 0                          & -\sqrt{\pgood}  \\
0                            & \sqrt{1\!-\!\pgood}
&\sqrt{\pgood}      & 0 \\
0                            & -\sqrt{\pgood}
& \sqrt{1\!-\!\pgood}& 0 \\
\sqrt{\pgood}        & 0
& 0                          & \sqrt{1\!-\!\pgood}
\end{smallmatrix}\right],\\
&U^{\bra{1}} \defeq \left[
\begin{smallmatrix}
\sqrt{1\!-\!\pbad}  & 0
& 0                          & -\sqrt{\pbad}  \\
0                            & \sqrt{1\!-\!\pbad}
&\sqrt{\pbad}      & 0 \\
0                            & -\sqrt{\pbad}
& \sqrt{1\!-\!\pbad}& 0 \\
\sqrt{\pbad}        & 0
& 0                          & \sqrt{1\!-\!\pbad}
\end{smallmatrix}\right].
\end{aligned}\]
\par
In Fig.~\ref{fig:QGEC:plots}, we present some
    numerical information rate lower bounds estimated for a QGEC
    and a variant of a QGEC (as depicted in Fig.~\ref{fig:QGEC}),
    equipped with ``trivial'' orthonormal ensemble and projective measurements.
Namely, the original channel in Fig.~\ref{fig:QGEC:plot:1}
and~\ref{fig:QGEC:plot:3} can be described by the CC-QSC
\begin{equation}
\label{eq:qgec:qsc:1}
\operator{N}^{y|x}(\rho_\system{S}) = 
\tr_\system{B}\!\left(\!
(V_\system{S}^\Herm V_\system{S}\!\tensor\! \braket{y})\cdot
\Phi^{\mathrm{CBF}}(\rho_\system{S}\!\tensor\!\braket{x})
\!\right),
\end{equation}
whereas that in Fig.~\ref{fig:QGEC:plot:2}
and~\ref{fig:QGEC:plot:4} is described by
\begin{equation}\hspace{0pt}
\label{eq:qgec:qsc:2}
\operator{N}^{y|x}\!(\rho_\system{S}) \!=\! 
\tr_\system{B}\!\!\left(\!
(\tilde{V}_\system{S}^\Herm \tilde{V}_\system{S}\!\tensor\!
    \bra{y}\hspace{-3pt}\ket{y})\!\cdot\!
(\id\!\tensor\!\Phi^{\mathrm{CBF}})(\rho_\system{S}\!\tensor\!
    \bra{x}\hspace{-3pt}\ket{x})\!
\right)\!,\hspace{-8pt}
\end{equation}
where $\{\bra{x}\}_{x\in\set{X}}$ and $\{\bra{y}\}_{y\in\set{Y}}$ are
some orthonormal basis of $\hilbert_{\system{A}}$ and $\hilbert_{\system{B}}$,
respectively.
In the latter case, the memory system $\system{S}$ is
extended as $\hilbert_{\system{S}} =
\hilbert_{\system{S}_{1}}
\tensor\hilbert_{\system{S}_{0}}$.
More specifically, in~\eqref{eq:qgec:qsc:2}, $\rho_\system{S}$ and
$\tilde{V}_\system{S}$ are operators on
$\hilbert_{\system{S}}$,
and $\Phi^{\mathrm{CBF}}$ acts on
$\DensOp(\hilbert_{\system{S}_{0}}\tensor
\hilbert_{\system{A}})$, and $\id$ is the identity map on
$\system{S}_{1}$.
For both scenarios, the input processes 
are binary symmetric i.i.d. processes,
\ie, $Q^{(n)}(\vx_1^n)\defeq 2^{-n}$ for all $\vx_1^n\in\{0,1\}^n$.
The lower bounds in those figures were obtained by minimizing the 
difference function $\Delta_{W}^{(n)}$ defined in~\eqref{eq:delta} w.r.t.
different classes of auxiliary channels (subject to certain time and threshold
constraints). For the case where the auxiliary channels are CC-QSCs, 
Alg.~\ref{alg:grad:1} was applied. For FSMC
auxiliary channels, we implemented the expectation-maximization type algorithm
in~\cite{sadeghi2009optimization} for comparison.
As already emphasized beforehand, these lower bounds represent rates that are
achievable with the help of a mismatched decoder~\cite{ganti2000mismatched}.
Fig.~\ref{fig:QGEC:plot:5} is an example illustrating the typical convergence 
time of different methods (including our own) for minimizing the difference
function.
In all of the above figures, $n=10^5$, and we have used Alg.~\ref{alg:SPA:2}
to \emph{estimate} the information rate. According to our experience, the error
of the estimation in this case lies within the line-width in the figures.
\section{Conclusion}\label{sec:7:conclusion}
In this article, we have considered the scenario of transmitting classical
information over a quantum channel with finite memory using separable-state
ensembles and local measurements. We defined the notion of CC-QSCs
as an equivalent way to describe such communication setups, and
demonstrated how NFGs can be used to visualize such channels.
We have shown that the
information rate of a quantum-state channel is independent of the initial
density operator under suitable conditions, and proposed algorithms for
estimating and bounding such information rate. The computations in such
algorithms can be carried out using the corresponding NFGs of the CC-QSC.
We emphasize that our approach for optimizing the lower bound is data-driven,
and does not require the knowledge of the true channel model.
\section*{Acknowledgment}
It is a great pleasure to acknowledge discussions on topics related to this
paper with Andi Loeliger.
\bibliographystyle{IEEEtran}
\bibliography{reference}

\begin{thebibliography}{10}
\providecommand{\url}[1]{#1}
\csname url@samestyle\endcsname
\providecommand{\newblock}{\relax}
\providecommand{\bibinfo}[2]{#2}
\providecommand{\BIBentrySTDinterwordspacing}{\spaceskip=0pt\relax}
\providecommand{\BIBentryALTinterwordstretchfactor}{4}
\providecommand{\BIBentryALTinterwordspacing}{\spaceskip=\fontdimen2\font plus
\BIBentryALTinterwordstretchfactor\fontdimen3\font minus
  \fontdimen4\font\relax}
\providecommand{\BIBforeignlanguage}[2]{{%
\expandafter\ifx\csname l@#1\endcsname\relax
\typeout{** WARNING: IEEEtran.bst: No hyphenation pattern has been}%
\typeout{** loaded for the language `#1'. Using the pattern for}%
\typeout{** the default language instead.}%
\else
\language=\csname l@#1\endcsname
\fi
#2}}
\providecommand{\BIBdecl}{\relax}
\BIBdecl

\bibitem{cao2017estimating}
M.~X. Cao and P.~O. Vontobel, ``Estimating the information rate of a channel
  with classical input and output and a quantum state,'' in \emph{Proc. IEEE
  Int. Symp. Inf. Theory}, Jun. 2017, pp. 3205--3209.

\bibitem{cao2019optimizing}
------, ``Optimizing bounds on the classical information rate of quantum
  channels with memory,'' in \emph{2019 IEEE International Symposium on
  Information Theory (ISIT)}.\hskip 1em plus 0.5em minus 0.4em\relax IEEE,
  2019, pp. 265--269.

\bibitem{bowen2004quantum}
G.~Bowen and S.~Mancini, ``Quantum channels with a finite memory,'' \emph{Phys.
  Rev. A}, vol.~69, no.~1, p. 012306, 2004.

\bibitem{kretschmann2005quantum}
D.~Kretschmann and R.~F. Werner, ``Quantum channels with memory,'' \emph{Phys.
  Rev. A}, vol.~72, no.~6, p. 062323, 2005.

\bibitem{caruso2014quantum}
F.~Caruso, V.~Giovannetti, C.~Lupo, and S.~Mancini, ``Quantum channels and
  memory effects,'' \emph{Rev. Mod. Phys.}, vol.~86, no.~4, p. 1203, 2014.

\bibitem{nielsen2011quantum}
M.~A. Nielsen and I.~L. Chuang, \emph{{Quantum Computation and Quantum
  Information}}, {10th Anniversary}~ed.\hskip 1em plus 0.5em minus 0.4em\relax
  Cambridge University Press, 2011.

\bibitem{wilde2017quantum}
M.~M. Wilde, \emph{{Quantum Information Theory}}, 2nd~ed.\hskip 1em plus 0.5em
  minus 0.4em\relax Cambridge University Press, 2017.

\bibitem{bose2003quantum}
S.~Bose, ``Quantum communication through an unmodulated spin chain,''
  \emph{Phys. Rev. Lett.}, vol.~91, no.~20, p. 207901, 2003.

\bibitem{ball2004exploiting}
J.~Ball, A.~Dragan, and K.~Banaszek, ``Exploiting entanglement in communication
  channels with correlated noise,'' \emph{Phys. Rev. A}, vol.~69, no.~4, p.
  042324, 2004.

\bibitem{shannon2001mathematical}
C.~E. Shannon, ``A mathematical theory of communication,'' \emph{Bell System
  Technical Journal}, vol.~27, no.~3, pp. 379--423, 1948.

\bibitem{cover2012elements}
T.~M. Cover and J.~A. Thomas, \emph{{Elements of Information Theory}},
  2nd~ed.\hskip 1em plus 0.5em minus 0.4em\relax John Wiley \& Sons, 2006.

\bibitem{gallager1968information}
R.~G. Gallager, \emph{{Information Theory and Reliable Communication}},
  1st~ed.\hskip 1em plus 0.5em minus 0.4em\relax Wiley, 1968.

\bibitem{arnold2006simulation}
D.-M. Arnold, H.-A. Loeliger, P.~O. Vontobel, A.~Kavčić, and W.~Zeng,
  ``Simulation-based computation of information rates for channels with
  memory,'' \emph{IEEE Trans. Inf. Theory}, vol.~52, no.~8, pp. 3498--3508,
  2006.

\bibitem{sadeghi2009optimization}
P.~Sadeghi, P.~O. Vontobel, and R.~Shams, ``Optimization of information rate
  upper and lower bounds for channels with memory,'' \emph{IEEE Trans. Inf.
  Theory}, vol.~55, no.~2, pp. 663--688, 2009.

\bibitem{sharma2001entropy}
V.~Sharma and S.~Singh, ``Entropy and channel capacity in the regenerative
  setup with applications to {M}arkov channels,'' in \emph{Proc. IEEE Int.
  Symp. Inf. Theory}, Jun. 2001, p. 283.

\bibitem{pfister2001achievable}
H.~D. Pfister, J.~B. Soriaga, and P.~H. Siegel, ``On the achievable information
  rates of finite state {ISI} channels,'' in \emph{Proc. IEEE Global Telecom.
  Conf.}, vol.~5, Nov. 2001, pp. 2992--2996.

\bibitem{loeliger2017factor}
H.-A. Loeliger and P.~O. Vontobel, ``Factor graphs for quantum probabilities,''
  \emph{IEEE Trans. Inf. Theory}, vol.~63, no.~9, pp. 5642--5665, 2017.

\bibitem{ganti2000mismatched}
A.~Ganti, A.~Lapidoth, and I.~E. Telatar, ``Mismatched decoding revisited:
  General alphabets, channels with memory, and the wide-band limit,''
  \emph{IEEE Trans. Inf. Theory}, vol.~46, no.~7, pp. 2315--2328, 2000.

\bibitem{arimoto1972algorithm}
S.~Arimoto, ``An algorithm for computing the capacity of arbitrary discrete
  memoryless channels,'' \emph{IEEE Trans. Inf. Theory}, vol.~18, no.~1, pp.
  14--20, 1972.

\bibitem{blahut1972computation}
R.~Blahut, ``Computation of channel capacity and rate-distortion functions,''
  \emph{IEEE Trans. Inf. Theory}, vol.~18, no.~4, pp. 460--473, 1972.

\bibitem{vontobel2008generalization}
P.~O. Vontobel, A.~Kavčić, D.~M. Arnold, and H.-A. Loeliger, ``A
  generalization of the {Blahut--Arimoto} algorithm to finite-state channels,''
  \emph{IEEE Trans. Inf. Theory}, vol.~54, no.~5, pp. 1887--1918, 2008.

\bibitem{hastings2009superadditivity}
M.~B. Hastings, ``Superadditivity of communication capacity using entangled
  inputs,'' \emph{Nature Physics}, vol.~5, no.~4, p. 255, 2009.

\bibitem{macchiavello2004transition}
C.~Macchiavello, G.~M. Palma, and S.~Virmani, ``Transition behavior in the
  channel capacity of two-quibit channels with memory,'' \emph{Phys. Rev. A},
  vol.~69, no.~1, p. 010303, 2004.

\bibitem{karimipour2006entanglement}
V.~Karimipour and L.~Memarzadeh, ``Entanglement and optimal strings of qubits
  for memory channels,'' \emph{Phys. Rev. A}, vol.~74, no.~6, p. 062311, 2006.

\bibitem{lupo2010transitional}
C.~Lupo and S.~Mancini, ``Transitional behavior of quantum {G}aussian memory
  channels,'' \emph{Phys. Rev. A}, vol.~81, no.~5, p. 052314, 2010.

\bibitem{kschischang2001factor}
F.~R. Kschischang, B.~J. Frey, and H.-A. Loeliger, ``Factor graphs and the
  sum-product algorithm,'' \emph{IEEE Trans. Inf. Theory}, vol.~47, no.~2, pp.
  498--519, 2001.

\bibitem{forney2001codes}
G.~D. Forney, ``Codes on graphs: Normal realizations,'' \emph{IEEE Trans. Inf.
  Theory}, vol.~47, no.~2, pp. 520--548, 2001.

\bibitem{loeliger2004introduction}
H.-A. Loeliger, ``An introduction to factor graphs,'' \emph{IEEE Signal
  Process. Mag.}, vol.~21, no.~1, pp. 28--41, 2004.

\bibitem{cao2017double}
M.~X. Cao and P.~O. Vontobel, ``Double-edge factor graphs: definition,
  properties, and examples,'' in \emph{Proc. IEEE Inf. Theory Workshop}, Jun.
  2017, pp. 136--140.

\bibitem{mushkin1989capacity}
M.~Mushkin and I.~Bar-David, ``Capacity and coding for the {Gilbert--Elliott}
  channels,'' \emph{IEEE Trans. Inf. Theory}, vol.~35, no.~6, pp. 1277--1290,
  1989.

\bibitem{ephraim2002hidden}
Y.~Ephraim and N.~Merhav, ``Hidden {M}arkov processes,'' \emph{IEEE Trans. Inf.
  Theory}, vol.~48, no.~6, pp. 1518--1569, 2002.

\bibitem{bahl1974optimal}
L.~Bahl, J.~Cocke, F.~Jelinek, and J.~Raviv, ``Optimal decoding of linear codes
  for minimizing symbol error rate,'' \emph{IEEE Trans. Inf. Theory}, vol.~20,
  no.~2, pp. 284--287, 1974.

\bibitem{jamiolkowski1972linear}
A.~Jamio{\l}kowski, ``Linear transformations which preserve trace and positive
  semidefiniteness of operators,'' \emph{Rep. Math. Phys.}, vol.~3, no.~4, pp.
  275--278, 1972.

\bibitem{bowen2005bounds}
G.~Bowen, I.~Devetak, and S.~Mancini, ``Bounds on classical information
  capacities for a class of quantum memory channels,'' \emph{Phys. Rev. A},
  vol.~71, no.~3, p. 034310, 2005.

\bibitem{fannes1973continuity}
M.~Fannes, ``A continuity property of the entropy density for spin lattice
  systems,'' \emph{Comm. Math. Phys.}, vol.~31, no.~4, pp. 291--294, 1973.

\end{thebibliography}
\onecolumn
\appendix
\begin{figure}[h!]
    \centering
    
    \begin{tikzpicture}[scale=\scalefactorA,every node/.style={transform shape},
    factor/.style={rectangle, minimum width=1cm, minimum height=.7cm, draw},
    sfactor/.style={rectangle, minimum size=.4cm, draw},
    darksolid/.style={rectangle, minimum size=.15cm, draw,fill = black,
    inner sep=0pt, outer sep = 0pt},
    label/.style={magenta,anchor=north east,xshift = .1cm}]
\node[darksolid] (S) {}; \node [left=0pt of S] {$\cs_0$};
\node[factor] (E1) [right=.7cm of S] {$W$};
\draw (S) -- (E1);
\node[darksolid] (X1) [above=.8cm of E1] {}; \node[above = 0pt of X1] {$\cx_1$};
\draw (X1) -- (E1);
\draw (E1.south) -- ([yshift=-.8cm]E1.south) node (Y1) [right] {$y_1$};
    
\node[factor] (E2) [right=1cm of E1] {$W$};
\draw (E1) -- (E2) node[above,midway] {$s_1$};
\node[darksolid] (X2) [above=.8cm of E2] {}; \node[above = 0pt of X2] {$\cx_2$};
\draw (X2) -- (E2);
\draw (E2.south) -- ([yshift=-.8cm]E2.south) node[right] {$y_2$};
    
\node[factor, draw=none] (Edummy1) [right=1cm of E2] {$\cdots$};
\draw (E2) -- (Edummy1) node[above,midway] {$s_2$};
\node at (X1-|Edummy1) {$\cdots$};
\node at (Y1-|Edummy1) {$\cdots$};
    
\node[factor] (El) [right=1cm of Edummy1] {$W$};
\draw (Edummy1) -- (El) node[above,midway] {$\cs_{\ell-1}$};
\node[darksolid] (Xl) [above=.8cm of El] {};
\node[above = 0pt of Xl] {$\cx_\ell$};
\draw (Xl) -- (El);
\draw (El.south) -- ([yshift=-.8cm]El.south) node[right] {$y_\ell$};
    
\node[factor] (El2) [right=1cm of El] {$W$};
\draw (El) -- (El2) node[above,midway] {$s_\ell$};
\node[darksolid] (Xl2) [above=.8cm of El2] {};
    \node[above = 0pt of Xl2] {$\cx_{\ell\!+\!1}$};
\draw (Xl2) -- (El2);
\draw (El2.south) -- ([yshift=-.8cm]El2.south) node[right] {$y_{\ell\!+\!1}$};
    
\node[factor, draw=none] (Edummy2) [right=1cm of El2] {$\cdots$};
\draw (El2) -- (Edummy2) node[above,midway] {$s_{\ell\!+\!1}$};
\node at (X1-|Edummy2) {$\cdots$};
\node at (Y1-|Edummy2) {$\cdots$};
    
\node[factor] (En) [right=1cm of Edummy2] {$W$};
\draw (Edummy2) -- (En) node[above,pos=0.4] {$s_{n\!-\!1}$};
\node[darksolid] (Xn) [above=.8cm of En] {}; \node[above = 0pt of Xn] {$\cx_n$};
\draw (Xn) -- (En);
\draw (En.south) -- ([yshift=-.8cm]En.south) node[right] {$y_{n}$};
	
\node[sfactor, inner sep=0pt] (ee) [right=.5cm of En] {$\mathbf{1}$};
\draw (En) -- (ee) node[above,midway] {$s_n$};
    
\begin{pgfonlayer}{bg}
    \draw[dashed, cyan, line width=1.5pt, fill=cyan!10]
        ([xshift=-2.1cm,yshift=2.8cm]E1) rectangle
        ([xshift=1.5cm,yshift=-2.4cm]ee);
    \draw[dashed, cyan, line width=1.5pt, fill=cyan!20]
        ([xshift=-.6cm,yshift=2.6cm]E1) rectangle
        ([xshift=1.3cm,yshift=-2.2cm]ee);
    \draw[dashed, cyan, line width=1.5pt, fill=cyan!30]
    ([xshift=-.7cm,yshift=2.4cm]E2) rectangle
    ([xshift=1.1cm,yshift=-2cm]ee);
    \draw[dashed, cyan, line width=1.5pt, fill=cyan!40]
        ([xshift=-.6cm,yshift=2.2cm]El) rectangle
        ([xshift=.9cm,yshift=-1.8cm]ee);
    \draw[dashed, cyan, line width=1.5pt, fill=cyan!50]
        ([xshift=-.7cm,yshift=2cm]El2) rectangle
        ([xshift=.7cm,yshift=-1.6cm]ee);
    \draw[dashed, cyan, line width=1.5pt, fill=cyan!60]
        ([xshift=-.7cm,yshift=1.8cm]En) rectangle
        ([xshift=.5cm,yshift=-1.4cm]ee);
\end{pgfonlayer}
\end{tikzpicture}
    \caption{Verification of~\eqref{eq:verfy:FSMC:conditional:distribution}.
        Note that every closing-the-box operation yields a function node
        representing the constant function~$1$.}
    \label{fig:FMSC:high:level:2}
\end{figure}
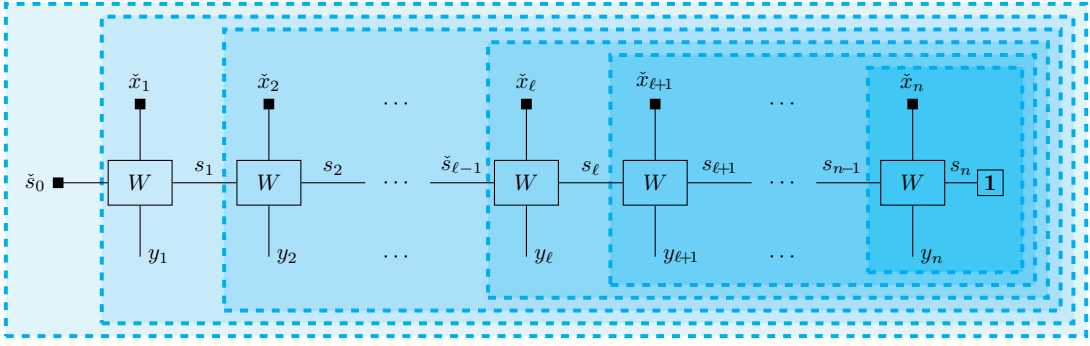
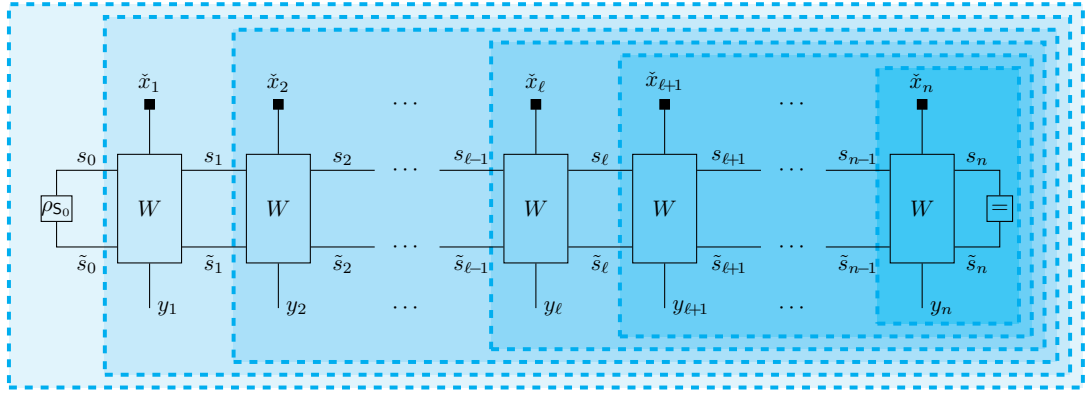
\begin{figure}[h!]
  \centering
  
    \begin{tikzpicture}[scale=\scalefactorA,every node/.style={transform shape},
    factor/.style={rectangle, minimum width=1cm, minimum height=1.7cm, draw},
    sfactor/.style={rectangle, minimum size=.4cm, draw},
    darksolid/.style={rectangle, minimum size=.15cm, draw,fill = black,
    inner sep=0pt, outer sep = 0pt},
    label/.style={magenta,anchor=north east,xshift = .1cm}]
\node[sfactor,inner sep=0pt] (S) {$\rho_{\system{S}_0}$};
\node[factor] (E1) [right=.7cm of S] {$W$};
\draw (S.north) |- ([yshift=.6cm]E1.west) node[above=-0.05cm,pos=.75] {$s_0$};
\draw (S.south) |- ([yshift=-.6cm]E1.west) node[below,pos=.75] {$\tilde{s}_0$};
\node[darksolid] (X1) [above=.7cm of E1] {}; \node[above = 0pt of X1] {$\cx_1$};
\draw (X1) -- (E1);
\draw (E1.south) -- ([yshift=-.7cm]E1.south) node (Y1) [right] {$y_1$};
    
\node[factor] (E2) [right=1cm of E1] {$W$};
\draw ([yshift=.6cm]E1.east) |- ([yshift=.6cm]E2.west) 
    node[above=-0.05cm,pos=.75] {$s_1$};
\draw ([yshift=-.6cm]E1.east) |- ([yshift=-.6cm]E2.west) 
    node[below,pos=.75] {$\tilde{s}_1$};
\node[darksolid] (X2) [above=.7cm of E2] {}; \node[above = 0pt of X2] {$\cx_2$};
\draw (X2) -- (E2);
\draw (E2.south) -- ([yshift=-.7cm]E2.south) node[right] {$y_2$};
    
\node[factor, draw=none] (Edummy1) [right=1cm of E2] {};
\node at ([yshift=.6cm]E2.east-|Edummy1) {$\cdots$};
\node at ([yshift=-.6cm]E2.east-|Edummy1) {$\cdots$};
\draw ([yshift=.6cm]E2.east) |- ([yshift=.6cm]Edummy1.west)
    node[above=-0.05cm,pos=.75] {$s_2$};
\draw ([yshift=-.6cm]E2.east) |- ([yshift=-.6cm]Edummy1.west) 
    node[below,pos=.75] {$\tilde{s}_2$};
\node at (X1-|Edummy1) {$\cdots$};
\node at (Y1-|Edummy1) {$\cdots$};
    
\node[factor] (El) [right=1cm of Edummy1] {$W$};
\draw ([yshift=.6cm]Edummy1.east) |- ([yshift=.6cm]El.west)
    node[above=-0.05cm,pos=.75] {$s_{\ell\!-\!1}$};
\draw ([yshift=-.6cm]Edummy1.east) |- ([yshift=-.6cm]El.west)
    node[below,pos=.75] {$\tilde{s}_{\ell\!-\!1}$};
\node[darksolid] (Xl)[above=.7cm of El] {};\node[above = 0pt of Xl]{$\cx_\ell$};
\draw (Xl) -- (El);
\draw (El.south) -- ([yshift=-.7cm]El.south) node[right] {$y_\ell$};
    
\node[factor] (El2) [right=1cm of El] {$W$};
\draw ([yshift=.6cm]El.east) |- ([yshift=.6cm]El2.west)
    node[above=-0.05cm,pos=.75] {$s_\ell$};
\draw ([yshift=-.6cm]El.east) |- ([yshift=-.6cm]El2.west)
    node[below,pos=.75] {$\tilde{s}_\ell$};
\node[darksolid] (Xl2) [above=.7cm of El2] {};
    \node[above = 0pt of Xl2] {$\cx_{\ell\!+\!1}$};
\draw (Xl2) -- (El2);
\draw (El2.south) -- ([yshift=-.7cm]El2.south) node[right] {$y_{\ell\!+\!1}$};
    
\node[factor, draw=none] (Edummy2) [right=1cm of El2] {};
\node at ([yshift=.6cm]El2.east-|Edummy2) {$\cdots$};
\node at ([yshift=-.6cm]El2.east-|Edummy2) {$\cdots$};
\draw ([yshift=.6cm]El2.east) |- ([yshift=.6cm]Edummy2.west)
    node[above=-0.05cm,pos=.75] {$s_{\ell\!+\!1}$};
\draw ([yshift=-.6cm]El2.east) |- ([yshift=-.6cm]Edummy2.west)
    node[below,pos=.75] {$\tilde{s}_{\ell\!+\!1}$};
\node at (X1-|Edummy2) {$\cdots$};
\node at (Y1-|Edummy2) {$\cdots$};
    
\node[factor] (En) [right=1cm of Edummy2] {$W$};
\draw ([yshift=.6cm]Edummy2.east) |- ([yshift=.6cm]En.west)
    node[above=-0.05cm,pos=.75] {$s_{n\!-\!1}$};
\draw ([yshift=-.6cm]Edummy2.east) |- ([yshift=-.6cm]En.west)
    node[below,pos=.75] {$\tilde{s}_{n\!-\!1}$};
\node[darksolid] (Xn) [above=.7cm of En] {}; \node[above = 0pt of Xn] {$\cx_n$};
\draw (Xn) -- (En);
\draw (En.south) -- ([yshift=-.7cm]En.south) node[right] {$y_n$};
    
\node[sfactor, inner sep=0pt] (ee) [right=.5cm of En] {$=$};
\draw ([yshift=.6cm]En.east) -| (ee.north) node[above=-0.05cm,pos=.25] {$s_n$};
\draw ([yshift=-.6cm]En.east) -| (ee.south) node[below,pos=.25] {$\tilde{s}_n$};
    
\begin{pgfonlayer}{bg}
    \draw[dashed, cyan, line width=1.5pt, fill=cyan!10]
        ([xshift=-2.2cm,yshift=3.2cm]E1) rectangle
        ([xshift=1.3cm,yshift=-2.8cm]ee);
    \draw[dashed, cyan, line width=1.5pt, fill=cyan!20]
        ([xshift=-.7cm,yshift=3cm]E1) rectangle
        ([xshift=1.1cm,yshift=-2.6cm]ee);
    \draw[dashed, cyan, line width=1.5pt, fill=cyan!30]
        ([xshift=-.7cm,yshift=2.8cm]E2) rectangle
        ([xshift=.9cm,yshift=-2.4cm]ee);
    \draw[dashed, cyan, line width=1.5pt, fill=cyan!40]
    ([xshift=-.7cm,yshift=2.6cm]El) rectangle 
    ([xshift=.7cm,yshift=-2.2cm]ee);
    \draw[dashed, cyan, line width=1.5pt, fill=cyan!50]
    ([xshift=-.7cm,yshift=2.4cm]El2) rectangle
    ([xshift=.5cm,yshift=-2cm]ee);
    \draw[dashed, cyan, line width=1.5pt, fill=cyan!60]
    ([xshift=-.7cm,yshift=2.2cm]En) rectangle
    ([xshift=.3cm,yshift=-1.8cm]ee);
\end{pgfonlayer}
\end{tikzpicture}
  \caption{Counterpart of Fig.~\ref{fig:FMSC:high:level:2} for QSCs.
    Note that every closing-the-box operation yields a function node representing
    a Kronecker-delta function node, i.e., a degree-two equality function node.}
  \label{fig:QFSM:closing:the:box}
\end{figure}
\begin{figure}[h!]
    \centering
    
    \begin{tikzpicture}[scale=\scalefactorA,every node/.style={transform shape},
    node/.style={draw=none},
    factor/.style={rectangle, minimum width=1cm, minimum height=.7cm, draw},
    sfactor/.style={rectangle, minimum size=.4cm, draw},
    darksolid/.style={rectangle, minimum size=.15cm, draw, fill = black,
        inner sep=0pt, outer sep = 0pt}]
\node[sfactor,inner sep=0pt] (S) {$P_{\rv{S}_0}$};
\node[factor] (E1) [right=.7cm of S] {$W$};
\draw (S) -- (E1) node[above,midway] {$s_0$};
\draw (E1.north) -- ([yshift=.8cm]E1.north) node (X1) [darksolid]{} 
    node[right] {$\cx_1$};
\draw (E1.south) -- ([yshift=-.8cm]E1.south) node (Y1) [darksolid]{} 
    node[right] {$\cy_1$};

\node[factor, draw=none] (Edummy1) [right=1cm of E1] {$\cdots$};
\draw (E1) -- (Edummy1) node[above,midway] {$s_1$};
\node at (X1-|Edummy1) {$\cdots$};
\node at (Y1-|Edummy1) {$\cdots$};

\node[factor] (El0) [right=1cm of Edummy1] {$W$};
\draw (Edummy1) -- (El0) node[above,midway] {$s_{\ell-2}$};
\draw (El0.north) -- ([yshift=.8cm]El0.north) node (Xl0) [darksolid]{} 
    node[right] {$\cx_{\ell\!-\!1}$};
\draw (El0.south) -- ([yshift=-.8cm]El0.south)
    node[darksolid]{} node[right] {$\cy_{\ell\!-\!1}$};
    
\node[factor] (El) [right=1cm of El0] {$W$};
\draw (El0) -- (El) node[above,pos=.65] {$s_{\ell\!-\!1}$};
\draw (El.north) -- ([yshift=.8cm]El.north) node (Xl) [darksolid]{} 
    node[right] {$\cx_\ell$};
\draw (El.south) -- ([yshift=-1.8cm]El.south) node[right] {$y_\ell$};
    
\node[factor] (El2) [right=1cm of El] {$W$};
\draw (El) -- (El2) node[above,midway] {$s_\ell$};
\draw (El2.north) -- ([yshift=.8cm]El2.north) node (Xl2) [darksolid]{} 
    node[right] {$\cx_{\ell\!+\!1}$};
\draw (El2.south) -- ([yshift=-.8cm]El2.south) node[right] {$y_{\ell\!+\!1}$};
	
\node[factor, draw=none] (Edummy2) [right=1cm of El2] {$\cdots$};
\draw (El2) -- (Edummy2) node[above,midway] {$s_{\ell+1}$};
\node at (X1-|Edummy2) {$\cdots$};
\node at (Y1-|Edummy2) {$\cdots$};
    
\node[factor] (En) [right=1cm of Edummy2] {$W$};
\draw (Edummy2) -- (En) node[above,midway] {$s_{n\!-\!1}$};
\draw (En.north) -- ([yshift=.8cm]En.north) node (Xn) [darksolid]{} 
    node[right] {$\cx_n$};
\draw (En.south) -- ([yshift=-.8cm]En.south) node[right] {$y_n$};
    
\node[sfactor, inner sep=0pt] (ee) [right=.5cm of En] {$\mathbf{1}$};
\draw (En) -- (ee) node[above,midway] {$s_n$};
    
\begin{pgfonlayer}{bg}
    \draw[dashed, black, line width=1.5pt, fill=yellow!20]
        ([xshift=-.8cm,yshift=1.6cm]S) rectangle 
        ([xshift=.6cm,yshift=-1.6cm]ee);
    \node[anchor=north west] at ([xshift=-.6cm,yshift=-1.6cm]S|-ee)
        {$P_{\rv{Y}_\ell,\rv{Y}_1^{\ell\!-\!1}|\rv{X}_1^\ell}
        (y_\ell,\cvy_1^{\ell\!-\!1}|\cvx_1^\ell)$};
    \draw[dashed, cyan, line width=1.5pt, fill=cyan!20]
        ([xshift=-.7cm,yshift=1.4cm]El2) rectangle 
        ([xshift=.4cm,yshift=-1.4cm]ee);
    \node[anchor=south east, fill=cyan!50] at ([xshift=.4cm,yshift=-1.4cm]ee)
        {$\mathbf{1}$};
    \draw[dashed, magenta, line width=1.5pt, fill=magenta!20]
    ([xshift=-.6cm,yshift=1.4cm]S) rectangle 
    ([xshift=.8cm,yshift=-1.4cm]El0);
\end{pgfonlayer}
\end{tikzpicture}
    \caption{Efficient simulation of the channel output at step $\ell$ given 
        channel input $\cvx_1^n$ and channel output $\cvy_1^{\ell-1}$ for
        an FSMC.}
    \label{fig:CFSM:channel:simulation:Y}
\end{figure}
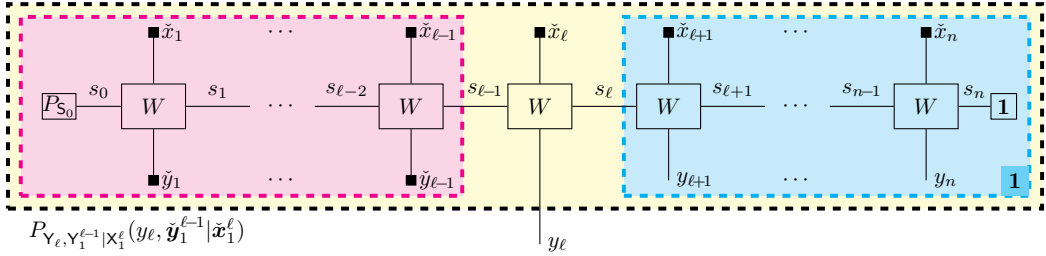
\begin{figure}[h!]
  \centering
  
    \begin{tikzpicture}[scale=\scalefactorA,every node/.style={transform shape},
    factor/.style={rectangle, minimum width=1cm, minimum height=1.7cm, draw},
    sfactor/.style={rectangle, minimum size=.4cm, draw},
    darksolid/.style={rectangle, minimum size=.15cm, draw,fill = black,
    inner sep=0pt, outer sep = 0pt}]
\node[sfactor,inner sep=0pt] (S) {$\rho_{\system{S}_0}$};
\node[factor] (E1) [right=.7cm of S] {$W$};
\draw (S.north) |- ([yshift=.6cm]E1.west) node[above=-0.05cm,pos=.75] {$s_0$};
\draw (S.south) |- ([yshift=-.6cm]E1.west) node[below,pos=.75] {$\tilde{s}_0$};
\draw (E1.north) -- ([yshift=.8cm]E1.north) node (X1) [darksolid]{} 
    node[right] {$\cx_1$};
\draw (E1.south) -- ([yshift=-.8cm]E1.south) node[darksolid]{}
    node (Y1) [right] {$\cy_1$};
    
\node[factor, draw=none] (Edummy1) [right=1cm of E1] {};
\node at ([yshift=.6cm]E1.east-|Edummy1) {$\cdots$};
\node at ([yshift=-.6cm]E1.east-|Edummy1) {$\cdots$};
\draw ([yshift=.6cm]E1.east) |- ([yshift=.6cm]Edummy1.west)
    node[above=-0.05cm,pos=.75] {$s_1$};
\draw ([yshift=-.6cm]E1.east) |- ([yshift=-.6cm]Edummy1.west)
    node[below,pos=.75] {$\tilde{s}_1$};
\node at (X1-|Edummy1) {$\cdots$};
\node at (Y1-|Edummy1) {$\cdots$};
    
\node[factor] (El0) [right=1cm of Edummy1] {$W$};
\draw ([yshift=.6cm]Edummy1.east) |- ([yshift=.6cm]El0.west)
    node[above=-0.05cm,pos=.75] {$s_{\ell\!-\!2}$};
\draw ([yshift=-.6cm]Edummy1.east) |- ([yshift=-.6cm]El0.west)
    node[below,pos=.75] {$\tilde{s}_{\ell\!-\!2}$};
\draw (El0.north) -- ([yshift=.8cm]El0.north) node (Xl0) [darksolid]{} 
    node[right] {$\cx_{\ell\!-\!1}$};
\draw (El0.south) -- ([yshift=-.8cm]El0.south) node[darksolid]{}
    node[right] {$\cy_{\ell\!-\!1}$};
    
\node[factor] (El) [right=1cm of El0] {$W$};
\draw ([yshift=.6cm]El0.east) |- ([yshift=.6cm]El.west)
    node[above=-0.05cm,pos=.85] {$s_{\ell\!-\!1}$};
\draw ([yshift=-.6cm]El0.east) |- ([yshift=-.6cm]El.west)
    node[below,pos=.85] {$\tilde{s}_{\ell\!-\!1}$};
\draw (El.north) -- ([yshift=.8cm]El.north) node (Xl) [darksolid]{} 
    node[right] {$\cx_\ell$};
\draw (El.south) -- ([yshift=-2.2cm]El.south) node[right] {$y_\ell$};
    
\node[factor] (El2) [right=1cm of El] {$W$};
\draw ([yshift=.6cm]El.east) |- ([yshift=.6cm]El2.west)
    node[above=-0.05cm,pos=.75] {$s_\ell$};
\draw ([yshift=-.6cm]El.east) |- ([yshift=-.6cm]El2.west)
    node[below,pos=.75] {$\tilde{s}_\ell$};
\draw (El2.north) -- ([yshift=.8cm]El2.north) node (Xl2) [darksolid]{} 
    node[right] {$\cx_{\ell\!+\!1}$};
\draw (El2.south) -- ([yshift=-.8cm]El2.south) node[right] {$y_{\ell\!+\!1}$};
    
\node[factor, draw=none] (Edummy2) [right=1cm of El2] {};
\node at ([yshift=.6cm]El2.east-|Edummy2) {$\cdots$};
\node at ([yshift=-.6cm]El2.east-|Edummy2) {$\cdots$};
\draw ([yshift=.6cm]El2.east) |- ([yshift=.6cm]Edummy2.west)
    node[above=-0.05cm,pos=.75] {$s_{\ell\!+\!1}$};
\draw ([yshift=-.6cm]El2.east) |- ([yshift=-.6cm]Edummy2.west)
    node[below,pos=.75] {$\tilde{s}_{\ell\!+\!1}$};
\node at (X1-|Edummy2) {$\cdots$};
\node at (Y1-|Edummy2) {$\cdots$};
    
\node[factor] (En) [right=1cm of Edummy2] {$W$};
\draw ([yshift=.6cm]Edummy2.east) |- ([yshift=.6cm]En.west)
    node[above=-0.05cm,pos=.75] {$s_{n\!-\!1}$};
\draw ([yshift=-.6cm]Edummy2.east) |- ([yshift=-.6cm]En.west)
    node[below,pos=.75] {$\tilde{s}_{n\!-\!1}$};
\draw (En.north) -- ([yshift=.8cm]En.north) node (Xn) [darksolid]{} 
    node[right] {$\cx_n$};
\draw (En.south) -- ([yshift=-.8cm]En.south) node[right] {$y_n$};
    
\node[sfactor, inner sep=0pt] (ee) [right=.5cm of En] {$=$};
\draw ([yshift=.6cm]En.east) -| (ee.north) node[above=-0.05cm,pos=.25] {$s_n$};
\draw ([yshift=-.6cm]En.east) -| (ee.south) node[below,pos=.25] {$\tilde{s}_n$};
    
\begin{pgfonlayer}{bg}
    \draw[dashed, black, line width=1.5pt, fill=yellow!20]
        ([xshift=-.8cm,yshift=2.4cm]S) rectangle 
        ([xshift=.8cm,yshift=-2.3cm]ee);
    \node[anchor=north west] at ([xshift=-.8cm,yshift=-2.3cm]S|-ee)
        {$P_{\rv{Y}_\ell,\rv{Y}_1^{\ell\!-\!1}|\rv{X}_1^\ell,\system{S}}
        (y_\ell,\cvy_1^{\ell\!-\!1}|\cvx_1^\ell;\rho_{\system{S}_0})$};
    \draw[dashed, cyan, line width=1.5pt, fill=cyan!20]
        ([xshift=-.7cm,yshift=2.2cm]El2) rectangle
        ([xshift=.6cm,yshift=-2.1cm]ee);
    \node[anchor=south east, fill=cyan!50] at ([xshift=.6cm,yshift=-2.1cm]ee)
        {$\delta$};
    \draw[dashed, magenta, line width=1.5pt, fill=magenta!20]
        ([xshift=-.6cm,yshift=2.2cm]S) rectangle 
        ([xshift=.8cm,yshift=-2.1cm]El0);
\end{pgfonlayer}
\end{tikzpicture}
  \caption{Efficient simulation of the channel output at step $\ell$ given 
    channel input $\cvx_1^n$ and channel output $\cvy_1^{\ell-1}$ for a QSC.}
  \label{fig:QFSM:channel:simulation:Y}
\end{figure}
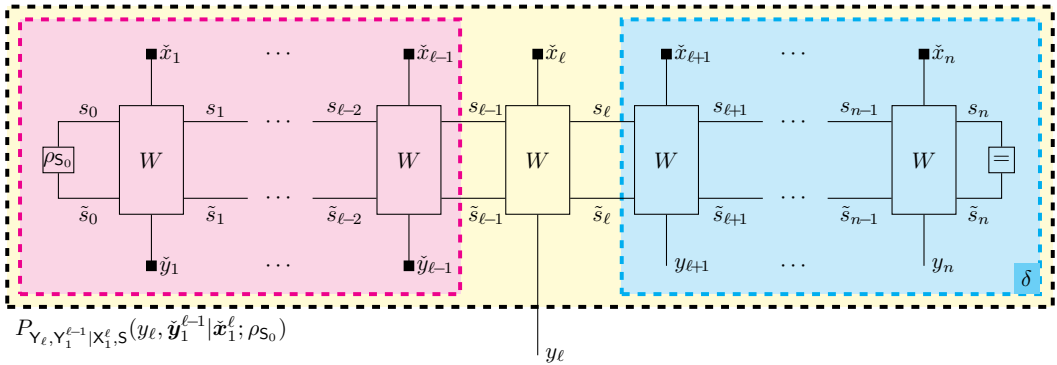
\begin{figure}[h!]
  \centering
  
    \begin{tikzpicture}[scale=\scalefactorB,every node/.style={transform shape},
    factor/.style={rectangle, minimum width=1cm, minimum height=.7cm, draw},
    sfactor/.style={rectangle, minimum size=.5cm, draw},
    darksolid/.style={rectangle, minimum size=.15cm, draw,fill = black,
    inner sep=0pt, outer sep = 0pt},
    label/.style={magenta,anchor=north east,xshift = .1cm}]
\node[sfactor,inner sep=0pt] (S) {$P_{\rv{S}_0}$};
\node[factor] (E1) [right=.7cm of S] {$W$};
\draw (S) -- (E1) node[above,midway] {$s_0$};
\node[sfactor] (X1) [above=.8cm of E1] {$Q$};
\draw (X1) -- (E1) node[right, midway] {$x_1$};
\draw (E1.south) -- ([yshift=-.8cm]E1.south) node (Y1) [darksolid]{}
    node[right] {$\cy_1$};
    
\node[factor] (E2) [right=1cm of E1] {$W$};
\draw (E1) -- (E2) node[above,midway] {$s_1$};
\node[sfactor] (X2) [above=.8cm of E2] {$Q$};
\draw (X2) -- (E2) node[right, midway] {$x_2$};
\draw (E2.south) -- ([yshift=-.8cm]E2.south) node[darksolid]{}
    node[right] {$\cy_2$};
    
\node[factor, draw=none] (Edummy1) [right=1cm of E2] {$\cdots$};
\draw (E2) -- (Edummy1) node[above,midway] {$s_2$};
\node at (X1-|Edummy1) {$\cdots$};
\node at (Y1-|Edummy1) {$\cdots$};
    
\node[factor] (El) [right=1cm of Edummy1] {$W$};
\draw (Edummy1) -- (El) node[above,midway] {$s_{\ell-1}$};
\node[sfactor] (Xl) [above=.8cm of El] {$Q$};
\draw (Xl) -- (El) node[right, midway] {$x_\ell$};
\draw (El.south) -- ([yshift=-.8cm]El.south) node[darksolid]{}
    node[right] {$\cy_\ell$};
    
    
\node[factor, draw=none] (Edummy2) [right=1cm of El] {$\cdots$};
\draw (El) -- (Edummy2) node[above,pos=.7] {$s_{\ell\!+\!1}$};
\node at (X1-|Edummy2) {$\cdots$};
\node at (Y1-|Edummy2) {$\cdots$};
    
\node[factor] (En) [right=1cm of Edummy2] {$W$};
\draw (Edummy2) -- (En) node[above,midway] {$s_{n\!-\!1}$};
\node[sfactor] (Xn) [above=.8cm of En] {$Q$};
\draw (Xn) -- (En) node[right, midway] {$x_n$};
\draw (En.south) -- ([yshift=-.8cm]En.south) node[darksolid]{}
    node[right] {$\cy_{n}$};
    
\node[sfactor, inner sep=0pt] (ee) [right=1cm of En] {$\mathbf{1}$};
\draw (En) -- (ee) node[above,midway] {$s_n$};
    
\begin{pgfonlayer}{bg}
    \draw[dashed, magenta, line width=1.5pt, fill=magenta!15]
        ([xshift=-1.2cm,yshift=2.5cm]S) rectangle
        ([xshift=.7cm,yshift=-3.2cm]En);
    \node[label] at ([xshift=.7cm,yshift=-3.2cm]En) {$\muY_{n}$};
    \draw[dashed, magenta, line width=1.5pt, fill=magenta!30]
        ([xshift=-1cm,yshift=2.3cm]S) rectangle
        ([xshift=.7cm,yshift=-2.6cm]El);
    \node[label] at ([xshift=.7cm,yshift=-2.6cm]El) {$\muY_{\ell}$};
    \draw[dashed, magenta, line width=1.5pt, fill=magenta!45]
        ([xshift=-.8cm,yshift=2.1cm]S) rectangle
        ([xshift=.7cm,yshift=-2cm]E2);
    \node[label] at ([xshift=.7cm,yshift=-2cm]E2) {$\muY_2$};
    \draw[dashed, magenta, line width=1.5pt, fill=magenta!60]
        ([xshift=-.6cm,yshift=1.9cm]S) rectangle
        ([xshift=.7cm,yshift=-1.4cm]E1);
    \node[label] at ([xshift=.7cm,yshift=-1.4cm]E1) {$\muY_1$};
\end{pgfonlayer}
\end{tikzpicture}
  \caption{The iterative computation of $\muY_{\ell}$ as 
    in~\eqref{eq:recursive:state:metric:Y:1} can be understood as a sequence
    of CTB operations as shown above.}
  \label{fig:CFSM:estimate:hY}
\end{figure}
\begin{figure}[h!]
  \centering
  
    \begin{tikzpicture}[scale=\scalefactorB,every node/.style={transform shape},
    factor/.style={rectangle, minimum width=1cm, minimum height=1.7cm, draw},
    sfactor/.style={rectangle, minimum size=.4cm, draw},
    darksolid/.style={rectangle, minimum size=.15cm, draw, fill = black,
    inner sep=0pt, outer sep = 0pt},
    label/.style={magenta, anchor=north east, xshift = .1cm}]
\node[sfactor,inner sep=0pt] (S) {$\rho_{\system{S}_0}$};
\node[factor] (E1) [right=.7cm of S] {$W$};
\draw (S.north) |- ([yshift=.6cm]E1.west) node[above=-0.05cm,pos=.75] {$s_0$};
\draw (S.south) |- ([yshift=-.6cm]E1.west) node[below,pos=.75] {$\tilde{s}_0$};
\node[sfactor] (X1) [above=.7cm of E1] {$Q$};
\draw (X1) -- (E1) node[right, midway] {$x_1$};
\draw (E1.south) -- ([yshift=-.7cm]E1.south) node[darksolid]{}
    node (Y1) [right] {$\cy_1$};
    
\node[factor] (E2) [right=1cm of E1] {$W$};
\draw ([yshift=.6cm]E1.east) |- ([yshift=.6cm]E2.west)
    node[above=-0.05cm,pos=.75] {$s_1$};
\draw ([yshift=-.6cm]E1.east) |- ([yshift=-.6cm]E2.west)
    node[below,pos=.75] {$\tilde{s}_1$};
\node[sfactor] (X2) [above=.7cm of E2] {$Q$};
\draw (X2) -- (E2) node[right, midway] {$x_2$};
\draw (E2.south) -- ([yshift=-.7cm]E2.south) node[darksolid]{}
    node[right] {$\cy_2$};
    
\node[factor, draw=none] (Edummy1) [right=1cm of E2] {};
\node at ([yshift=.6cm]E2.east-|Edummy1) {$\cdots$};
\node at ([yshift=-.6cm]E2.east-|Edummy1) {$\cdots$};
\draw ([yshift=.6cm]E2.east) |- ([yshift=.6cm]Edummy1.west) node[above=-0.05cm,pos=.75] {$s_2$};
\draw ([yshift=-.6cm]E2.east) |- ([yshift=-.6cm]Edummy1.west)
    node[below,pos=.75] {$\tilde{s}_2$};
\node at (X1-|Edummy1) {$\cdots$};
\node at (Y1-|Edummy1) {$\cdots$};
    
\node[factor] (El) [right=1cm of Edummy1] {$W$};
\draw ([yshift=.6cm]Edummy1.east) |- ([yshift=.6cm]El.west)
    node[above=-0.05cm,pos=.75] {$s_{\ell-1}$};
\draw ([yshift=-.6cm]Edummy1.east) |- ([yshift=-.6cm]El.west)
    node[below,pos=.75] {$\tilde{s}_{\ell-1}$};
\node[sfactor] (Xl) [above=.7cm of El] {$Q$};
\draw (Xl) -- (El) node[right, midway] {$x_\ell$};
\draw (El.south) -- ([yshift=-.7cm]El.south) node[darksolid]{}
    node[right] {$\cy_\ell$};
    
    
\node[factor, draw=none] (Edummy2) [right=1cm of El] {};
\node at ([yshift=.6cm]El.east-|Edummy2) {$\cdots$};
\node at ([yshift=-.6cm]El.east-|Edummy2) {$\cdots$};
\draw ([yshift=.6cm]El.east) |- ([yshift=.6cm]Edummy2.west)
    node[above=-0.05cm,pos=.85] {$s_{\ell\!+\!1}$};
\draw ([yshift=-.6cm]El.east) |- ([yshift=-.6cm]Edummy2.west)
    node[below,pos=.85] {$\tilde{s}_{\ell\!+\!1}$};
\node at (X1-|Edummy2) {$\cdots$};
\node at (Y1-|Edummy2) {$\cdots$};
    
\node[factor] (En) [right=1cm of Edummy2] {$W$};
\draw ([yshift=.6cm]Edummy2.east) |- ([yshift=.6cm]En.west)
    node[above=-0.05cm,pos=.75] {$s_{n\!-\!1}$};
\draw ([yshift=-.6cm]Edummy2.east) |- ([yshift=-.6cm]En.west)
    node[below,pos=.75] {$\tilde{s}_{n\!-\!1}$};
\node[sfactor] (Xn) [above=.7cm of En] {$Q$};
\draw (Xn) -- (En) node[right, midway] {$x_n$};
\draw (En.south) -- ([yshift=-.7cm]En.south) node[darksolid]{}
    node[right] {$\cy_n$};
    
\node[sfactor, inner sep=0pt] (ee) [right=1cm of En] {$=$};
\draw ([yshift=.6cm]En.east) -| (ee.north) node[above=-0.05cm,pos=.25] {$s_n$};
\draw ([yshift=-.6cm]En.east) -| (ee.south) node[below,pos=.25] {$\tilde{s}_n$};
    
\begin{pgfonlayer}{bg}
    \draw[dashed, magenta, line width=1.5pt, fill=magenta!15]
        ([xshift=-1.2cm,yshift=2.9cm]S) rectangle
        ([xshift=.7cm,yshift=-3.6cm]En);
    \node[label] at ([xshift=.7cm,yshift=-3.6cm]En) {$\sigmaY_n$};
    \draw[dashed, magenta, line width=1.5pt, fill=magenta!30]
        ([xshift=-1cm,yshift=2.7cm]S) rectangle
        ([xshift=.7cm,yshift=-3cm]El);
    \node[label] at ([xshift=.7cm,yshift=-3cm]El) {$\sigmaY_{\ell}$};
    \draw[dashed, magenta, line width=1.5pt, fill=magenta!45]
        ([xshift=-.8cm,yshift=2.5cm]S) rectangle
        ([xshift=.7cm,yshift=-2.4cm]E2);
    \node[label] at ([xshift=.7cm,yshift=-2.4cm]E2) {$\sigmaY_2$};
    \draw[dashed, magenta, line width=1.5pt, fill=magenta!60]
        ([xshift=-.6cm,yshift=2.3cm]S) rectangle
        ([xshift=.7cm,yshift=-1.8cm]E1);
    \node[label] at ([xshift=.7cm,yshift=-1.8cm]E1) {$\sigmaY_1$};
\end{pgfonlayer}
\end{tikzpicture}
  \caption{The iterative computation of $\sigmaY_{\ell}$ as 
    in~\eqref{eq:recursive:quantum:state:metric:Y:1} can be understood as a
    sequence of CTB operations as shown above.}
  \label{fig:QFSM:estimate:hY}
\end{figure}
\begin{figure}[h!]
  \centering
  
    \begin{tikzpicture}[scale=\scalefactorB,every node/.style={transform shape},
    factor/.style={rectangle, minimum width=1cm, minimum height=.7cm, draw},
    sfactor/.style={rectangle, minimum size=.4cm, draw},
    darksolid/.style={rectangle, minimum size=.15cm, draw, fill = black,
    inner sep=0pt, outer sep = 0pt},
    label/.style={magenta, anchor=north east, xshift = .1cm}]
\node[sfactor,inner sep=0pt] (S) {$P_{\rv{S}_0}$};
\node[factor] (E1) [right=.7cm of S] {$W$};
\draw (S) -- (E1) node[above,midway] {$s_0$};
\node[darksolid] (X1) [above=.3cm of E1.north,anchor = south] {};
\draw (X1) -- (E1) node[right, pos = 0] {$\cx_1$};
\draw (X1) -- ([yshift=.3cm]X1.north) 
    node (pX1) [sfactor, anchor = south, pos=1] {$Q$};
\draw (E1.south) -- ([yshift=-.8cm]E1.south) node (Y1) [darksolid]{} 
    node[right] {$\cy_1$};
    
\node[factor] (E2) [right=1cm of E1] {$W$};
\draw (E1) -- (E2) node[above,midway] {$s_1$};
\node[darksolid] (X2) [above=.3cm of E2.north,anchor = south] {};
\draw (X2) -- (E2) node[right, pos = 0] {$\cx_2$};
\draw (X2) -- ([yshift=.3cm]X2.north)
    node[sfactor, anchor = south, pos=1] {$Q$};
\draw (E2.south) -- ([yshift=-.8cm]E2.south) 
    node[darksolid]{} node[right] {$\cy_2$};
    
\node[factor, draw=none] (Edummy1) [right=1cm of E2] {$\cdots$};
\draw (E2) -- (Edummy1) node[above,midway] {$s_2$};
\node at (pX1-|Edummy1) {$\cdots$};
\node at (Y1-|Edummy1) {$\cdots$};
    
\node[factor] (El) [right=1cm of Edummy1] {$W$};
\draw (Edummy1) -- (El) node[above,midway] {$s_{\ell-1}$};
\node[darksolid] (Xl) [above=.3cm of El.north,anchor = south] {};
\draw (Xl) -- (El) node[right, pos = 0] {$\cx_\ell$};
\draw (Xl) -- ([yshift=.3cm]Xl.north)
    node[sfactor, anchor = south, pos=1] {$Q$};
\draw (El.south) -- ([yshift=-.8cm]El.south) node[darksolid]{}
    node[right] {$\cy_\ell$};
    
    
\node[factor, draw=none] (Edummy2) [right=1cm of El] {$\cdots$};
\draw (El) -- (Edummy2) node[above,pos=.7] {$s_{\ell\!+\!1}$};
\node at (pX1-|Edummy2) {$\cdots$};
\node at (Y1-|Edummy2) {$\cdots$};
    
\node[factor] (En) [right=1cm of Edummy2] {$W$};
\draw (Edummy2) -- (En) node[above,midway] {$s_{n\!-\!1}$};
\node[darksolid] (Xn) [above=.3cm of En.north,anchor = south] {};
\draw (Xn) -- (En) node[right, pos = 0] {$\cx_n$};
\draw (Xn) -- ([yshift=.3cm]Xn.north)
    node[sfactor, anchor = south, pos=1] {$Q$};
\draw (En.south) -- ([yshift=-.8cm]En.south) node[darksolid]{}
    node[right] {$\cy_{n}$};
    
\node[sfactor, inner sep=0pt] (ee) [right=1cm of En] {$\mathbf{1}$};
\draw (En) -- (ee) node[above,midway] {$s_n$};
    
\begin{pgfonlayer}{bg}
    \draw[dashed, magenta, line width=1.5pt, fill=magenta!15]
        ([xshift=-1.2cm,yshift=2.5cm]S) rectangle
        ([xshift=.7cm,yshift=-3.2cm]En);
    \node[label] at ([xshift=.7cm,yshift=-3.2cm]En) {$\muXY_{n}$};
    \draw[dashed, magenta, line width=1.5pt, fill=magenta!30]
        ([xshift=-1cm,yshift=2.3cm]S) rectangle
        ([xshift=.7cm,yshift=-2.6cm]El);
    \node[label] at ([xshift=.7cm,yshift=-2.6cm]El) {$\muXY_{\ell}$};
    \draw[dashed, magenta, line width=1.5pt, fill=magenta!45]
        ([xshift=-.8cm,yshift=2.1cm]S) rectangle
        ([xshift=.7cm,yshift=-2cm]E2);
    \node[label] at ([xshift=.7cm,yshift=-2cm]E2) {$\muXY_2$};
    \draw[dashed, magenta, line width=1.5pt, fill=magenta!60]
        ([xshift=-.6cm,yshift=1.9cm]S) rectangle
        ([xshift=.7cm,yshift=-1.4cm]E1);
    \node[label] at ([xshift=.7cm,yshift=-1.4cm]E1) {$\muXY_1$};
\end{pgfonlayer}
\end{tikzpicture}
  \caption{The iterative computation of $\muXY_{\ell}$ can be understood as a 
    sequence of CTB operations as shown above.}
  \label{fig:CFSM:estimate:hXY}
\end{figure}
\begin{figure}[h!]
  \centering
  
    \begin{tikzpicture}[scale=\scalefactorB,every node/.style={transform shape},
    factor/.style={rectangle, minimum width=1cm, minimum height=1.7cm, draw},
    sfactor/.style={rectangle, minimum size=.4cm, draw},
    darksolid/.style={rectangle, minimum size=.15cm, draw, fill = black,
    inner sep=0pt, outer sep = 0pt},
    label/.style={magenta, anchor=north east, xshift = .1cm}]
\node[sfactor,inner sep=0pt] (S) {$\rho_{\system{S}_0}$};
\node[factor] (E1) [right=.7cm of S] {$W$};
\draw (S.north) |- ([yshift=.6cm]E1.west) node[above=-0.05cm,pos=.75] {$s_0$};
\draw (S.south) |- ([yshift=-.6cm]E1.west) node[below,pos=.75] {$\tilde{s}_0$};
\node[darksolid] (X1) [above=.3cm of E1.north,anchor = south] {};
\draw (X1) -- (E1) node[right, pos = 0] {$\cx_1$};
\draw (X1) -- ([yshift=.3cm]X1.north)
    node (pX1) [sfactor, anchor = south, pos=1] {$Q$};
\draw (E1.south) -- ([yshift=-.7cm]E1.south) node[darksolid]{}
    node (Y1) [right] {$\cy_1$};
    
\node[factor] (E2) [right=1cm of E1] {$W$};
\draw ([yshift=.6cm]E1.east) |- ([yshift=.6cm]E2.west)
    node[above=-0.05cm,pos=.75] {$s_1$};
\draw ([yshift=-.6cm]E1.east) |- ([yshift=-.6cm]E2.west)
    node[below,pos=.75] {$\tilde{s}_1$};
\node[darksolid] (X2) [above=.3cm of E2.north,anchor = south] {};
\draw (X2) -- (E2) node[right, pos = 0] {$\cx_2$};
\draw (X2) -- ([yshift=.3cm]X2.north) node[sfactor, anchor = south, pos=1] {$Q$};
\draw (E2.south) -- ([yshift=-.7cm]E2.south) node[darksolid]{} 
    node[right] {$\cy_2$};
    
\node[factor, draw=none] (Edummy1) [right=1cm of E2] {};
\node at ([yshift=.6cm]E2.east-|Edummy1) {$\cdots$};
\node at ([yshift=-.6cm]E2.east-|Edummy1) {$\cdots$};
\draw ([yshift=.6cm]E2.east) |- ([yshift=.6cm]Edummy1.west)
    node[above=-0.05cm,pos=.75] {$s_2$};
\draw ([yshift=-.6cm]E2.east) |- ([yshift=-.6cm]Edummy1.west)
    node[below,pos=.75] {$\tilde{s}_2$};
\node at (pX1-|Edummy1) {$\cdots$};
\node at (Y1-|Edummy1) {$\cdots$};
    
\node[factor] (El) [right=1cm of Edummy1] {$W$};
\draw ([yshift=.6cm]Edummy1.east) |- ([yshift=.6cm]El.west)
    node[above=-0.05cm,pos=.75] {$s_{\ell-1}$};
\draw ([yshift=-.6cm]Edummy1.east) |- ([yshift=-.6cm]El.west)
    node[below,pos=.75] {$\tilde{s}_{\ell-1}$};
\node[darksolid] (Xl) [above=.3cm of El.north,anchor = south] {};
\draw (Xl) -- (El) node[right, pos = 0] {$\cx_\ell$};
\draw (Xl) -- ([yshift=.3cm]Xl.north) node[sfactor, anchor = south, pos=1] {$Q$};
\draw (El.south) -- ([yshift=-.7cm]El.south) node[darksolid]{}
    node[right] {$\cy_\ell$};
    
    
\node[factor, draw=none] (Edummy2) [right=1cm of El] {};
\node at ([yshift=.6cm]El.east-|Edummy2) {$\cdots$};
\node at ([yshift=-.6cm]El.east-|Edummy2) {$\cdots$};
\draw ([yshift=.6cm]El.east) |- ([yshift=.6cm]Edummy2.west)
    node[above=-0.05cm,pos=.85] {$s_{\ell\!+\!1}$};
\draw ([yshift=-.6cm]El.east) |- ([yshift=-.6cm]Edummy2.west)
    node[below,pos=.85] {$\tilde{s}_{\!\ell+\!1}$};
\node at (pX1-|Edummy2) {$\cdots$};
\node at (Y1-|Edummy2) {$\cdots$};

\node[factor] (En) [right=1cm of Edummy2] {$W$};
\draw ([yshift=.6cm]Edummy2.east) |- ([yshift=.6cm]En.west)
    node[above=-0.05cm,pos=.75] {$\tilde{s}_{n\!-\!1}$};
\draw ([yshift=-.6cm]Edummy2.east) |- ([yshift=-.6cm]En.west)
    node[below,pos=.75] {$\tilde{s}_{n\!-\!1}$};
\node[darksolid] (Xn) [above=.3cm of En.north,anchor = south] {};
\draw (Xn) -- (En) node[right, pos = 0] {$\cx_n$};
\draw (Xn) -- ([yshift=.3cm]Xn.north)
    node[sfactor, anchor = south, pos=1] {$Q$};
\draw (En.south) -- ([yshift=-.7cm]En.south) node[darksolid]{}
    node[right] {$\cy_n$};
    
\node[sfactor, inner sep=0pt] (ee) [right=1cm of En] {$=$};
\draw ([yshift=.6cm]En.east) -| (ee.north) node[above=-0.05cm,pos=.25] {$s_n$};
\draw ([yshift=-.6cm]En.east) -| (ee.south) node[below,pos=.25] {$\tilde{s}_n$};
    
\begin{pgfonlayer}{bg}
    \draw[dashed, magenta, line width=1.5pt, fill=magenta!15]
        ([xshift=-1.2cm,yshift=3cm]S) rectangle
        ([xshift=.7cm,yshift=-3.6cm]En);
    \node[label] at ([xshift=.7cm,yshift=-3.6cm]En) {$\sigmaXY_n$};
    \draw[dashed, magenta, line width=1.5pt, fill=magenta!30]
        ([xshift=-1cm,yshift=2.8cm]S) rectangle
        ([xshift=.7cm,yshift=-3cm]El);
    \node[label] at ([xshift=.7cm,yshift=-3cm]El) {$\sigmaXY_{\ell}$};
    \draw[dashed, magenta, line width=1.5pt, fill=magenta!45]
        ([xshift=-.8cm,yshift=2.6cm]S) rectangle
        ([xshift=.7cm,yshift=-2.4cm]E2);
    \node[label] at ([xshift=.7cm,yshift=-2.4cm]E2) {$\sigmaXY_2$};
    \draw[dashed, magenta, line width=1.5pt, fill=magenta!60]
        ([xshift=-.6cm,yshift=2.4cm]S) rectangle
        ([xshift=.7cm,yshift=-1.8cm]E1);
    \node[label] at ([xshift=.7cm,yshift=-1.8cm]E1) {$\sigmaXY_1$};
\end{pgfonlayer}
\end{tikzpicture}
  \caption{The iterative computation of $\sigmaXY_{\ell}$ as 
    in~\eqref{eq:def:quantum:channel:state:metric:XY:1} can be understood as a
    sequence of CTB operations as shown above.}
  \label{fig:QFSM:estimate:hXY}
\end{figure}
\end{document}